\documentclass[acmsmall,nonacm,screen]{acmart}

\usepackage{xcolor}
\usepackage{xspace}
\usepackage{enumitem}
\usepackage{graphicx}
\usepackage{cleveref}
\usepackage{algorithm}
\usepackage[noend]{algpseudocode}
\usepackage[font=footnotesize,labelfont=bf]{caption} 

\setcitestyle{numbers,sort&compress}

\theoremstyle{plain}
\newtheorem{theorem}{Theorem}[section]
\newtheorem{definition}[theorem]{Definition}

\newtheorem{claim}[theorem]{Claim}
\newtheorem{lemma}[theorem]{Lemma}

\newtheorem*{claim*}{Claim}

\newtheorem{remark}[theorem]{Remark}

\newcommand{\ftp}[1]{\left\lfloor#1\right\rfloor}

\acmBooktitle{}

\begin{document}

\title{Locally-iterative $(\Delta+1)$-Coloring in Sublinear (in $\Delta$) Rounds}

\author{Xinyu Fu}
\email{xyfu@smail.nju.edu.cn}

\author{Yitong Yin}
\email{yinyt@nju.edu.cn}

\author{Chaodong Zheng}
\email{chaodong@nju.edu.cn}

\affiliation{
	\department{State Key Laboratory for Novel Software Technology}
	\institution{Nanjing University}
	\country{China}
}

\begin{abstract}
Distributed graph coloring is one of the most extensively studied problems in distributed computing. There is a canonical family of distributed graph coloring algorithms known as the \emph{locally-iterative} coloring algorithms, first formalized in [Szegedy and Vishwanathan, STOC~'93]. In such algorithms, every vertex iteratively updates its own color according to a predetermined function of the current coloring of its local neighborhood. Due to the simplicity and naturalness of its framework, locally-iterative coloring algorithms are of great significance both in theory and practice.

In this paper, we give a locally-iterative $(\Delta+1)$-coloring algorithm with runtime $O(\Delta^{3/4}\log\Delta)+\log^*{n}$, using messages of size $O(\log{n})$. This is the first locally-iterative $(\Delta+1)$-coloring algorithm with sublinear-in-$\Delta$ runtime, and answers the main open question raised by previous best result [Barenboim, Elkin, and Goldberg, JACM~'21].
The key component of our algorithm is a new locally-iterative procedure that transforms an $O(\Delta^2)$-coloring to a $(\Delta+O(\Delta^{3/4}\log\Delta))$-coloring in $o(\Delta)$ time.
As an application of our result, we also devise a self-stabilizing algorithm for $(\Delta+1)$-coloring with $O(\Delta^{3/4}\log\Delta)+\log^*{n}$ stabilization time, using $O(\log{n})$-bit messages. To the best of our knowledge, this is the first self-stabilizing algorithm for $(\Delta+1)$-coloring in the CONGEST model with sublinear-in-$\Delta$ stabilization time.
\end{abstract}


\maketitle


\section{Introduction}\label{sec:intro}

Distributed graph coloring is one of the most fundamental and extensively studied problems in distributed computing~\cite{luby86,alon86,linial87,cole86,goldberg87,szegedy93,kuhn06,barenboim13,barenboim14,fraigniaud16,barenboim16sublinear,barenboim16locality,harris16,chang18,maus20,barenboim21,ghaffari22}.
As a locally checkable labeling problem, distributed graph coloring is widely considered to be one of the benchmark problems for answering the fundamental question ``what can be computed locally''~\cite{naor93}. The problem also has a wide range of applications, including channel allocation, scheduling, etc.~\cite{guellati10,kuhn09}.

Given a graph $G=(V,E)$ and a palette $Q$ of $q=|Q|$ colors, a \emph{$q$-coloring} is a mapping $\phi:V\to Q$. 
A $q$-coloring is \emph{proper} if $\phi(u)\neq \phi(v)$ for every edge $(u,v)\in E$. Distributed graph coloring is often studied in the synchronous message-passing model~\cite{peleg00}. 
Here, a communication network is represented by an $n$-vertex graph $G=(V,E)$ with maximum degree $\Delta$. Each vertex $v\in V$ hosts a processor and each edge $(u,v)\in E$ denotes a communication link between $u$ and $v$. Each $v\in V$ has a unique identifier $id(v)$ from the set $[n]=\{0,1,\cdots,n-1\}$. In each synchronous \emph{round}, vertices perform local computation and exchange messages with their neighbors. The time complexity of an algorithm is the maximum number of rounds required for all vertices to terminate.

There is a natural family of distributed graph coloring algorithms known as the \emph{locally-iterative coloring algorithms}, introduced by Szegedy and Vishwanathan~\cite{szegedy93}. Throughout the execution of such algorithms, a proper coloring of the network graph is maintained and updated from round to round. Moreover, in each round, for each vertex $v\in V$, its next color is computed from its current color and the current colors of its neighbors $N(v)=\{u\in V\mid(u,v)\in E)\}$.

\begin{definition}
[\textbf{Locally-iterative Coloring Algorithms}]
\label{def:locally-iter-alg}
In the synchronous message-passing model, an algorithm for graph coloring is said to be \emph{locally-iterative} if it maintains a sequence of proper colorings $\phi_t$ of the input $G=(V,E)$ such that:
\begin{itemize}
	\item The initial coloring $\phi_0$ is constructed locally in the sense that, for every vertex $v$, its initial color $\phi_0(v)$ is computed locally from $id(v)$.
	\item In each round $t\ge 1$, every vertex $v$ computes its next color $\phi_{t}(v)$ based only on its current color $\phi_{t-1}(v)$ along with the multiset of colors $\{\phi_{t-1}(u)\mid u\in N(v)\}$ appearing in $v$'s neighborhood.
	Particularly, in each round $t\geq 1$, every vertex $v$ only broadcasts $\phi_{t-1}(v)$ to its neighbors.
\end{itemize}
\end{definition}

\begin{remark}
[Uniformity]
\label{remark:uniform-alg}
By our definition, in a locally-iterative coloring algorithm $\mathcal{A}$, in every round $t\geq 1$, the colors of vertices are updated according to a \emph{uniform} rule which is oblivious to the current round number and the identity of the vertex running it. Formally,
\begin{displaymath}
\phi_{t}(v)\gets\textsc{Update}_\mathcal{A}\left(\phi_{t-1}(v),\{\phi_{t-1}(u)\mid u\in N(v)\}\right).
\end{displaymath}
\end{remark}

Due to the simplicity and naturalness of its framework, locally-iterative algorithms are of great significance both in theory and practice. Indeed, in computer science and even physics, many classical algorithms with a wide range of applications are locally-iterative in nature, such as distance-vector routing~\cite{bertsekas92} and belief propagation~\cite{pearl88}. In this paper, we seek fast locally-iterative algorithms that can compute a proper $(\Delta+1)$-coloring.

A heuristic $\Omega(\Delta\log\Delta+\log^* n)$ lower bound for locally-iterative $(\Delta+1)$-coloring algorithms was proposed by Szegedy and Vishwanathan in \cite{szegedy93}, and it holds unless there exists ``a very special type of coloring that can be very efficiently reduced''. For a long time, this bound matched the fastest algorithm~\cite{kuhn06}. Nevertheless, such ``special'' type of coloring was found in a recent breakthrough: Barenboim, Elkin, and Goldberg~\cite{barenboim21} devised a locally-iterative $(\Delta+1)$-coloring algorithm with $O(\Delta)+\log^*n$ runtime, breaking the long-standing barrier.

On the other hand, the landscape for general distributed graph coloring algorithms is somewhat different. According to \cite{barenboim21}, all $(\Delta+1)$-coloring algorithms developed before 2009 are locally-iterative, including retrospectively, those before Szegedy and Vishwanathan's work~\cite{goldberg87,linial87}. After 2009, a series of algorithms that are not locally-iterative were developed, achieving linear-in-$\Delta$~\cite{barenboim09linear,kuhn09,barenboim14} or even sublinear-in-$\Delta$~\cite{barenboim16sublinear,fraigniaud16,maus20} runtime.
Notice that by encoding these general algorithms' complete internal states as ``colors'', it is possible to simulate some of them in a locally-iterative manner. However, this approach has several limitations and drawbacks (we will elaborate more on this later), and a major one being large message size due to large state space.

Therefore, an important question---which is also the main open problem raised in \cite{barenboim21}---is, can one compute a proper $(\Delta+1)$-coloring with a locally-iterative algorithm in $o(\Delta)+\log^*n$ time, using only small messages?

\subsection{Our results}

We answer the above question affirmatively, which is formally stated in the following theorem.

\begin{theorem}
[\textbf{Efficient Locally-iterative Coloring Algorithm}]
\label{thm:alg-main}
There exists a locally-iterative coloring algorithm such that, for any input graph with $n$ vertices and maximum degree $\Delta$, produces a proper $(\Delta+1)$-coloring within $O(\Delta^{3/4}\log{\Delta})+\log^*{n}$ rounds, using messages of $O(\log{n})$ bits.
\end{theorem}


\noindent Our algorithm is the first locally-iterative $(\Delta+1)$-coloring algorithm achieving a runtime with sublinear dependency on the maximum degree $\Delta$. The core of this algorithm is a locally-iterative procedure that transforms a proper $O(\Delta^2)$-coloring to a proper $(\Delta+O(\Delta^{3/4}\log\Delta))$-coloring within $o(\Delta)$ rounds.
Inside this procedure we work on special proper colorings that correspond to (arb)defective colorings, and reduce the number of used colors quadratically in a locally-iterative fashion.
Combine this procedure with Linial's well-known color-reduction procedure~\cite{linial87} and the folklore reduce-one-color-per-round procedure, we obtain the complete algorithm.

\paragraph{An application in self-stabilizing coloring}

Fault-tolerance is another central topic in distributed computing. \emph{Self-stabilization}, a concept coined by Edsger W.\ Dijkstra~\cite{dijkstra74}, is a property that, roughly speaking, guarantees a distributed system starting from an arbitrary state eventually converges to a desired behavior.
This concept is regarded as ``a milestone in work on fault tolerance'' by Leslie Lamport~\cite{lamport85}.
Indeed, over the last four decades, lots of self-stabilizing distributed algorithms have been devised~\cite{dolev00,altisen19}, and several of them have seen practical applications~\cite{datta94,chen05}.

In this paper, we adopt the same self-stabilizing setting assumed by Barenboim, Elkin, and Goldenberg~\cite{barenboim21}. In this setting, the memory of each vertex consists of two parts: the immutable \emph{read-only memory (ROM)} and the mutable \emph{random access memory (RAM)}. The ROM part is faultless but cannot change during execution; and it may be used to store hard-wired data such as vertex identity and graph parameters, as well as the program code. The RAM part on the other hand, can change during algorithm execution; and it is for storing the internal states of the algorithm. The RAM part is subject to error, controlled by an adversary called \emph{Eve}. In particular, at any moment during the execution, the adversary can examine the entire memory (including both ROM and RAM) of all vertices, and then make arbitrary changes to the RAM part of all vertices.

An algorithm is \emph{self-stabilizing} if it can still compute a proper solution once the adversary stops disrupting its execution. Specifically, assuming that $T_0$ is the last round in which the adversary makes any changes to vertices' RAM areas, if it is always guaranteed that by the end of round $T_0+T$ a desired solution (e.g., a proper $(\Delta+1)$-coloring in our context) is produced by the algorithm, then the algorithm is self-stabilizing with \emph{stabilization time} $T$.

With suitable adjustments, we are able to transform our locally-iterative coloring algorithm into a self-stabilizing one, whose guarantees are formally stated in the following theorem.
\begin{theorem}
[\textbf{Efficient Self-stabilizing Coloring Algorithm}]
\label{thm:alg-self-stab}
There exists a self-stabilizing coloring algorithm such that, for any input graph with $n$ vertices and maximum degree $\Delta$, produces a proper $(\Delta+1)$-coloring with $O(\Delta^{3/4}\log{\Delta})+\log^*{n}$ stabilization time, using messages of $O(\log{n})$ bits.
\end{theorem}%
\noindent To the best of our knowledge, in the CONGEST model, this is the first self-stabilizing algorithm for $(\Delta+1)$-coloring with sublinear-in-$\Delta$ stabilization time.

\paragraph{Reconfigurable locally-iterative coloring and self-stabilization}

In adopting the locally-iterative algorithm to the self-stabilizing setting, we cope with a strong level of ``asynchrony'' among vertices, as the adversary can manipulate vertices' internal states. We have also crafted an error-correction procedure ensuring that once the adversary stops disrupting algorithm execution, any vertex with an ``improper'' state will be detected instantly and resets itself to some proper state.

Interestingly, we find that if a locally-iterative algorithm supports above ``reconfiguration'' (that is, state resetting upon detecting illicit status), and if the algorithm's correctness is still enforced under such reconfiguration, then the locally-iterative algorithm in consideration can be modified into a self-stabilizing algorithm with relative ease, and limited or no complexity overhead will be imposed. Formally, we propose the notion of \emph{reconfigurable locally-iterative coloring algorithms}. 

\begin{definition}
[\textbf{Reconfigurable Locally-iterative Coloring Algorithms}]
\label{def:strong-locally-iter-alg}
A locally-iterative coloring algorithm $\mathcal{A}$ is \emph{reconfigurable} if it satisfies the following properties:
\begin{itemize}
	\item The algorithm has a build-in function $\textsc{Legit}_{\mathcal{A}}$ that, upon inputting a vertex $v$'s color and the colors of $v$'s neighbors, outputs a single bit indicating whether $v$'s status is legit or illicit.
	\item For each vertex $v$, its status after initialization is legit:
	$\textsc{Legit}_{\mathcal{A}}\left(\phi_0(v),\{\phi_0(u)\mid u\in N(v)\}\right)=1.$
	\item Normal execution maintains legitimacy. Moreover, when external interference occur, resetting illicit vertices to initial colors resumes legitimacy. Specifically, for any round $t\geq 1$, let $V_{\phi_{t-1}}$ denote the set of vertices that have illicit status in (a not necessarily proper) coloring $\phi_{t-1}$, if every vertex updates its color using the following rule, then in $\phi_{t}$ all vertices' status are legit.
	\begin{equation*}
		\phi_{t}(v)\gets\left\{
			\begin{aligned}
			& \phi_0(v) & \textrm{if } & v\in V_{\phi_{t-1}}, \\
			& \textsc{Update}_\mathcal{A}\left(\phi_{t-1}(v),\{\phi_{t-1}(u)\mid u\in N(v)\}\right) & \textrm{if } & v\in V\setminus V_{\phi_{t-1}}.
			\end{aligned}
		\right.
	\end{equation*}
\end{itemize}
\end{definition}

The locally-iterative coloring algorithm proposed by Barenboim, Elkin, and Goldenberg~\cite{barenboim21} supports reconfiguration. Indeed, they are able to derive a self-stabilizing algorithm that is also locally-iterative with stabilization time $O(\Delta+\log^*{n})$. On the other hand, for our locally-iterative algorithm to handle the strong asynchrony brought by reconfiguration, in designing the self-stabilizing algorithm, we have to allow vertices to send different messages to different neighbors, violating the ``broadcast'' requirement of locally-iterative algorithms.
Nonetheless, one could also imagine an ``edge orientation'' version of locally-iterative algorithms, in which each vertex $v$ maintains a state $\sigma^{(v\to u)}$ for each incident edge $(v,u)$. By exchanging $\sigma^{(v\to u)}$ with neighbor $u$ for each edge $(v,u)$ in every round, vertex $v$ can update $\sigma^{(v\to u)}$ and its color $\phi(v)$. That is:
\begin{align*}
\sigma^{(v\to u)}_t & \gets \textsc{Update}_\mathcal{A} \left( \sigma^{(v\to u)}_{t-1}, \sigma^{(u\to v)}_{t-1}, \left\{\left(\sigma^{(v\to w)}_{t-1},\sigma^{(w\to v)}_{t-1}\right)\mid w\in N(v)\setminus\{u\}\right\} \right), \\
\phi_t(v) & \gets \textsc{ComputeColor}_{\mathcal{A}} \left(\left\{\sigma^{(v\to w)}_{t}\mid w\in N(v)\right\}\right).
\end{align*}
Such ``edge orientation'' locally-iterative algorithms are common in physics (e.g., belief propagation).
Meanwhile, under this alternative ``edge orientation'' interpretation, our coloring algorithm supports reconfiguration, allowing it to be modified to a self-stabilizing algorithm with relative ease.



\subsection{Related Work}\label{sec-app:related-work}

The study of distributed graph coloring dates back to the early days of distributed computing. Cole and Vishkin~\cite{cole86} initiated the study of distributed graph coloring on basic graphs such as rings and paths, and developed a deterministic $3$-coloring algorithm with running time $\log^*n + O(1)$. Goldberg and Plotkin~\cite{goldberg87parallel} devised the first $(\Delta+1)$-coloring algorithm for general graphs with $2^{O(\Delta)}+O(\log^* n)$ running time. Linial~\cite{linial87} devised an algorithm that computes an $O(\Delta^2)$-coloring using $\log^*n+O(1)$ time, implying an $O(\Delta^2+\log^*n)$ time $(\Delta+1)$-coloring algorithm. Szegedy and Vishwanathan~\cite{szegedy93} introduced the notion of locally-iterative graph coloring and derived a randomized algorithm along with a heuristic lower bound, this latter bound is attained by Kuhn and Wattenhofer's algorithm~\cite{kuhn06}. All works mentioned above are locally-iterative. Before this paper, the fastest locally-iterative $(\Delta+1)$-coloring algorithm is from Barenboim, Elkin, and Goldenberg~\cite{barenboim21}, which has an $O(\Delta)+\log^*n$ running time.

If we loose the restriction on being locally-iterative, $(\Delta+1)$-coloring algorithms with linear-in-$\Delta$ runtime were first proposed in \cite{barenboim09linear,kuhn09,barenboim14}, then faster algorithms with sublinear-in-$\Delta$ runtime were also discovered~\cite{barenboim16sublinear,fraigniaud16,barenboim21,maus20}. The current best upper bound for $(\Delta+1)$-coloring focusing on $\Delta$-dependency is devised by Maus and Tonoyan~\cite{maus20}, achieving a running time of $O(\sqrt{\Delta\log\Delta}+\log^*n)$. For randomized algorithms, in 1986, the seminal work of Luby~\cite{luby86} and Alon, Noga, Babai~\cite{alon86} showed that distributed $(\Delta+1)$-coloring can be solved within $O(\log{n})$ rounds. Barenboim, Elkin, Pettie and Schneider~\cite{barenboim16locality} improved this bound to $O(\log\Delta)+2^{O(\sqrt{\log\log n})}$. Some improvements have been obtained on this upper bound while maintaining the term $2^{O(\sqrt{\log\log n})}$~\cite{harris16, chang18}. This term is improved to $\text{poly}(\log\log n)$ with the use of better network decomposition techniques~\cite{rozhovn20}. More recently, the upper bound is improved to $O(\log^3\log n)$ in both the CONGEST and the LOCAL model by Ghaffari and Kuhn~\cite{ghaffari22}, and by Halld\'{o}rsson, Nolin, Tonoyan~\cite{halldorsson21}. There are also deterministic algorithms focusing on $n$-dependency. Rozho\v{n} and Ghaffari~\cite{rozhovn20} derived the first $\text{poly}(\log{n})$ rounds algorithm with runtime $O(\log^7 n)$ using network decomposition. It is reduced to $O(\log^5 n)$ with improvements on network decomposition~\cite{ghaffari21}. Very recently, this bound is improved by Ghaffari and Kuhn~\cite{ghaffari22} to $O(\log^3 n)$ rounds, without using network decomposition.

Distributed graph coloring is also extensively studied in the context of self-stabilization~\cite{dolev00,altisen19}. There are algorithms devised for coloring bipartite graphs~\cite{sur93,kosowski06}, planar graphs~\cite{ghosh93,huang05}, and general graphs~\cite{goddard04,hedetniemi03,gradinariu00}. See \cite{guellati10} for a survey on results obtained before 2010.
In \cite{barenboim21}, Barenboim, Elkin and Goldenberg devised the first sublinear-in-$n$ self-stabilizing $(\Delta+1)$-coloring algorithm that works in the CONGEST model, achieving $O(\Delta+\log^*n)$ stabilization time. We improve this bound to sublinear-in-$\Delta$ in this paper.

\section{Preliminaries}\label{sec:model-and-preliminary}

\paragraph{Graph coloring}

Let $G=(V,E)$ be an undirected graph.
Let $q>0$ be a positive integer and $Q$ be a \emph{palette} of $q=|Q|$ colors. A \emph{$q$-coloring} $\phi: V\to Q$ of graph $G$ assigns each vertex $v\in V$ one of the $q$ colors from $Q$, and is said to be:
\begin{itemize}
	\item \emph{proper} if  $\phi(u)\neq\phi(v)$ for every edge $(u,v)\in E$;
	\item \emph{$d$-defective} if for every $v$, the number of neighbors $u\in N(v)$ with $\phi(u)=\phi(v)$ is at most $d$;
	\item \emph{$a$-arbdefective} if we can define an orientation for each edge such that the out-degree of the oriented graph induced by each color class is at most $a$.
\end{itemize}

\paragraph{Cover-free set systems}

The existence of \emph{$\Delta$-cover-free set systems} is crucial for Linial's celebrated coloring algorithm~\cite{linial87}. We use such set systems in our algorithms extensively as well.

\begin{definition}\label{def:cover-free-set-system}
Let $U\neq\emptyset$ be a finite set and $\Delta>0$ be an integer. Set system $\mathcal{F}\subseteq 2^{U}$ with ground set $U$ is \emph{$\Delta$-cover-free} if for every $S_0\in\mathcal{F}$ and every $\Delta$ sets $S_1,\cdots,S_{\Delta}\in\mathcal{F}\setminus\{S_0\}$, it holds that $S_0\nsubseteq\bigcup_{i=1}^{\Delta}S_i$.
\end{definition}

\begin{theorem}[Erd\H{o}s, Frankl, F\"{u}redi \cite{erdos85}]\label{thm:set-families}
For any integers $n>\Delta>0$, there exists $m\in\mathbb{N}^+$ satisfying
\begin{displaymath}
	m\le
	\begin{cases}
		4(\Delta+1)^2\log^2{n} & \text{if }n> 8(\Delta+1)^3,\\
		4(\Delta+1)^2 & \text{if }n\le 8(\Delta+1)^3,
	\end{cases}
\end{displaymath}
such that for every finite set $U$ of size $|U|\ge m$, there exists a $\Delta$-cover-free set system $\mathcal{F}\subseteq 2^{U}$ of size $|\mathcal{F}|=n$ with ground set $U$.
\end{theorem}

Barenboim, Elkin, and Kuhn~\cite{barenboim14} generalized $\Delta$-cover-free set systems to a notion of \emph{$\Delta$-union-$(\rho+1)$-cover-free set systems}, and proved their existence for reasonably small parameters. Our algorithms utilize such generalized cover-free set systems as well.

\begin{definition}\label{def:cover-free-set-system-extension}
Let $U\neq\emptyset$ be a finite set and $\Delta,\rho$ be two positive integers. Set system $\mathcal{F}\subseteq 2^{U}$ with ground set $U$ is \emph{$\Delta$-union-$(\rho+1)$-cover-free} if for every $S_0\in\mathcal{F}$ and every $\Delta$ sets $S_1,\cdots,S_{\Delta}\in\mathcal{F}\setminus\{S_0\}$, there exists at least one element $x\in S_0$ that appears in at most $\rho$ sets among $S_1,S_2,\cdots,S_{\Delta}$, that is,
$$|\{i\mid x\in S_i, 1\le i\le \Delta \}| \leq \rho.$$
\end{definition}

\begin{theorem}[Theorem 3.9 of \cite{barenboim14}]\label{thm:set-families-extension}
For any integers $n>\Delta>\rho>0$, there exists a $\Delta$-union-$(\rho+1)$-cover-free set family $\mathcal{F}\subseteq 2^{U}$ of size $|\mathcal{F}|=n$ with ground set $U$ satisfying
$$|U|\leq 4\cdot\left(\frac{\Delta+1}{\rho+1}\right)^2\cdot\log^2{n}.$$
\end{theorem}

\section{The Locally-iterative Coloring Algorithm}\label{sec:alg-framework}

\paragraph{A natural but unsuccessful attempt}

Recall there are general $(\Delta+1)$-coloring algorithms with sublinear-in-$\Delta$ runtime, such as \cite{barenboim16sublinear,fraigniaud16,maus20,barenboim21}. Imagine taking a general coloring algorithm and ``encode'' the complete internal state of a vertex as its ``color''. Then, by exchanging internal states with neighbors, it seems one could simulate a general coloring algorithm in a locally-iterative manner. Unfortunately, this approach has several limitations and drawbacks.

The first drawback is large message size. During conversion, all variables that carry over multiple rounds in the general algorithm have to be encoded. For complex algorithms (e.g., \cite{fraigniaud16}), the space required to store these variables---which corresponds to the message size for the converted locally-iterative algorithm---will be large. Another limitation is that the conversion works only if the general algorithm is ``broadcast'' in nature. Specifically, if the general algorithm sends different messages to different neighbors, then in the converted algorithm, for a vertex $v$, knowing neighbor $u$'s internal state is insufficient: though $v$ can compute the set of messages $u$ will send, $v$ does not know which one is targeted for it. Lastly, the conversion may also pose additional requirements. For example, many algorithms use time to synchronize vertices' behavior (e.g., \cite{barenboim16sublinear}), thus in the converted algorithms we must encode round number. This implies we must know an upper bound on the running of the algorithm (so as to allocate proper number of bits to encode round number), which could depend on various parameters. Notice that the general algorithm may be oblivious of these parameters, yet for the conversion to work these parameters have to be known at prior.

\paragraph{Our approach}

\begin{figure}[t!]
\centering
\includegraphics[width=0.82\linewidth]{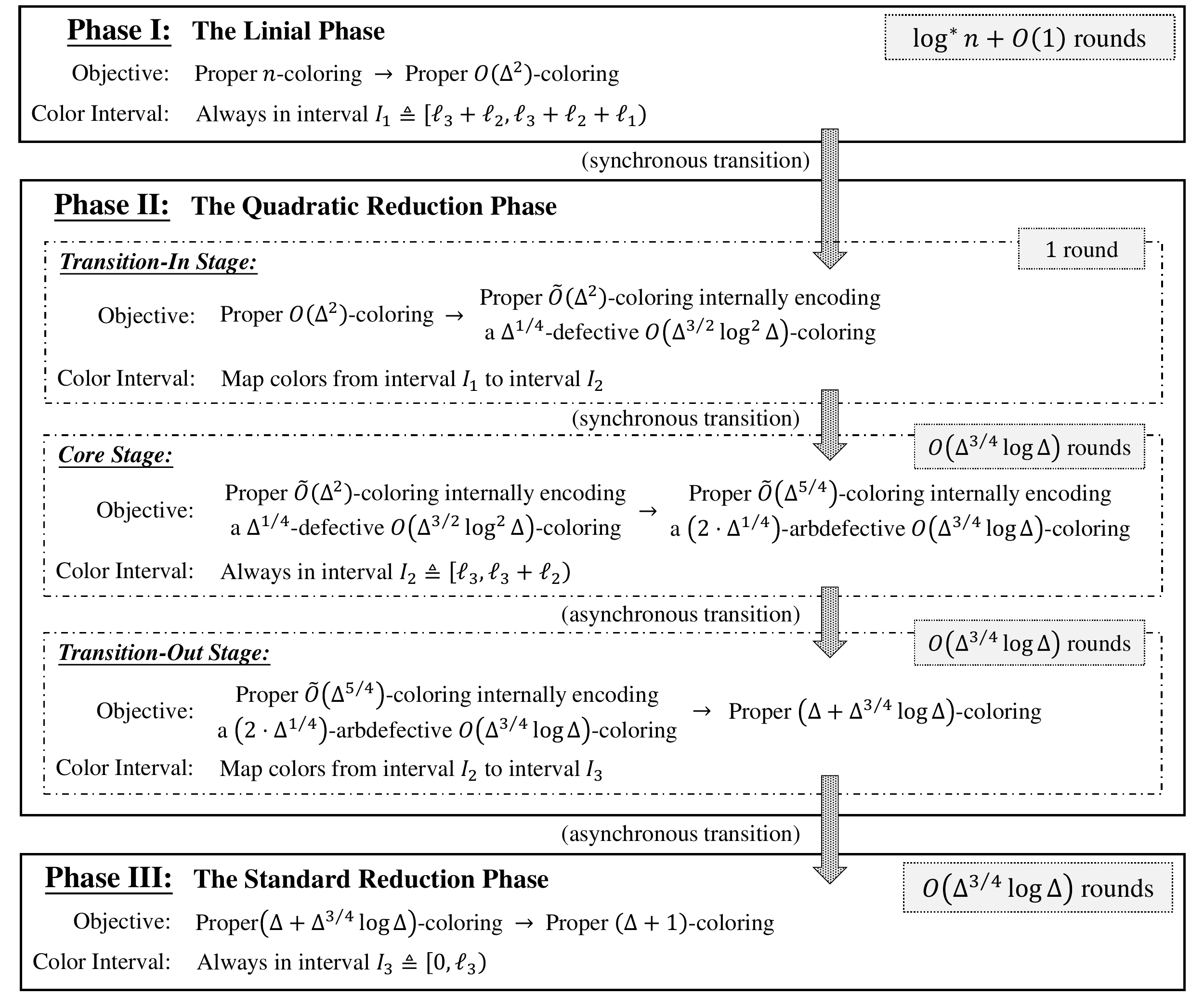}
\vspace{-2ex}
\caption{Structure of the locally-iterative $(\Delta+1)$-coloring algorithm.}\label{fig:alg-structure}
\vspace{-3ex}
\end{figure}

Our algorithm employs a three \emph{phases} framework that is used by many distributed coloring algorithms (e.g., \cite{linial87,kuhn06,barenboim16sublinear,barenboim21}). Specifically, in our algorithm: (1) the first ``Linial phase'' transforms an $n$-coloring to an $O(\Delta^2)$-coloring in $\log^*{n}+O(1)$ rounds; (2) the second ``quadratic reduction phase'' transforms an $O(\Delta^2)$-coloring to an $(\Delta+O(\Delta^{3/4}\log\Delta))$-coloring in $O(\Delta^{3/4}\log{\Delta})$ rounds; and (3) the last ``standard reduction phase'' transforms an $(\Delta+O(\Delta^{3/4}\log\Delta))$-coloring to a $(\Delta+1)$-coloring in $O(\Delta^{3/4}\log\Delta)$ rounds.

Inside the quadratic reduction phase are two \emph{transition stages} and one \emph{core stage}. In the first transition stage, which is the \emph{transition-in stage}, vertices use one round to transform Linial phase's proper $O(\Delta^2)$-coloring to a proper $\tilde{O}(\Delta^2)$-coloring which internally encodes a $\Delta^{1/4}$-defective $O(\Delta^{3/2}\log^2\Delta)$-coloring.%
\footnote{Throughout the paper, we use $\tilde{O}(\cdot)$ to hide poly-logarithmic terms in $\Delta$ (but \emph{not} in $n$) in the standard ${O}(\cdot)$ notation.}
Then, in the core stage, vertices use $O(\Delta^{3/4}\log\Delta)$ rounds to transform the proper $\tilde{O}(\Delta^2)$-coloring to a proper $\tilde{O}(\Delta^{5/4})$-coloring. Internally, the core stage is transforming the $\Delta^{1/4}$-defective $O(\Delta^{3/2}\log^2\Delta)$-coloring to a $(2\cdot\Delta^{1/4})$-arbdefective $O(\Delta^{3/4}\log\Delta)$-coloring. Lastly, in the second transition stage, which is the \emph{transition-out stage}, vertices use another $O(\Delta^{3/4}\log\Delta)$ rounds to transform the proper $\tilde{O}(\Delta^{5/4})$-coloring to a proper $(\Delta+O(\Delta^{3/4}\log\Delta))$-coloring.

We stress that, during execution, although our algorithm internally is working on improper colorings such as defective coloring and arbdefective coloring, with the help of cover-free set systems and coding, we ensure the coloring vertices produce at the end of each round is proper.

Being locally-iterative also means our algorithm cannot depend on the current round number to determine which phase it is in. To solve this issue, we assign each phase an interval so that vertices running that phase will have colors in the corresponding interval. By assigning disjoint intervals to different phases, vertices can correctly determine its progress by observing its current color.
More specifically, the intervals used by the three phases are $I_1$, $I_2$, and $I_3$, respectively, where
\begin{gather*}
|I_1|=\ell_1\text{, }|I_2|=\ell_2\text{, and }|I_3|=\ell_3, \\
I_1\triangleq[\ell_3+\ell_2,\ell_3+\ell_2+\ell_1)\text{, }I_2\triangleq[\ell_3,\ell_3+\ell_2)\text{, and }I_3\triangleq[0,\ell_3).
\end{gather*}
To give the precise values for $\ell_1,\ell_2,\ell_3$, we first define three integers and three primes numbers:
\begin{gather*}
m_1 \triangleq 4\Delta^{3/2}\log^2(n_{r^*}) \text{,~~} m_2 \triangleq 4\sqrt{\Delta}\log^2(n_{r^*}) \text{,~~and~~} m_3 \triangleq 16\sqrt{\Delta}\log^2(\lambda^2 m_2);\\
\lambda \in (\sqrt{m_1}+1,2(\sqrt{m_1}+1)] \text{,~~} \mu \in (\sqrt{\Delta}+\sqrt{m_3},2(\sqrt{\Delta}+\sqrt{m_3})] \text{,~~and~~} \tau \in (\sqrt{m_3},2\sqrt{m_3}].
\end{gather*}
Due to the Bertrand-Chebyshev theorem~\cite{chebyshev1852}, prime numbers $\lambda,\mu,\tau$ must exist. Then, we set:
\begin{align*}
|I_1| &= \ell_1=n+O\left(\log^*n\cdot\Delta^2\cdot\log^2{n}\right),\\
|I_2| &= \ell_2=2\lambda^3(\mu+1)\cdot m_3=O\left(\Delta^{13/4}\log^5\Delta\right),\\
|I_3| &= \ell_3=\Delta+(2\sqrt{m_3}+1)\cdot\mu=\Delta+O\left(\Delta^{3/4}\log{\Delta}\right).
\end{align*}

\Cref{fig:alg-structure} provides a graphical overview of the algorithm structure.

Before presenting the algorithm in detail, we state the key guarantees enforced by each phase.

\begin{lemma}[\textbf{Linial Phase}]\label{lemma:phase1-property-simple}
By the end of round $r^*=\log^*{n}+O(1)$, all vertices have completed the Linial phase, producing a proper coloring $\phi_{r^*}$ where $\phi_{r^*}(v)\in[\ell_3+\ell_2,\ell_3+\ell_2+O(\Delta^2))\subseteq I_1$ for every vertex $v$. Moreover, $\phi_{t}$ is proper for every round $t\in[1,r^*]$.
\end{lemma}

\begin{lemma}[\textbf{Quadratic Reduction Phase}]\label{lemma:phase2-property-simple}
By the end of round $r^*+2+3\lambda$, all vertices have completed the quadratic reduction phase, producing a proper coloring $\phi_{r^*+2+3\lambda}$ where $\phi_{r^*+2+3\lambda}(v)\in I_3$ for every vertex $v$. Moreover, $\phi_{t}$ is proper for every round $t\in[r^*+1,r^*+2+3\lambda]$.
\end{lemma}

\begin{lemma}[\textbf{Standard Reduction Phase}]\label{lemma:phase3-property-simple}
By the end of round $r^*+1+3\lambda+(2\sqrt{m_3}+1)\mu$, the coloring $\phi_{r^*+1+3\lambda+(2\sqrt{m_3}+1)\mu}$ is a proper $(\Delta+1)$-coloring. Moreover, $\phi_t$ is proper for every round $t\in[r^*+3+3\lambda,\infty)$.
\end{lemma}

In the reminder of this section, we will introduce each phase in detail and sketch the proof for above lemmas. We will conclude with a proof of the main theorem---i.e., \Cref{thm:alg-main}. Omitted details and proofs of the second phase and the last phase are provided in \Cref{sec-app:quadratic-reduction-phase} and \Cref{sec-app:standard-reduction-phase}.

\subsection{The Linial phase}\label{subsec:phase1}

The Linial phase runs a locally-iterative version of Linial's well-known coloring algorithm~\cite{linial87}. More specifically, let $n_0=n$ and for $i\ge 1$, define
$$
n_i=
\begin{cases}
	4(\Delta+1)^2\log^2 (n_{i-1}) & \text{if }n_{i-1}> 8(\Delta+1)^3,\\
	4(\Delta+1)^2 &\text{if }n_{i-1}\le 8(\Delta+1)^3.
\end{cases}
$$
Let $r^*$ be the smallest $r\ge 0$ such that $n_r\le 4(\Delta+1)^2$.
It has been shown $r^*\le\log^*n+O(1)$. (See, e.g., Section 3.10 of \cite{barenboim13}.)
During the Linial phase, vertices will reduce the number of colors used to $n_i$ after $i$ rounds, thus within $\log^*{n}+O(1)$ rounds the algorithm produces a proper $O(\Delta^2)$-coloring.

Recall the Linial phase assigns vertices with colors in interval $I_1=[\ell_3+\ell_2,\ell_3+\ell_2+\ell_1)$. We set:
$$|I_1|=\ell_1=\sum_{i=0}^{r^*} n_i=n+O(\log^*n\cdot\Delta^2\cdot\log^2{n}).$$
Furthermore, we partition $I_1$ into $r^*+1$ sub-intervals $I_1^{(0)},I_1^{(1)},\cdots,I_1^{(r^*)}$, such that for each $0\leq t\leq r^*$:
$I_1^{(t)}\triangleq[\ell_3+\ell_2+\sum_{t+1\le i\le r^*}n_i~~,~~\ell_3+\ell_2+\sum_{t\le i\le r^*}n_i)$.
Notice that $I_1^{(r^*)}=\left[\ell_3+\ell_2,\ell_3+\ell_2+n_{r^*}\right)$. In general, during the Linial phase, after $t$ rounds where $t\in[0,r^*]$, all vertices' colors are in interval $I_1^{(t)}$.

We now give the complete description of the Linial phase, which contains $r^*$ rounds. Recall that each vertex $v$ has a unique identity $id(v)\in[n]$, the initial color $\phi_0(v)$ of vertex $v$ is:
$$\phi_0(v)\gets\ell_3+\ell_2+\sum_{i=1}^{r^*}n_i+id(v),$$
clearly $\phi_0(v)\in I_1^{(0)}$. In any round $t\in[1,r^*]$, every vertex $v$ can correctly determine the value of $t$ by observing $\phi_{t-1}(v)$. Let $\mathcal{F}_{t-1}$ be a $\Delta$-cover-free set system of size $|\mathcal{F}_{t-1}|=n_{t-1}$ with ground set $I_1^{(t)}$, whose existence is guaranteed by \Cref{thm:set-families}. The elements of $\mathcal{F}_{t-1}$ are $\mathcal{F}_{t-1}\triangleq\{S_{t-1}^{(k)}\mid\text{ integers } k\in I_1^{(t-1)}\}$.
For any vertex $v$, the color $\phi_t(v)$ is set to be the smallest number in $S_{t-1}^{(\phi_{t-1}(v))}$, excluding all elements of $S_{t-1}^{(\phi_{t-1}(u ))}$ for all $v$'s neighbors $u\in N(v)$. Due to the $\Delta$-cover-freeness of $\mathcal{F}_{t-1}$, such color must exist.
Pseudocode of the Linial phase is provided in \Cref{sec-app:pseudocode}.

At this point, the following stronger version of \Cref{lemma:phase1-property-simple} is immediate by an induction on $t$.

\begin{lemma}\label{lemma:phase1-property}
For every $0\le t\le r^*$, the coloring $\phi_t$ is proper, and $\phi_t(v)\in I_1^{(t)}$ for every vertex $v\in V$.
\end{lemma}

\subsection{The quadratic reduction phase}\label{subsec:phase2}

The second phase is the most interesting and complex component of our algorithm, it is also the key for achieving sublinear-in-$\Delta$ runtime. This phase contains two transition stages and one core stage. Once the transition-in stage---which maps colors from interval $I_1$ to interval $I_2$---is done, during the core stage, vertices work with colors in interval $I_2$ and reduce the number of used colors quadratically; then in the transition-out stage, vertices map colors from interval $I_2$ to interval $I_3$.

Recall that we set $\ell_2=2\lambda^3(\mu+1)\cdot m_3$, hence for every color $(\ell_3+i)\in I_2$ where $i\in[\ell_2]$, we can use a unique quadruple $\langle a,b,c,d\rangle$ to identify it, where:
\begin{align*}
	a &= \ftp{i/\left(2\lambda(\mu+1)m_3\right)}, & c &= \ftp{\left(i-a\cdot2\lambda(\mu+1)m_3-b\cdot2\lambda(\mu+1)\right)/\left(\mu+1\right)},\\
	b &= \ftp{\left(i-a\cdot2\lambda(\mu+1)m_3\right)/\left(2\lambda(\mu+1)\right)}, & d &= i\bmod(\mu+1).
\end{align*}
In other words, $i= a\cdot2\lambda(\mu+1)m_3 + b\cdot2\lambda(\mu+1) + c\cdot(\mu+1) + d$. It is easy to verify that:
$$a\in[\lambda^2]\text{, }b\in[m_3]\text{, }c\in[2\lambda]\text{, and }d\in[\mu+1].$$
In the reminder of this paper, for any round $t\geq r^*+1$, for any vertex $v$, if $\phi_t(v)\in I_2$, then we use $a_t(v),b_t(v),c_t(v),d_t(v)$ to denote the values of $a(v),b(v),c(v),d(v)$ in $\phi_t(v)$.
Moreover, we often use $\ell_3+\langle a_t(v),b_t(v),c_t(v),d_t(v)\rangle$ to denote the color of $v$ at the end of round $t$ if $\phi_t(v)\in I_2$.

\paragraph{Algorithm description}

We now introduce each stage in detail, and we begin with the transition-in stage, which takes one round and transforms $\phi_{r^*}$ to a proper coloring with colors from the interval $I_2=[\ell_3,\ell_3+\ell_2)$. (Recall that the Linial phase takes $r^*$ rounds.) Specifically, we employ the defective coloring algorithm developed by Barenboim, Elkin, and Kuhn~\cite{barenboim14}, with suitable parameters tailored for our purpose.
The core of this approach is the usage of $\Delta$-union-$(\rho+1)$-cover-free set systems.
Recall \Cref{def:cover-free-set-system-extension}, a $\Delta$-union-$(\rho+1)$-cover-free set system is a set system $\mathcal{F}$ such that for every $\Delta+1$ distinct sets $S_0,S_1,\cdots,S_{\Delta}\in\mathcal{F}$, it holds that there exists at least one element $x\in S_0$ that appears in at most $\rho$ sets among $S_1,S_2,\cdots,S_{\Delta}$. In round $r^*+1$, let $\mathcal{F}_a\triangleq\{S_a^{(i)}\mid \text{ integers }i\in I_1^{(r^*)}\}$ be a $\Delta$-union-$\left(\Delta^{1/4}+1\right)$-cover-free set family with $[m_1]$ as its ground set. Such $\mathcal{F}_a$ exists due to \Cref{thm:set-families-extension}. Recall that by the end of round $r^*$, for each vertex $v$, its color $\phi_{r^*}(v)\in I_1^{(r^*)}$, and $v$ will send $\phi_{r^*}(v)$ to all its neighbors. In round $r^*+1$, for each vertex $v$, it chooses $a(v)$ from $S_a^{(\phi_{r^*}(v))}$. In particular, for every element $x\in S_a^{(\phi_{r^*}(v))}$, vertex $v$ computes the set of neighbors that also have $x$ in their respective $S_a^{(\cdot)}$ sets: $N'(v,x)\triangleq\{u\mid u\in N(v), \phi_{r^*}(u)\in I_1^{(r^*)}, x\in S_a^{(\phi_{r^*}(u))}\}$. Let $\hat{x}$ be the smallest element in $S_a^{(\phi_{r^*}(v))}$ satisfying $|N'(v,\hat{x})|\leq\Delta^{1/4}$, vertex $v$ then assigns $a(v)=\hat{x}\in[\lambda^2]$. Due to \Cref{def:cover-free-set-system-extension} and \Cref{thm:set-families-extension}, every vertex can find such $\hat{x}$ in round $r^*+1$.

By the end of round $r^*+1$, vertex $v$'s $a(v)$ may collide with up to $\Delta^{1/4}$ of its neighbors, as $\mathcal{F}_a$ is a $\Delta$-union-$\left(\Delta^{1/4}+1\right)$-cover-free set family. To resolve this issue, we build another $\Delta^{1/4}$-cover-free set family to assign different $b$ values to these potential colliding neighbors. Specifically, for each vertex $v$, let $N''(v)\triangleq\{u\mid u\in N(v),\phi_{r^*}(u)\in I_1^{(r^*)},a(v)\in S_a^{(\phi_{r^*}(u))}\}$ be the set of neighbors that might have colliding $a$ value. We know $|N''(v)|\leq\Delta^{1/4}$ due to previous discussion. Now, let $\mathcal{F}_b\triangleq\{S_b^{(i)}\mid \text{ integers }i\in I_1^{(r^*)}\}$ be a $\Delta^{1/4}$-cover-free set family with ground set $[m_2]\subseteq[m_3]$. Such $\mathcal{F}_b$ exists due to \Cref{thm:set-families}. Vertex $v$ assigns $b(v)$ to be an element in $S_b^{(\phi_{r^*}(v))}\setminus\bigcup_{u\in N''(v)} S_b^{(\phi_{r^*}(u))}$, which is guaranteed to exist due to \Cref{def:cover-free-set-system} and \Cref{thm:set-families}.

Lastly, we note that every vertex $v$ initializes $c(v)$ and $d(v)$ during the transition-in stage, though they are not used here. See \Cref{alg:phase2-stage1} for the pseudocode of the transition-in stage.

\begin{algorithm}[t!]
\caption{The transition-in stage of the quadratic reduction phase at $v\in V$ in round $t=r^*+1$}\label{alg:phase2-stage1}
\begin{algorithmic}[1]
\State Send $\phi_{t-1}(v)$ to all neighbors.
\If {($\phi_{t-1}(v)\in I_1$)}
	\State Determine the value of $t$ based on $\phi_{t-1}(v)$.
	\If{($t=r^*+1$)}
		\For {(every element $x\in S_a^{(\phi_{t-1}(v))}$)}
			\State $N'(v,x) \gets \left\{u \mid u\in N(v), \phi_{t-1}(u)\in I_1^{(t-1)}, x\in S_a^{(\phi_{t-1}(u))} \right\}$.
		\EndFor
		\State $a_{t}(v)\gets\min\left\{ x \mid x\in S_a^{(\phi_{t-1}(v))}, \left|N'(v,x)\right|\leq\Delta^{1/4} \right\}$.
		\State $N''(v)\gets\left\{u\mid u\in N(v),\phi_{t-1}(u)\in I_1^{(t-1)},a_{t}(v)\in S_a^{(\phi_{t-1}(u))}\right\}$.
		\State $b_{t}(v)\gets\min S_b^{(\phi_{t-1}(v))}\setminus\bigcup_{u\in N''(v)}S_b^{(\phi_{t-1}(u))}$.
		\State $c_{t}(v)\gets 0$, $d_{t}(v)\gets\mu$.
		\State $\phi_{t}(v) \gets \ell_3 + \langle a_{t}(v), b_{t}(v), c_{t}(v), d_{t}(v) \rangle$.
	\EndIf
\EndIf
\end{algorithmic}
\end{algorithm}

Once the transition-in stage is done, the $a$ values of all vertices correspond to a $\Delta^{1/4}$-defective coloring, using a palette containing $\lambda^2$ colors, as $a\in[\lambda^2]$. The main objective of the core stage is to start from this $\Delta^{1/4}$-defective $\lambda^2$-coloring to gradually obtain a $(2\cdot\Delta^{1/4})$-arbdefective $\lambda$-coloring.
Notice that this reduces the number of colors used---or more precisely, the range of the $a$ values of all vertices---from $[\lambda^2]$ to $[\lambda]$.
To achieve this quadratic reduction, for every vertex $v$, we interpret the first coordinate $a(v)$ of its color quadruple in the following manner:
$$a(v)=\hat{a}(v)\cdot\lambda+\tilde{a}(v)\text{, where }\hat{a}(v)=\lfloor a(v)/\lambda\rfloor\text{ and }\tilde{a}(v)=a(v)\bmod\lambda.$$
During the core stage, we run a locally-iterative arbdefective coloring algorithm inspired by \cite{barenboim21} that makes a series of updates to $a(v)$ so that eventually $\hat{a}(v)=0$, reducing $a(v)$ from $[\lambda^2]$ to $[\lambda]$.

More specifically, for each vertex $v$, in each round $t$ in the core stage where $a_{t-1}(v)\geq\lambda$, it will count the number of neighbors that also have colors in interval $I_2$ and satisfy ``$a_{t-1}(u)\neq a_{t-1}(v)$ and $a_{t-1}(u)\equiv a_{t-1}(v)\bmod\lambda$''. Denote this set of neighbors as:
$$M_t(v)\triangleq\{u\mid u\in N(v),\phi_{t-1}(u)\in I_2,\hat{a}_{t-1}(u)\neq \hat{a}_{t-1}(v),\tilde{a}_{t-1}(u)= \tilde{a}_{t-1}(v)\}.$$

If $|M_t(v)|>\Delta^{1/4}$, then $v$ updates $a_t(v)$ according to the following rule:
$$a_t(v) \gets \hat{a}_{t-1}(v)\cdot\lambda + ((\hat{a}_{t-1}(v)+\tilde{a}_{t-1}(v))\bmod\lambda).$$
Moreover, vertex $v$ keeps its $b,c,d$ values unchanged.

Otherwise, if $|M_t(v)|\leq\Delta^{1/4}$, then $v$ updates $a_t(v)$ to be $a_{t-1}(v)\bmod\lambda$, or equivalently:
$$a_t(v) \gets \tilde{a}_{t-1}(v).$$
Notice this step reduces the range of $a(v)$ from $[\lambda^2]$ to $[\lambda]$, completing the core stage for vertex $v$.
At this point, vertex $v$ will also set $c(v)$ in its color quadruple. ($d(v)$ is not used during core stage.)
\begin{align*}
    c_t(v)\gets &1+\max_{u\in M'_t(v)}\{c_{t-1}(u)\},\\
    &\text{ where }M'_t(v)\triangleq\{u\mid u\in N(v),\phi_{t-1}(u)\in I_2,\hat{a}_{t-1}(u)=0,\tilde{a}_{t-1}(u)=\tilde{a}_{t-1}(v)\}.
\end{align*}

The $c$ values of vertices implicitly define the orientations of edges: for neighbors $u$ and $v$, vertex $v$ points to vertex $u$ if and only if $c(v)\geq c(u)$.%
\footnote{In case $c(u)=c(v)$, the orientation of edge $(u,v)$ can be determined by comparing $id(u)$ and $id(v)$. However, our algorithm does not require $v$ to know $id(u)$, or vise versa. Instead, when $c(u)=c(v)$, vertex $v$ treats $(u,v)$ as pointing to $u$, and vertex $u$ treats $(u,v)$ as pointing to $v$. We shall show our algorithm still works under such interpretation.}
By guaranteeing that the out-degree of the oriented graph induced by each $a$ value is at most $2\cdot\Delta^{1/4}$, the $a$ and $c$ values of vertices together constitute a $(2\cdot\Delta^{1/4})$-arbdefective coloring during the core stage.
We also note that, since the maximum $c$ value attained by any vertex can increase by at most one in each round, and since we can show every vertex will reduce its $a$ value to $[\lambda]$ by the end of round $r^*+2+\lambda$, the algorithm guarantees the $c$ value of any vertex will never exceed $\lambda+1$.

To ensure $\phi_t$ is proper when $|M_t(v)|\leq\Delta^{1/4}$, vertex $v$ uses a $(2\cdot\Delta^{1/4})$-cover-free set family to assign its $b$ value. (The $b$ value generated by the transition-in stage already guarantees $\phi_t$ is proper when $|M_t(v)|>\Delta^{1/4}$.)
It can be seen as a variant of Linial's algorithm in that each vertex has a ``forbidden color list''.
More specifically, in our setting, recall that $\tau$ is a prime satisfying $2\cdot\Delta^{1/4}\log(\lambda^2 m_2)<\tau\leq2\cdot(2\cdot\Delta^{1/4}\log(\lambda^2 m_2))$. We construct a $(2\cdot\Delta^{1/4})$-cover-free set family in the following manner. For every integer $i\in[\lambda^2 m_2]$, we associate a unique polynomial $P_i$ of degree $\log(\lambda^2 m_2)$ over finite field $GF(\tau)$ to it.
(This is possible since the number of such polynomials is at least $(\tau-1)^{1+\log(\lambda^2 m_2)}>\lambda^2 m_2$.)
Let $\mathcal{F}_c\triangleq\{S_c^{(0)},S_c^{(1)},\cdots,S_c^{(\lambda^2 m_2-1)}\}$ be a set family of size $\lambda^2 m_2$, where $S_c^{(i)}\triangleq\{x\cdot \tau + P_i(x)\mid x\in[\tau]\}$ for every $i\in[\lambda^2 m_2]$. Since the degree of the polynomials $P_i$ is $\log(\lambda^2 m_2)$, the intersection of any two sets in $\mathcal{F}_c$ contains at most $\log(\lambda^2 m_2)$ elements. Since every set in $\mathcal{F}_c$ contains $\tau>2\cdot\Delta^{1/4}\log(\lambda^2 m_2)$ elements, $\mathcal{F}_c$ is $(2\cdot\Delta^{1/4})$-cover-free. Now, in a round where $|M_t(v)|\leq\Delta^{1/4}$, recall that vertex $v$ updates $a_t(v)$ to be $\tilde{a}_{t-1}(v)$. After this update, $v$'s $a$ value may collide with the $a$ values of the vertices in $M'_t(v)$, as well as the $a$ values of the vertices in
$$\overline{M}'_t(v)\triangleq\{u\mid u\in N(v),\phi_{t-1}(u)\in I_2,\hat{a}_{t-1}(u)\neq0,\tilde{a}_{t-1}(u)=\tilde{a}_{t-1}(v)\}.$$
We will show $|M'_t(v)\cup\overline{M}'_t(v)|\leq 2\cdot\Delta^{1/4}$, hence $b_t(v)$ can take the smallest value in the set below:
$$S_c^{(a_{t-1}(v)\cdot m_2 + b_{t-1}(v))}\setminus \left( \left\{b_{t-1}(u)\mid u\in M'_{t}(v)\right\} \bigcup \left(\cup_{u\in \overline{M}'_{t}(v)} S_c^{(a_{t-1}(u)\cdot m_2+b_{t-1}(u))}\right) \right).$$

Lastly, we note that vertices may complete the core stage at different times: once a vertex $v$ has $a_t(v)\in[\lambda]$ in some round $t$, its core stage is considered done and it may start the transition-out stage. As a result, starting from the core stage, vertices may proceed at different paces.

Complete pseudocode of the core stage is given in \Cref{alg:phase2-stage2}.

\begin{algorithm}[t!]
\caption{The core stage of the quadratic reduction phase at $v\in V$ in round $t\geq r^*+2$}\label{alg:phase2-stage2}
\begin{algorithmic}[1]
\State Send $\phi_{t-1}(v)$ to all neighbors.
\If{($\phi_{t-1}(v)\in I_2$ \textbf{and} $a_{t-1}(v)\geq\lambda$)}
	\State $M_t(v)\gets\left\{u\mid u\in N(v),\phi_{t-1}(u)\in I_2,\hat{a}_{t-1}(u)\neq \hat{a}_{t-1}(v),\tilde{a}_{t-1}(u)= \tilde{a}_{t-1}(v)\right\}$.
	\State $\overline{M}_t(v)\gets\left\{u\mid u\in N(v),\phi_{t-1}(u)\in I_2,\hat{a}_{t-1}(u)=\hat{a}_{t-1}(v),\tilde{a}_{t-1}(u)=\tilde{a}_{t-1}(v)\right\}$.
	\If{($|M_t(v)|\leq\Delta^{1/4}$)}
		\State $M'_{t}(v) \gets \left\{ u\mid u\in N(v),\phi_{t-1}(u)\in I_2,\hat{a}_{t-1}(u)=0,\tilde{a}_{t-1}(u)=\tilde{a}_{t-1}(v) \right\}$.
		\State $\overline{M}'_{t}(v) \gets \left\{ u\mid u\in N(v),\phi_{t-1}(u)\in I_2,\hat{a}_{t-1}(u)\neq 0,\tilde{a}_{t-1}(u)=\tilde{a}_{t-1}(v) \right\}$.
		\State $a_t(v) \gets \tilde{a}_{t-1}(v)$.\label{alg-line:phase2-stage2-a-update-rule-reduce}
		\State $b_t(v) \gets \min S_c^{(a_{t-1}(v)\cdot m_2 + b_{t-1}(v))}\setminus \left( \left\{b_{t-1}(u)\mid u\in M'_{t}(v)\right\} \bigcup \left(\cup_{u\in \overline{M}'_{t}(v)} S_c^{(a_{t-1}(u)\cdot m_2+b_{t-1}(u))}\right) \right)$.
		\State $c_t(v) \gets 1+\max\{c_{t-1}(u)\mid u\in M'_{t}(v)\}$.
	\Else
		\State $a_t(v) \gets \hat{a}_{t-1}(v)\cdot\lambda + ((\hat{a}_{t-1}(v)+\tilde{a}_{t-1}(v))\bmod\lambda)$.\label{alg-line:phase2-stage2-a-update-rule}
	\EndIf
	\State $\phi_{t}(v) \gets \ell_3 + \langle a_t(v), b_t(v), c_t(v), d_t(v) \rangle$.
\EndIf
\end{algorithmic}
\end{algorithm}

We continue to describe the transition-out stage, in which vertices produce a proper $(\Delta+O(\Delta^{3/4}\log\Delta))$-coloring using colors in interval $I_3$. The approach we took during the transition-out stage is inspired by the techniques developed by Barenboim~\cite{barenboim16sublinear}. Nonetheless, important adjustments are made on both the implementation and the analysis, as we are in the more restrictive locally-iterative setting, and have to take the ``asynchrony'' that vertices may start the transition-out stage in different rounds into consideration.

In a round $t$, for a vertex $v\in V$ with $\phi_{t-1}(v)\in I_2$ and $a_{t-1}(v)\in[\lambda]$, it runs the transition-out stage.
If $a_{t-1}(v)\leq a_{t-1}(u)<\lambda$ is satisfied for every $u\in N(v)$ with $\phi_{t-1}(u)\in I_2$, and if $d_{t-1}(v)=\mu$ (recall $d(v)$ always equal to $\mu$ during the core stage), then $v$ uses this round to update $d(v)$, making preparation for the transition. In particular, vertex $v$ considers a family of $\mu$ polynomials $P_{(t-1,v,0)}$, $P_{(t-1,v,1)}$, $\cdots$, $P_{(t-1,v,\mu-1)}$ over finite field $GF(\mu)$. For any $i\in[\mu]$, we define:
$$P_{(t-1,v,i)}(x) \triangleq (\lfloor b_{t-1}(v)/\tau \rfloor \cdot x^2 + (b_{t-1}(v)\bmod\tau)\cdot x + i)\bmod\mu.$$
Notice the core stage ensures $b_{t-1}(v)\in[\tau^2]$. Next, we define $L_{(t-1,i)}(v)$ and $L_{t-1}(v)$:
$$L_{(t-1,i)}(v)\triangleq\{P_{(t-1,v,i)}(x)+x\cdot\mu\mid x\in[\mu]\}\text{~~and~~}L_{t-1}(v)\triangleq\{\phi_{t-1}(u)\mid u\in N(v),\phi_{t-1}(u)\in I_3\}.$$
In other words, $L_{t-1}(v)$ contains the phase three colors that are already occupied by the neighbors of $v$. With $L_{(t-1,i)}(v)$ and $L_{t-1}(v)$, vertex $v$ sets $d_t(v)$ to be an integer $\hat{i}\in[\mu]$ that minimizes $|L_{(t-1,i)}(v)\cap L_{t-1}(v)|$. Since $|L_{t-1}(v)|\leq\Delta$, and since $L_{(t-1,i)}(v)\cap L_{(t-1,i')}(v)=\emptyset$ for any $i\neq i'$, by the pigeonhole principle, we have $|L_{(t-1,\hat{i})}(v)\cap L_{t-1}(v)|\leq\Delta/\mu$.

Once a vertex $v$ sets $d(v)$ to a value other than $\mu$, its preparation for the transition is done, and will attempt to maps its current color in $I_2$ to another color in $I_3$.

Specifically, in a round $t$, for a vertex $v\in V$ with $\phi_{t-1}(v)\in I_2$, it will update its color to $I_3$ if the following conditions are met: (a) $a_{t-1}(v)\leq a_{t-1}(u)<\lambda$ is satisfied for every $u\in N(v)$ with $\phi_{t-1}(u)\in I_2$; (b) $d_{t-1}(v)\neq\mu$; and (c) $d_{t-1}(u)\neq\mu$ is satisfied for every $u\in A_{t-1}(v)$, where $A_{t-1}(v)\triangleq\{u\mid u\in N(v),\phi_{t-1}(u)\in I_2, a_{t-1}(u)=a_{t-1}(v), c_{t-1}(u)\leq c_{t-1}(v)\}$ denotes the neighbors $v$ points to with colliding $a$ value. The update rule is, let $\hat{k}\in[\mu]$ be the smallest integer satisfying:
$$P_{(t-1,v,d_{t-1}(v))}(\hat{k}) + \mu\cdot{\hat{k}}~~\in~~L_{(t-1,d_{t-1}(v))}(v) \setminus \left( L_{t-1}(v) \bigcup \left(\cup_{u\in A_{t-1}(v)} L_{(t-1,d_{t-1}(u))}(u)\right) \right),$$
then set $\phi_t(v) \gets P_{(t-1,v,d_{t-1}(v))}(\hat{k}) + \mu\cdot{\hat{k}}$.

We will show $\phi_t(v)$ exists and $\phi_t(v)\in I_3$.
See \Cref{alg:phase2-stage3} for transition-out stage's pseudocode.

\begin{algorithm}[t!]
\caption{The transition-out stage of the quadratic reduction phase at $v\in V$ in round $t$}\label{alg:phase2-stage3}
\begin{algorithmic}[1]
\State Send $\phi_{t-1}(v)$ to all neighbors.
\If{($\phi_{t-1}(v)\in I_2$ \textbf{and} $a_{t-1}(v)<\lambda$)}\label{alg-line:phase2-stage3-if-cond-1}
	\If{(every $u\in N(v)$ with $\phi_{t-1}(u)\in I_2$ has  $a_{t-1}(v) \leq a_{t-1}(u)<\lambda$)}\label{alg-line:phase2-stage3-if-cond-2}
		\State $A_{t-1}(v)\gets\{u\mid u\in N(v),\phi_{t-1}(u)\in I_2, a_{t-1}(u)=a_{t-1}(v), c_{t-1}(u)\leq c_{t-1}(v)\}$.
		\For{(every $i\in[\mu]$)}
			\State $L_{(t-1,i)}(v)\gets\{ (\lfloor{b_{t-1}(v)/\tau}\rfloor\cdot{x^2}+(b_{t-1}(v)\bmod\tau)\cdot{x}+i)\bmod\mu + x\cdot\mu \mid x\in[\mu] \}$.
		\EndFor
		\State $L_{t-1}(v)\gets\{\phi_{t-1}(u)\mid u\in N(v),\phi_{t-1}(u)\in I_3\}$.
		\If{($d_{t-1}(v)=\mu$)}\label{alg-line:phase2-stage3-if-cond-3}
			\State Let $\hat{i}$ be an integer in $[\mu]$ that minimizes $|L_{(t-1,i)}(v)\cap L_{t-1}(v)|$.
			\State $d_t(v)\gets\hat{i}$.\label{alg-line:intermediate-to-reduction-d-rule}
			\State $\phi_t(v) \gets \ell_3 + \langle a_{t}(v),b_{t}(v),c_{t}(v),d_{t}(v)\rangle$.
		\ElsIf{(every $u\in A_{t-1}(v)$ has $d_{t-1}(u)\neq\mu$)}\label{alg-line:phase2-stage3-if-cond-4}
			\State Let $\hat{k}\in[\mu]$ be the smallest integer satisfying:\label{alg-line:phase2-stage3-k-rule}
			\Statex \hspace{12ex}
			\begin{small}
			$P_{(t-1,v,d_{t-1}(v))}(\hat{k}) + \mu\cdot{\hat{k}} \in L_{(t-1,d_{t-1}(v))}(v) \setminus \left( L_{t-1}(v) \bigcup \left(\cup_{u\in A_{t-1}(v)} L_{(t-1,d_{t-1}(u))}(u)\right) \right)$.
			\end{small}
			\State $\phi_t(v) \gets P_{(t-1,v,d_{t-1}(v))}(\hat{k}) + \mu\cdot{\hat{k}}$.
		\EndIf
	\EndIf
\EndIf
\end{algorithmic}
\end{algorithm}

\paragraph{Overview of analysis}

For the one-round transition-in stage, we have the following lemma, where the property $|\{u\mid u\in N(v),\phi_{r^*+1}(u)\in I_2,a_{r^*+1}(u)=a_{r^*+1}(v)\}|\leq\Delta^{1/4}$ means the $a$ values of all vertices correspond to a $\Delta^{1/4}$-defective coloring at the end of the transition-in stage.

\begin{lemma}\label{lemma:phase2-stage1-property}
By the end of round $r^*+1$, the coloring $\phi_{r^*+1}$ is proper, and $\phi_{r^*+1}(v)\in I_2$ for every $v\in V$. Moreover, for every $v\in V$, it holds that $|\{u\mid u\in N(v),\phi_{r^*+1}(u)\in I_2,a_{r^*+1}(u)=a_{r^*+1}(v)\}|\leq\Delta^{1/4}$.
\end{lemma}

Then, for the core stage, we have the following three lemmas.

\Cref{lemma:phase2-stage2-proper-coloring} concerns with correctness, it shows that vertices running the core stage always maintain a proper coloring with their $\langle a,b\rangle$ tuples. In fact, this lemma also covers the correctness for the majority of the transition-out stage, as for every vertex, in all but the last round of its transition-out stage, its $\langle a,b\rangle$ tuple remains unchanged.

\begin{lemma}\label{lemma:phase2-stage2-proper-coloring}
For every round $t\geq r^*+1$, let $V'_t$ be the set of vertices running the quadratic reduction phase in round $t$: $V'_t\triangleq\{v\mid \phi_t(v)\in I_2,v\in V\}$. Then, $\phi_t$ corresponds to a proper coloring for the subgraph $G'_t$ induced by the vertices in $V'_t$: for every $v\in V'_t$, it holds that $\phi_t(v)\notin\{\phi_t(u)\mid u\in N(v)\cap V'_t\}$. More precisely, if we regard the pair $\langle a_t(v), b_t(v)\rangle$ as the color of $v$, then this coloring is also proper in graph $G'_t$: for every $v\in V'_t$, it holds that $\langle a_t(v), b_t(v)\rangle \notin \{\langle a_t(u), b_t(u)\rangle \mid u\in N(v)\cap V'_t\}$.
\end{lemma}

\Cref{lemma:phase2-stage2-time-complexity} shows the core stage costs at most $(r^*+2+\lambda)-(r^*+1)=O(\Delta^{3/4}\log{\Delta})$ rounds for any vertex, as in our algorithm, once a vertex finds its $a$ value in $[\lambda]$, its core stage is done.

\begin{lemma}\label{lemma:phase2-stage2-time-complexity}
For every vertex $v\in V$, let $t^*_v$ be the smallest round number such that $\phi_{t^*_v}(v)\in I_2$ and $a_{t^*_v}(v)\in[\lambda]$ are both satisfied. Then, for every vertex $v\in V$, it holds that $t^*_v\leq r^*+2+\lambda$.
\end{lemma}

\Cref{lemma:phase2-stage2-bounded-arboricity} shows the $a$ values of the vertices running the core stage maintain a $(2\cdot\Delta^{1/4})$-arbdefective coloring, and we use the $c$ values of vertices to determine the orientation of edges. Together with \Cref{lemma:phase2-stage1-property}, one can see that the $a$ values of vertices transform from a $\Delta^{1/4}$-defective $\lambda^2$-coloring (recall by definition $a\in[\lambda^2]$) to a $(2\cdot\Delta^{1/4})$-arbdefective $\lambda$-coloring (recall the core stage of a vertex ends when its $a\in[\lambda]$) during the core stage.

\begin{lemma}\label{lemma:phase2-stage2-bounded-arboricity}
For every round $t\geq r^*+1$, for every $v\in V$ with $\phi_t(v)\in I_2$, it holds that $|\{u\mid u\in N(v),\phi_t(u)\in I_2, a_t(u)=a_t(v), c_t(u)\leq c_t(v)\}|\leq 2\cdot\Delta^{1/4}$.
\end{lemma}

Lastly, for the transition-out stage, we have \Cref{lemma:phase2-stage3-time-cost} for bounding its time cost, and \Cref{lemma:phase2-stage3-proper-color} for showing its correction. Notice that \Cref{lemma:phase2-stage2-time-complexity} and \Cref{lemma:phase2-stage3-time-cost} together show the time cost of the transition out stage is $(r^*+2+3\lambda)-(r^*+2+\lambda)=O(\Delta^{3/4}\log{\Delta})$ rounds for every vertex.

\begin{lemma}\label{lemma:phase2-stage3-time-cost}
For every vertex $v\in V$, let $t^{\#}_v$ be the smallest round number such that $\phi_{t^{\#}_v}(v)\in I_3$. Then, we have $t^{\#}_v \leq r^*+2+3\lambda$.
\end{lemma}

\begin{lemma}\label{lemma:phase2-stage3-proper-color}
For every vertex $v\in V$, let $t^{\#}_v$ be the smallest round number such that $\phi_{t^{\#}_v}(v)\in I_3$. Then, we have $\phi_{t^{\#}_v}(v)\notin\{\phi_{t^{\#}_v}(u)\mid u\in N(v), \phi_{t^{\#}_v}(u)\in I_3\}$.
\end{lemma}

With the above lemmas, we are able to prove \Cref{lemma:phase2-property-simple}. See \Cref{sec-app:quadratic-reduction-phase} for the omitted proofs.

\subsection{The standard reduction phase}\label{subsec:phase3}

In the standard reduction phase, each vertex $v$ maps color $\phi_{t^{\#}_v}\in I_3$ to another color in $[\Delta+1]\subset I_3$, completing $(\Delta+1)$-coloring. Here, $t^{\#}_v$ denotes the smallest round number such that $\phi_{t^{\#}_v}(v)\in I_3$.
Hence, $t^{\#}_v+1$ is the first round in which $v$ runs the standard reduction phase.

For each vertex $v$, for each round $t\geq t^{\#}_v+1$, if every neighbor $u\in N(v)$ has also entered the standard reduction phase, and if $v$ has the maximum color value in its one-hop neighborhood, then $v$ will update its color to be the minimum value in $[\Delta+1]$ that still has not be used by any of its neighbors. Clearly, such color must exist. In all other cases, $v$ keeps its color unchanged in round $t$.
Effectively, this procedure reduces the maximum color value used by any vertex by at least one in each round. Hence, within $\ell_3-(\Delta+1)$ rounds into the third phase, a proper $(\Delta+1)$-coloring is obtained.
Pseudocode of this phase is given in \Cref{alg:phase3} in \Cref{sec-app:pseudocode}.

The following two lemmas show: (1) the time cost of the standard reduction phase, which also bounds the total runtime of our algorithm; and (2) the correctness of this phase. See \Cref{sec-app:standard-reduction-phase} for their proofs. We also note that they together immediately imply \Cref{lemma:phase3-property-simple}.

\begin{lemma}\label{lemma:phase3-time-cost}
Every vertex $v$ has its color in $[\Delta+1]$ within $r^*+1+3\lambda+(2\sqrt{m_3}+1)\mu$ rounds.
\end{lemma}

\begin{lemma}\label{lemma:phase3-proper-color}
In every round $t\geq r^*+3+3\lambda$, the coloring $\phi_{t}$ is proper.
\end{lemma}

\subsection{Proof of the main theorem of the locally-iterative coloring algorithm}\label{subsec:alg-summary}

With the above lemmas in hand, we are ready to prove the main theorem for the locally-iterative coloring algorithm (i.e., \Cref{thm:alg-main}).

By \Cref{lemma:phase3-time-cost}, every vertex has its color in $[\Delta+1]$ at the end of round $r^*+1+3\lambda+2(\sqrt{m_3}+1)\mu=O(\Delta^{3/4}\log{\Delta})+\log^*{n}$. By \Cref{lemma:phase2-property-simple}, every vertex has its color in $I_3$ by the end of round $r^*+2+3\lambda$, and runs the standard reduction phase ever since. Hence, for every round $t\geq r^*+2+3\lambda+2(\sqrt{m_3}+1)\mu$, every vertex's color remains in $[\Delta+1]$ at the end of that round. On the other hand, \Cref{lemma:phase1-property-simple}, \Cref{lemma:phase2-property-simple}, and \Cref{lemma:phase3-property-simple} together suggest that the algorithm always maintains a proper coloring.
Lastly, recall the definition of interval length $\ell_1$, $\ell_2$, and $\ell_3$, any used color can be encoded by $O(\log{n})$ bits.
Since vertices only broadcast colors to neighbors, the bound on message size holds.

\section{The Self-stabilizing Coloring Algorithm}\label{sec:alg-stab-framework}

Though there are generic techniques for converting general coloring algorithms into self-stabilizing ones (e.g., \cite{lenzen09}), such approach often results in large message size, hence not suitable for our setting.
In this paper, we develop a new self-stabilizing coloring algorithm based on our locally-iterative coloring algorithm. It uses $O(\log{n})$-bits messages and stabilizes in $O(\Delta^{3/4}\log{\Delta})+\log^*{n}$ rounds. For this algorithm to work properly, in the ROM area of a vertex $v$, we store its identity $id(v)$, graph parameters $n$ and $\Delta$, and the program code. In the RAM area of $v$, we store the colors of its local neighborhood, a boolean vector $T_v$ of size $\Delta$, and other variables that are used during execution.

The boolean vector $T_v$ is used to determine the orientation of the edges incident to $v$, replacing the role of $c(v)$.
More specifically, in the self-stabilizing algorithm, for each edge $(v,u)$, vertex $v$ maintains a bit in the vector $T_v$ denoted as $T_v[u]$, and we treat $v$ points to $u$ if and only if $T_v[u]=1$. The reason that we replace $c(v)$ with bit vector $T_v$ is that in the self-stabilizing setting, the adversary can employ a certain strategy to grow the $c$ values indefinitely.

A side effect of replacing $c(v)$ with a vector $T_v$ is that vertex $v$ must maintain a variable for \emph{each} incident edge to determine its orientation. Moreover, for two neighbors $v$ and $u$ to correctly determine the orientation of edge $(v,u)$, bit entries $T_v[u]$ and $T_u[v]$ must be exchanged. Therefore, for every vertex $v$, it has to send different information to different neighbors (particularly, $T_v[u]$ for each neighbor $u$), making our self-stabilizing algorithm no longer locally-iterative per \Cref{def:locally-iter-alg}. Nonetheless, as mentioned in the introduction section, if we interpret locally-iterative from an ``edge orientation'' view point and allow vertices to maintain a state for each incident edge,
then our self-stabilizing algorithm becomes locally-iterative. In this section, for consistency and ease of presentation, we still introduce the self-stabilizing algorithm from the ``vertex centric'' view point. Moreover, we keep the $c$ entry in vertices' color quadruples, but they are not used throughout.

For each vertex $v$, the self-stabilizing algorithm still contains three phases: the Linial phase, the quadratic reduction phase, and the standard reduction phase. Initially, every vertex $v$ sets its color to $\phi_0(v)=\ell_3+\ell_2+\sum_1^{r^*} n_i+ id(v)$. At the beginning of each round $t$, for every neighbor $u\in N(v)$, vertex $v$ sends a message to $u$ including its current color $\phi_{t-1}(v)$ and a boolean variable $T_v[u]$. After receiving messages from neighbors, vertex $v$ will perform an \emph{error-checking} procedure to determine if it is in a proper state. If the error-checking passes then we say $v$ is in a \emph{proper} state, and $v$ updates its color and vector $T_v$ according to its local information and the messages received from neighbors. Otherwise, if the error-checking fails, $v$ is in an \emph{improper} state. In such case, $v$ resets its color.

Before presenting the self-stabilizing algorithm in more detail, we state the correctness guarantee enforced by its error-checking mechanism.

\begin{lemma}[\textbf{Correctness of the Self-stabilizing Algorithm}]\label{lemma:self-stabilizing-correctness}
If $T_0$ is the last round in which the adversary makes any changes to the RAM areas of vertices, then for every round $t\geq T_0+2$, for every vertex $v$, the error-checking procedure will not reset vertex $v$'s color.
\end{lemma}

In the reminder of this section, we will introduce the three phases of the self-stabilizing algorithm and state their time complexity. We will conclude with a proof of the main theorem---i.e., \Cref{thm:alg-self-stab}. Omitted details and missing proofs are provided in \Cref{sec-app:alg-stab}.

\subsection{The Linial phase and the transition-in stage of the quadratic reduction phase}

At the beginning of a round $t$, if a vertex $v$ finds its color $\phi_{t-1}(v)$ not in interval $I_2\cup I_3$, it will do error-checking to see if any of the following conditions is satisfied:
\begin{itemize}
	\item Its color collide with some neighbor.
	\item Its color $\phi_{t-1}(v)$ is not in $\left(\bigcup_{i=1}^{r^*}I_1^{(i)}\right)\cup I_2\cup I_3$ (which implies $v$ should be running the first iteration of the Linial phase), but that color is not $\ell_3+\ell_2+\sum_{i=1}^{r^*}n_i+id(v)$.
\end{itemize}
If any of these conditions is satisfied, then vertex $v$ treats itself in an improper state and resets its color to $\ell_3+\ell_2+\sum_{i=1}^{r^*}n_i+id(v)$. That is, it sets $\phi_t(v)=\ell_3+\ell_2+\sum_{i=1}^{r^*}n_i+id(v)$.

Otherwise, if vertex $v$ satisfies none of the conditions, then it is in a proper state with $\phi_{t-1}(v)\in I_1$. In such case, vertex $v$ first determines which interval $I_1^{(t')}$ it is in, and then runs either the Linial phase or the transition-in stage of the quadratic reduction phase, according to the value of $t'$.

If $0\leq t'< r^*$, then vertex $v$ computes a $\Delta$-cover-free set family $\mathcal{F}_{t'}$ as in the locally-iterative algorithm, and sets its new color to be the smallest number in $S_{t'}^{(\phi_{t-1}(v))}$, excluding all elements of $S_{t'}^{(\phi_{t-1}(v))}$ for all $v$'s neighbors $u\in N(v)$ satisfying $\phi_{t-1}(u)\in I_1^{(t')}$.

If $t'=r^*$, then vertex $v$ transforms its color from interval $I_1$ to $I_2$, effectively running the transition-in stage of the quadratic reduction phase. The transition-in stage of the self-stabilizing algorithm is similar to the one in the locally-iterative algorithm. The only difference is that vertices may end the Linial phase and start the transition-in stage in different rounds. This brings the side effect that the $a$ values of all vertices are no longer guaranteed to be $\Delta^{1/4}$-defective. Instead, we maintain a $\Delta^{1/4}$-arbdefective $\lambda^2$-coloring. Specifically, each vertex $v$ still computes $a(v)$ based on its color and the colors of its neighbors using the defective coloring algorithm; moreover, vertex $v$ again uses $b(v)$ to differentiate itself from the neighbors with the same $a$ value. On the other hand, vertex $v$ sets $T_v[u]=1$ if $a(v)$ might collide with neighbor $u$, otherwise $v$ sets $T_v[u]=0$. (Recall that $T_v[u]$ and $T_u[v]$ determine the orientation of edge $(u,v)$ in arbdefective coloring schemes.)

Following lemma shows the time cost of the algorithm up to end of the transition-in stage.

\begin{lemma}\label{lemma:self-stabilizing-time-complexity-I1}
If $T_0$ is the last round in which the adversary makes any changes to the RAM areas of vertices, then for every round $t\geq T_0+r^*+2$, every vertex $v$ has $\phi_t(v)\in I_2\cup I_3$.
\end{lemma}

\subsection{The core stage of the quadratic reduction phase}

A vertex $v$ with $\phi(v)\in I_2$ and  $a(v)\geq \lambda$ should run the core stage. Nonetheless, before proceeding, it will do error-checking to see if any of the following conditions is satisfied:
\begin{itemize}
	\item There exists a neighbor $u$ of $v$ such that $a(v)=a(u)$ and $b(v)=b(u)$, effectively implying $u$ and $v$ have identical color.
	\item There exists a neighbor $u$ of $v$ such that $a(v)=a(u)$ yet $T_v[u]+T_u[v]=0$, implying that the orientation of edge $(u,v)$ is still undetermined when $a(v)=a(u)$.
	\item The number of neighbors $u\in N(v)$ with $T_v[u]=1$ is larger than $\Delta^{1/4}$, violating the bounded arboricity assumption during the core stage.
	\item There exists a vertex $u\in N(v)\cup \{v\}$ with its color in $I_2$ and $a(u)\geq\lambda$, yet $b(u)\geq m_2$, violating the range of $b$ values during the core stage.
\end{itemize}
If any of these conditions is satisfied, then vertex $v$ resets its color. Otherwise, it executes the core stage of the quadratic reduction phase to reduce its $a$ value from $[\lambda,\lambda^2)$ to $[0,\lambda)$.

The procedure we use in the self-stabilizing settings to transform a $\Delta^{1/4}$-arbdefective $\lambda^2$-coloring to a $(2\cdot\Delta^{1/4})$-arbdefective $\lambda$-coloring is almost identical to the one we used in the locally-iterative settings.
The only difference is that we have altered the definition of some variables to incorporate relevant bits in $T_v$. This is because, in the self-stabilizing setting, vertices start the core stage with an arbdefective coloring instead of a defective coloring.

Once the reduction of the $a$ value occurs in some round $t$, vertex $v$ obtains an $a_t(v)\in[\lambda]$, and updates $b_t(v)$ to differentiate itself from the neighbors that may have colliding $a$ value.
It also sets $T_v[u]=1$ for certain entries in $T_v$, recording the orientation of corresponding edges. Notice that $T_v$ here is used to maintain the arboricity of a $(2\cdot\Delta^{1/4})$-arbdefective $\lambda$-coloring for vertices with $ a(v)<\lambda$, whereas in the transition-in stage, $T_v$ is used to maintain the arboricity of a $\Delta^{1/4}$-arbdefective $\lambda^2$-coloring for vertices with $ a(v)\geq\lambda$.

Following lemma shows the time cost of the self-stabilizing algorithm up to end of the core stage.

\begin{lemma}\label{lemma:self-stabilizing-time-complexity-I2-part1}
Assume $T_0$ is the last round in which the adversary makes any changes to the RAM areas of vertices, for every vertex $v$, let $t^*_v\geq T_0+r^*+2$ be the smallest round number such that either ``$\phi_{t^*_v}(v)\in I_2$ and $a_{t^*_v}(v)\in [\lambda]$'' or ``$\phi_{t^*_v}(v)\in I_3$'' is satisfied. Then, it holds that $t^*_v\leq T_0+r^*+3+\lambda$.
\end{lemma}

\subsection{The transition-out stage of the quadratic reduction phase}

At the beginning of a round $t$, if vertex $v$ has color $\phi_{t-1}(v)\in I_2$ and $a(v)\in[\lambda]$, then it is in the transition-out stage. Once again, it does the following error-checking before proceeding.
\begin{itemize}
	\item There exists a neighbor $u$ of $v$ such that $a(v)=a(u)$ and $b(v)=b(u)$, effectively implying $u$ and $v$ have identical color.
	\item There exists a neighbor $u$ of $v$ such that $a(v)=a(u)$ yet $T_v[u]+T_u[v]=0$, implying that the orientation of edge $(u,v)$ is still undetermined when $a(v)=a(u)$.
	\item The number of neighbors $u\in N(v)$ with $T_v[u]=1$ is larger than $2\cdot\Delta^{1/4}$, violating the bounded arboricity assumption during the transition-out stage.
\end{itemize}
If any of these conditions is satisfied, then vertex $v$ treats itself in an improper state and resets its color. Otherwise, it executes the transition-out stage to transform its color from $I_2$ to $I_3$.

For each vertex $v$, the transformation is similar to the transition-out stage of the locally-iterative algorithm, except that: (1) we replace the constraints on $c(v)$ with corresponding constraints on $T_v$; and (2) we add an error-checking mechanism for $d(v)$ as the adversary can arbitrarily change it. If the error-checking for $d(v)$ fails, vertex $v$ resets $d(v)$ to $\mu$, so that later it can obtain a proper $d(v)$.
Such resetting occurs at most once for each vertex once the adversary stops disrupting.

The following lemma gives the time cost of the self-stabilizing algorithm up to the end of the quadratic reduction phase.

\begin{lemma}\label{lemma:self-stabilizing-time-complexity-I2-part2}
Assume $T_0$ is the last round in which the adversary makes any changes to the RAM areas of vertices, for every vertex $v$, let $t^{\#}_v> T_0+r^*+\lambda+3$ be the smallest round number such that $\phi_{t^{\#}_v}(v)\in I_3$. Then, it holds that $t^{\#}_v\leq T_0+r^*+3+4\lambda$.
\end{lemma}

\subsection{The standard reduction phase}

For a vertex $v$ with its color in $I_3$, it considers itself in the standard reduction phase, whose error-checking procedure is very simple: if the color of itself collides with any neighbor, then it resets $\phi(v)$ to $\ell_3+\ell_2+\sum_1^{r^*} n_i+id(v)$. Otherwise, vertex $v$ considers itself in a proper state, and runs exactly the same standard reduction procedure described in the locally-iterative settings.

The following lemma gives an upper bound on the stabilization time of the self-stabilizing algorithm.

\begin{lemma}\label{lemma:self-stabilizing-time-complexity-I3}
Assume $T_0$ is the last round in which the adversary makes any changes to the RAM areas of vertices, for every round  $t\geq T_0+r^*+4\lambda+2+2(\sqrt{m_3}+1)\mu$ , every vertex $v$ has $\phi_{t}(v)\in [\Delta+1]$.
\end{lemma}

\subsection{Proof of the main theorem of the self-stabilizing coloring algorithm}

We can use above lemmas to prove \Cref{thm:alg-self-stab}---the main theorem of the self-stabilizing coloring algorithm.
Assume $T_0$ is the last round in which the adversary disrupts execution, by \Cref{lemma:self-stabilizing-time-complexity-I3}, every vertex has a color in $[\Delta+1]$ by the end of round $T_0+O(\Delta^{3/4}\log{\Delta})+\log^*n$, and that color will remain in $[\Delta+1]$ ever since. On the other hand, due to \Cref{lemma:self-stabilizing-correctness}, in every round $t\geq T_0+O(\Delta^{3/4}\log{\Delta})+\log^*n$, the error-checking procedure passes. As the error-checking procedure always checks whether neighbors have conflicting colors, in every round $t\geq T_0+O(\Delta^{3/4}\log{\Delta})+\log^*n$, the coloring produced at the end of that round is proper.
Lastly, recall the definition of interval length $\ell_1$, $\ell_2$, and $\ell_3$, and recall in each round, for each vertex $v$ and each of its neighbor $u$, vertex  $v$ only sends its color along with a bit $T_v[u]$ to $u$, hence the size of every message $v$ sends is $O(\log{n})$.
This completes the proof of \Cref{thm:alg-self-stab}.

\section{Conclusion}\label{sec:conclusion}

In this paper, we give the first locally-iterative $(\Delta+1)$-coloring algorithm with sublinear-in-$\Delta$ running time. This algorithm can also be transformed into a self-stabilizing algorithm, achieving sublinear-in-$\Delta$ stabilization time.
We introduce a notion of reconfiguration machinery for the locally-iterative algorithms that can be made  self-stabilizing with relative ease.
And interestingly, although the last self-stabilizing algorithm that we obtain is not locally-iterative \emph{per se}, 
it can be interpreted as a locally-iterative algorithm on the variables representing edge orientations, whereas this ``edge orientation'' variant of the locally-iterative algorithm supports reconfiguration.

Looking ahead, a natural question to ask is can locally-iterative algorithms do faster? Due to the trade-off between the runtime of the intermediate phase and the number of colors used in the coloring produced by the intermediate phase, $\tilde{O}(\Delta^{3/4})+\log^*n$ might be the best  achievable upper bound in the current algorithmic framework. Nevertheless, the possibility that more elaborate tools or more clever techniques could result in faster algorithms still exist, and this is a very interesting direction worth further exploration. On the other hand, compared with the seminal work by Barenboim, Elkin and Goldenberg~\cite{barenboim21}, our algorithm is more sophisticated and is not applicable in some settings (that algorithms in \cite{barenboim21} could work), such as the Bit-Round model. Finding a more elegant and ``natural'' sublinear-in-$\Delta$ locally-iterative coloring algorithm and perhaps supporting more settings, is another direction for future research.

\bibliographystyle{ACM-Reference-Format}
\bibliography{podc23-ref-v1}

\clearpage
\appendix
\section*{Appendix}

\section{Pseudocode of the Linial Phase and the Standard Reduction Phase of the Locally-iterative Algorithm}\label{sec-app:pseudocode}

\begin{algorithm}[h!]
\caption{The Linial phase at $v\in V$ in round $1\leq t\leq r^*$}\label{alg:phase1}
\begin{algorithmic}[1]
\Statex /* Initialization: $\phi_0(v) \gets \ell_3+\ell_2+\sum_{i=1}^{r^*}n_i+id(v)$. */
\State Send $\phi_{t-1}(v)$ to all neighbors.
\If {($\phi_{t-1}(v)\in I_1$)}
	\State Determine the value of $t$ based on $\phi_{t-1}(v)$.
	\If{($1\leq t\leq r^*$)}
		\State $\phi_t(v)\gets\min S_{t-1}^{(\phi_{t-1}(v))}\setminus\bigcup_{u\in N(v)\text{ and }\phi_{t-1}(u)\in I_1^{(t-1)}} S_{t-1}^{(\phi_{t-1}(u))}$.
	\EndIf
\EndIf
\end{algorithmic}
\end{algorithm}

\begin{algorithm}[h!]
\caption{The standard reduction phase at $v\in V$ in round $t$}\label{alg:phase3}
\begin{algorithmic}[1]
\State Send $\phi_{t-1}(v)$ to all neighbors.
\If{($\phi_{t-1}(v)\in I_3$)}
	\If {(for every $u\in N(v)$ it holds that $\phi_{t-1}(u)\in I_3$)}
		\If {(for every $u\in N(v)$ it holds that $\phi_{t-1}(u)<\phi_{t-1}(v)$)}
			\State $\phi_t(v)\gets\min([\Delta+1]\setminus \{\phi_{t-1}(u)\mid u\in N(v)\})$.
		\EndIf
	\EndIf
\EndIf
\end{algorithmic}
\end{algorithm}

\section{Omitted Details and Proofs of The Quadratic Reduction Phase of the Locally-iterative Algorithm}\label{sec-app:quadratic-reduction-phase}

In this section, we provide missing details on the description of the quadratic reduction phase, and prove \Cref{lemma:phase2-stage1-property} to \Cref{lemma:phase2-stage3-proper-color}. We conclude this section with a proof of \Cref{lemma:phase2-property-simple}.

\subsection{Transition-in stage}\label{subsec:phase2-stage1}

There are no missing details on the description of the transition-in stage. Moreover, by the description provided in the main body of the paper, it is easy to see that $\phi_{r^*+1}$ is a proper coloring, and the $a_{r^*+1}$ values of all vertices correspond to a $\Delta^{1/4}$-defective coloring, immediately giving \Cref{lemma:phase2-stage1-property}.

\subsection{Core stage}\label{subsec:phase2-stage2}

Recall that in the core stage, when $|M_t(v)|\leq\Delta^{1/4}$ for a vertex $v$, it assigns $b_t(v)$ to take the smallest value in the following set:
$$S_c^{(a_{t-1}(v)\cdot m_2 + b_{t-1}(v))}\setminus \left( \left\{b_{t-1}(u)\mid u\in M'_{t}(v)\right\} \bigcup \left(\cup_{u\in \overline{M}'_{t}(v)} S_c^{(a_{t-1}(u)\cdot m_2+b_{t-1}(u))}\right) \right).$$

There are some details worth clarifying regarding the above expression. First, the indices $a_{t-1}(v)\cdot m_2 + b_{t-1}(v)$ and $a_{t-1}(u)\cdot m_2 + b_{t-1}(u)$ in the above expression are valid. To see this, notice that when the transition-in stage is done, according to the transition-in stage algorithm, each vertex's $b$ value is in $[m_2]$. Hence, when vertex $v$ reduces its $a$ value from $[\lambda^2]$ to $[\lambda]$ in round $t$, we have $a_{t-1}(v)\cdot m_2 + b_{t-1}(v)\in[\lambda^2 m_2]$. Moreover, for each vertex $u\in\overline{M}'_{t}(v)$, by the definition of $\overline{M}'_{t}(v)$ and the above algorithm description, its value of $b$ has not changed since the transition-in stage is done (otherwise it would be the case $\hat{a}_{t-1}(u)=0$), thus the value of $b_{t-1}(u)$ must be in $[m_2]$. Therefore, for each vertex $u\in\overline{M}'_{t}(v)$, we also have $a_{t-1}(u)\cdot m_2 + b_{t-1}(u)\in[\lambda^2 m_2]$. Second, the above expression gives a non-empty set. To see this, notice that by definition any $S^{(\cdot)}_c$ contains at least $\tau>2\cdot\Delta^{1/4}\log^2(\lambda^2 m_2)$ elements, and we are eliminating at most $2\cdot\Delta^{1/4}\log(\lambda^2 m_2)$ elements from it with the expression after the set-minus symbol, as $|M'_t(v)\cup\overline{M}'_t(v)|\leq 2\cdot\Delta^{1/4}$. Lastly, after vertex $v$ updates its $b$ value, we have $b_t(v)\in[m_3]$. This is because $b_t(v)$ is drawn from $S_c^{(a_{t-1}(v)\cdot m_2 + b_{t-1}(v))}$, which by definition only contains elements in $[\tau^2]\subseteq[m_3]$.

\paragraph{Analysis}

We now formally prove the correctness of the core stage and analyze its time complexity.
We first show the following claim is true as it will be frequently used later.

\begin{claim}\label{claim:phase2-stage2_last_round_equal_a}
For every round $t\geq r^*+2$, for every pair of neighboring vertices $u$ and $v$, if $\phi_t(u), \phi_t(v)\in I_2$ and $a_t(u)=a_t(v)\geq\lambda$, then $\phi_{t-1}(u), \phi_{t-1}(v) \in I_2$ and $a_{t-1}(u)=a_{t-1}(v)\geq\lambda$.
\end{claim}

\begin{proof}
Since $t-1\geq r^*+1$, vertex $v$ cannot be in the Linial phase in round $t-1$. In such scenario, by our algorithm, if $\phi_{t-1}(v)\notin I_2$, then it cannot be the case that $\phi_{t}(v)\in I_2$. Hence, if $\phi_t(v)\in I_2$, then $\phi_{t-1}(v)\in I_2$; similarly, if $\phi_t(u)\in I_2$, then $\phi_{t-1}(u)\in I_2$. Moreover, $u$ and $v$ must both be executing \Cref{alg:phase2-stage2} in round $t$.

Next, we prove $a_t(u)=a_t(v)\geq\lambda$ implies $a_{t-1}(u)=a_{t-1}(v)$. For the sake of contradiction, assume $a_{t-1}(u)\neq a_{t-1}(v)$. Since $a_t(u)=a_t(v)\geq\lambda$, in round $t$, both $u$ and $v$ update $a$ using the rule in Line \ref{alg-line:phase2-stage2-a-update-rule} of \Cref{alg:phase2-stage2}. If $a_{t-1}(u)\neq a_{t-1}(v)$, then either $\hat{a}_{t-1}(u)\neq \hat{a}_{t-1}(v)$ or $\tilde{a}_{t-1}(u)\neq \tilde{a}_{t-1}(v)$. In case $\hat{a}_{t-1}(u)\neq \hat{a}_{t-1}(v)$, assume $\hat{a}_{t-1}(u)<\hat{a}_{t-1}(v)$ without loss of generality. Then after the update, we have $a_t(u)\leq\hat{a}_{t-1}(u)\cdot\lambda+(\lambda-1)<\hat{a}_{t-1}(v)\cdot\lambda\leq a_t(v)$, meaning $a_t(u)\neq a_t(v)$, resulting in a contradiction. Otherwise, in case $\hat{a}_{t-1}(u)=\hat{a}_{t-1}(v)$ and $\tilde{a}_{t-1}(u)\neq\tilde{a}_{t-1}(v)$, then $\tilde{a}_t(u)=((\hat{a}_{t-1}(u)+\tilde{a}_{t-1}(u))\bmod\lambda)\neq((\hat{a}_{t-1}(v)+\tilde{a}_{t-1}(v))\bmod\lambda)=\tilde{a}_t(v)$. Once again, we have $a_t(u)\neq a_t(v)$, resulting in a contradiction. By now, we conclude $a_{t-1}(u)=a_{t-1}(v)$.

Lastly, notice that if both $u$ and $v$ execute \Cref{alg:phase2-stage2} in round $t$, and if $a_{t-1}(u)=a_{t-1}(v)<\lambda$, then it must be the case that $a_t(u)=a_t(v)<\lambda$, violating the lemma assumption. Hence, we know $a_{t-1}(u)=a_{t-1}(v)\geq\lambda$.
\end{proof}

Next, we bound the size of the set $M'_t(v)\cup\overline{M}'_t(v)$. It is crucial in showing that the $\langle a,b\rangle$ pairs of vertices maintain a proper coloring during the core stage of the quadratic reduction phase.

\begin{claim}\label{claim:phase2-stage2_collisions_bound}
For every round $t\geq r^*+2$, for every vertex $v$, if $\phi_{t-1}(v)\in I_2$ and $a_{t-1}(v)\geq\lambda$ and $|M_t(v)|\leq\Delta^{1/4}$, then $|M'_t(v)\cup\overline{M}'_t(v)|\leq 2\cdot\Delta^{1/4}$.
\end{claim}

\begin{proof}
Since $t\geq r^*+2$, $\phi_{t-1}(v)\in I_2$, and $a_{t-1}(v)\geq\lambda$, vertex $v$ executes \Cref{alg:phase2-stage2} in round $t$.

Define $\overline{M}_t(v)\triangleq\{u\mid u\in N(v),\phi_{t-1}(u)\in I_2,\hat{a}_{t-1}(u)=\hat{a}_{t-1}(v),\tilde{a}_{t-1}(u)=\tilde{a}_{t-1}(v)\}$. Notice that by definition $M_t(v)\cup\overline{M}_t(v)=M'_t(v)\cup\overline{M}'_t(v)$, thus we only need to prove $|M_t(v)\cup\overline{M}_t(v)|\leq 2\cdot\Delta^{1/4}$. Recall the lemma assumption $|M_t(v)|\leq\Delta^{1/4}$, thus we focus on showing $|\overline{M}_t(v)|\leq\Delta^{1/4}$.

For each vertex $u\in\overline{M}_t(v)$, by the definition of $\overline{M}_t(v)$, we have $u\in N(v)$, $\phi_{t-1}(u)\in I_2$, and $a_{t-1}(u)=a_{t-1}(v)\geq\lambda$. By repeatedly applying \Cref{claim:phase2-stage2_last_round_equal_a}, we conclude that $a_{r^*+1}(u)=a_{r^*+1}(v)$; that is, when the transition-in stage is done, $u$ and $v$ have identical $a$ value. Recall \Cref{lemma:phase2-stage1-property}, we know when the transition-in stage is done, the number of neighbors of $v$ that have identical $a$ value with $v$ cannot exceed $\Delta^{1/4}$. Therefore, $|\overline{M}_t(v)|\leq\Delta^{1/4}$. This completes the proof of the claim.
\end{proof}

We are now ready to prove that $\phi$---or more precisely, the $\langle a,b\rangle$ pairs of vertices---maintains a proper coloring for vertices with colors in interval $I_2$.\footnote{During the transition-out stage, vertices will not alter their $a,b,c$ values. Hence, at this point, we can already argue that the algorithm maintains a proper coloring for vertices with colors in interval $I_2$ throughout their second phase.}

\begin{proof}[Proof of \Cref{lemma:phase2-stage2-proper-coloring}]
We prove the lemma by induction on $t$. For the base case $t=r^*+1$, by \Cref{lemma:phase2-stage1-property}, $V'_{r^*+1}=V$ and $\phi_{r^*+1}$ corresponds to a proper coloring. Since $c_{r^*+1}(v)=0,d_{r^*+1}(v)=\mu$ for every vertex $v$, we further conclude the $\langle a,b\rangle$ pairs of all vertices correspond to a proper coloring.

Assume the claim holds for round $t=i$ where $i\geq r^*+1$, we now consider round $t=i+1$.

In round $i+1$, a vertex $v\in V'_{i+1}$ running the quadratic reduction phase may have $a_{i}(v)\in [\lambda,\lambda^2)$ or $a_{i}(v)\in [0,\lambda)$. In the former case, $v$ may update its $a$ value from $[\lambda,\lambda^2)$ to $[\lambda,\lambda^2)$ or reduce its $a$ value from $[\lambda,\lambda^2)$ to $[0,\lambda)$ in round $i+1$. In the latter case, $v$ leaves its $a$ and $b$ values unchanged in round $i+1$. We consider these three scenarios separately.

\textsc{Scenario I:} vertex $v$ running \Cref{alg:phase2-stage2} updates its $a$ value from $[\lambda,\lambda^2)$ to $[\lambda,\lambda^2)$ in round $i+1$. For any vertex $u\in N(v)\cap V'_{i+1}$ with $a_{i+1}(u)\neq a_{i+1}(v)$, the claim holds trivially. On the other hand, by claim \Cref{claim:phase2-stage2_last_round_equal_a}, any vertex $u\in N(v)\cap V'_{i+1}$ with $a_{i+1}(u)=a_{i+1}(v)\geq\lambda$ has $u\in N(v)\cap V'_{i}$ and $a_{i}(u)=a_{i}(v)\geq\lambda$ as well. By the induction hypothesis, we have $b_{i} (u)\neq b_{i}(v)$. By \Cref{alg:phase2-stage2}, vertex $u$ and $v$ both update $a$ from  $[\lambda,\lambda^2)$ to $[\lambda,\lambda^2)$ in round $i+1$. Moreover, we have $b_{i+1}(u)=b_{i}(u)$ and $b_{i+1}(v)=b_i(v)$, implying $b_{i+1}(u)\neq b_{i+1}(v)$. Hence, the $\langle a,b\rangle$ pairs of all vertices in $V'_{i+1}$ correspond to a proper coloring of $G'_{i+1}$.

\textsc{Scenario II:} vertex $v$ running \Cref{alg:phase2-stage2} reduces its $a$ value from $[\lambda,\lambda^2)$ to $[0,\lambda)$ in round $i+1$. For any vertex $u\in N(v)\cap V'_{i+1}$ with $a_{i+1}(u)\neq a_{i+1}(v)$, the claim holds trivially. So consider a vertex $u\in N(v)\cap V'_{i+1}$ with $a_{i+1}(u)=a_{i+1}(v)<\lambda$.
By \Cref{alg:phase2-stage2}, vertex $u$ either: (a) satisfies $a_i(u)<\lambda$ and does not change its $a,b$ values in round $i+1$; or (b) reduces its $a$ value from $[\lambda,\lambda^2)$ to $[0,\lambda)$ in round $i+1$. In both cases, it is easy to verify that $\tilde{a}_{i}(u)=\tilde{a}_{i}(v)$ must hold. This implies $u\in M'_{i+1}(v)\cup\overline{M}'_{i+1}(v)$. Now, since $v$ reduces its $a$ value from $[\lambda,\lambda^2)$ to $[0,\lambda)$ in round $i+1$, the condition $|M_{i+1}(v)|\leq\Delta^{1/4}$ must be satisfied in round $i+1$. Hence, by \Cref{claim:phase2-stage2_collisions_bound} and the method we used to update $b(v)$, it holds that $b_{i+1}(v)\neq b_{i+1}(u)$ for any $u\in M'_{i+1}(v)\cup\overline{M}'_{i+1}(v)$.

\textsc{Scenario III:} vertex $v$ running \Cref{alg:phase2-stage2} leaves its $a$ value and $b$ value unchanged in round $i+1$ since $a_i(v)\in[\lambda]$. For any vertex $u\in N(v)\cap V'_{i+1}$ with $a_{i+1}(u)\neq a_{i+1}(v)$, the claim holds trivially. So consider a vertex $u\in N(v)\cap V'_{i+1}$ with $a_{i+1}(u)=a_{i+1}(v)<\lambda$.
By \Cref{alg:phase2-stage2}, vertex $u$ either: (a) satisfies $a_i(u)<\lambda$ and does not change its $a,b$ values in round $i+1$; or (b) reduces its $a$ value from $[\lambda,\lambda^2)$ to $[0,\lambda)$ in round $i+1$. In the former case, we know $a_i(v)=a_{i+1}(v)=a_{i+1}(u)=a_i(u)$. By the induction hypothesis, we know $b_i(u)\neq b_i(v)$. Since $a_i(v)=a_i(u)<\lambda$, by \Cref{alg:phase2-stage2}, we conclude $b_{i+1}(u)=b_i(u)\neq b_i(v)=b_{i+1}(v)$. This completes the proof of the inductive step for case (a). Next, consider case (b), in which $u$ reduces its $a$ value from $[\lambda,\lambda^2)$ to $[0,\lambda)$ in round $i+1$. From the perspective of $u$, by an analysis similar to \textsc{Scenario II}, we know $v\in M'_{i+1}(u)\cup\overline{M}'_{i+1}(u)$. Moreover, by \Cref{claim:phase2-stage2_collisions_bound} and the method we used to update $b(u)$, it holds that $b_{i+1}(u)\neq b_{i+1}(v)$ for any $v\in M'_{i+1}(u)\cup\overline{M}'_{i+1}(u)$. This completes the proof of the inductive step for case (b).
\end{proof}

We continue to show the core stage maintains a $(2\cdot\Delta^{1/4})$-arbdefective coloring with the $a$ values of the vertices. Recall that vertices use the $c$ values to implicitly determine the orientation of edges: vertex $v$ points to vertex $u$ if and only if $c(v)\geq c(u)$. To simplify presentation, for each vertex $v$, in a round $t$ during its core stage, we use $A_t(v)$ to define the set of vertices $v$ points to:
$$A_t(v)\triangleq\{u\mid u\in N(v),\phi_t(u)\in I_2, a_t(u)=a_t(v), c_t(u)\leq c_t(v)\}.$$

The following proof employs a similar strategy as that of Lemma 6.2 in \cite{barenboim21}.

\begin{proof}[Proof of \Cref{lemma:phase2-stage2-bounded-arboricity}]
For every vertex $v\in V$, let $t^*_v$ be the smallest round number such that $\phi_{t^*_v}(v)\in I_2$ and $a_{t^*_v}(v)\in[\lambda]$ are both satisfied. Fix some round $\hat{t}\geq r^*+1$, we prove the lemma by considering two complement scenarios: either $\hat{t}<t^*_v$ or $\hat{t}\geq t^*_v$.

\textsc{Scenario I:} $\hat{t}<t^*_v$. In this scenario, for every round $t\in[r^*+1,\hat{t}]$, the value of $a_t(v)$ is at least $\lambda$. We shall prove a superset of $A_{\hat{t}}(v)$ is of size at most $\Delta^{1/4}$. Specifically, we claim the size of $\{u\mid u\in N(v),\phi_{\hat{t}}(u)\in I_2,a_{\hat{t}}(u)=a_{\hat{t}}(v)\}$ is at most $\Delta^{1/4}$. To see this, choose an arbitrary vertex $u\in\{u\mid u\in N(v),\phi_{\hat{t}}(u)\in I_2,a_{\hat{t}}(u)=a_{\hat{t}}(v)\}$. Since $\phi_{\hat{t}}(u),\phi_{\hat{t}}(v)\in I_2$ and $a_{\hat{t}}(u)=a_{\hat{t}}(v)\geq\lambda$, by repeatedly applying \Cref{claim:phase2-stage2_last_round_equal_a}, we know $\phi_{r*+1}(u),\phi_{r^*+1}(v)\in I_2$ and $a_{r^*+1}(u)=a_{r^*+1}(v)\geq\lambda$. Due to \Cref{lemma:phase2-stage1-property}, we know the number of neighbors of $v$ satisfying $a_{r^*+1}(u)=a_{r^*+1}(v)$ cannot exceed $\Delta^{1/4}$. Therefore, $|A_{\hat{t}}(v)|\leq|\{u\mid u\in N(v),\phi_{\hat{t}}(u)\in I_2,a_{\hat{t}}(u)=a_{\hat{t}}(v)\}|\leq\Delta^{1/4}$, as required.

\textsc{Scenario II:} $\hat{t}\geq t^*_v$. In this scenario, we prove the claim by induction, starting from round $t^*_v$.

Consider round $t^*_v$, if $t^*_v=r^*+1$, then due to \Cref{lemma:phase2-stage1-property}, $|A_{t^*_v}(v)|\leq\Delta^{1/4}$, as required. Otherwise, we have $t^*_v>r^*+1$, implying $v$ runs \Cref{alg:phase2-stage2} in round $t^*_v$. By the definition of $t^*_v$, vertex $v$ reduces its $a$ value from $[\lambda,\lambda^2)$ to $[0,\lambda)$ in round $t^*_v$. Hence, by \Cref{alg:phase2-stage2}, $|M_{t^*_v}(v)|\leq\Delta^{1/4}$. Next, we argue $|\overline{M}_{t^*_v}(v)|\leq\Delta^{1/4}$. To see this, choose an arbitrary vertex $u\in\overline{M}_{t^*_v}(v)$. By the definition of $\overline{M}_{t^*_v}(v)$, we know $a_{t^*_v-1}(u)=a_{t^*_v-1}(v)\geq\lambda$. By repeatedly applying \Cref{claim:phase2-stage2_last_round_equal_a}, we know $\phi_{r*+1}(u),\phi_{r^*+1}(v)\in I_2$ and $a_{r^*+1}(u)=a_{r^*+1}(v)\geq\lambda$. Due to \Cref{lemma:phase2-stage1-property}, we know the number of neighbors of $v$ satisfying $a_{r^*+1}(u)=a_{r^*+1}(v)$ cannot exceed $\Delta^{1/4}$. Therefore, $|\overline{M}_{t^*_v}(v)|\leq\Delta^{1/4}$. At this point, we conclude $|A_{t^*_v}(v)|\leq|\{u\mid u\in N(v),\phi_t(u)\in I_2,a_{t^*_v}(u)=a_{t^*_v}(v)\}|\leq|M_{t^*_v}(v)\cup\overline{M}_{t^*_v}(v)|\leq 2\cdot\Delta^{1/4}$. This completes the proof of the base case.

Assume $|A_{i}(v)|\leq 2\cdot\Delta^{1/4}$ holds for round $i\geq t^*_v$, consider round $i+1$. Since $i\geq t^*_v$, we have $a_{i}(v)\in [0,\lambda)$. Thus in round $i+1\geq r^*+2$, by \Cref{alg:phase2-stage2}, vertex $v$ does not change its $a,b,c$ values. Particularly, $a_{i+1}(v)=a_i(v)$ and $c_{i+1}(v)=c_i(v)$. On the other hand, for any vertex $u\in N(v)$ satisfying $\phi_{i+1}(u)\in I_2$ and $a_{i+1}(u)=a_{i+1}(v)<\lambda$, by the definition of $t^*_u$, it holds that $t^*_u\leq i+1$. If $t^*_u<i+1$, then we have $a_{i+1}(u)=a_{i}(u)$ and $c_{i+1}(u)=c_{i}(u)$. Thus, in the case $t^*_u< i+1$, if $u\in A_{i+1}(v)$ then $u\in A_{i}(v)$. Otherwise, consider the case $t^*_u=i+1$. Since $t^*_u=i+1\geq r^*+2$, in round $i+1$, vertex $u$ runs \Cref{alg:phase2-stage2} and reduces its $a$ value from $[\lambda,\lambda^2)$ to $[0,\lambda)$. By the method \Cref{alg:phase2-stage2} updates vertices' $c$ values, it must be $c_{i+1}(u)> c_{i+1}(v)$. Thus, in the case $t^*_u=i+1$, vertex $u\notin A_{i+1}(v)$. At this point, we can conclude $A_{i+1}(v)\subseteq A_{i}(v)$. By the induction hypothesis, $|A_{i+1}(v)|\leq 2\cdot\Delta^{1/4}$. This completes the proof of the inductive step.
\end{proof}

We conclude this part by bounding the duration of the core stage: starting from round $r^*+1$, within $\lambda+2=O(\Delta^{3/4}\log\Delta)$ rounds, all vertices have their $a$ values in $[\lambda]$. Recall this guarantee is summarized in \Cref{lemma:phase2-stage2-time-complexity}.

\begin{proof}[Proof of \Cref{lemma:phase2-stage2-time-complexity}]\label{proof:lemma:phase2-stage2-time-complexity}
By \Cref{lemma:phase2-stage1-property}, every vertex $v\in V$ has $\phi_{r*+1}(v)\in I_2$. If vertex $v$ has $a_{r^*+1}(v)\in[0,\lambda)$, then trivially $t^*_v=r^*+1$ and we are done. Otherwise, vertex $v$ has $a_{r^*+1}(v)\in [\lambda,\lambda^2)$, and runs \Cref{alg:phase2-stage2} from round $r^*+2$ to $t^*_v$ (both inclusive). To bound the value of $t^*_v$ when $a_{r^*+1}(v)\in [\lambda,\lambda^2)$, consider a vertex $u\in N(v)$ such that $\phi_{r^*+1}(u)\in I_2$.

\smallskip Our first claim is, if $a_{r^*+1}(u)\neq a_{r^*+1}(v)$, then in rounds $[r^*+2,\min\{r^*+1+\lambda,t^*_v-1\}]$, there are at most two rounds such that $u$ and $v$ both have their colors in $I_2$ and have identical $\tilde{a}$ value (by the end of those rounds). To prove this claim, consider three scenarios depending on the value of $t^*_u$.

\textsc{Scenario I:} $t^*_u=r^*+1$. Consider a round $t\in[r^*+2,\min\{r^*+1+\lambda,t^*_v-1\}]$. Since $t\leq t^*_v-1$, vertex $v$ updates its $a$ value in rounds $r^*+2$ to $t$ (both inclusive) using Line \ref{alg-line:phase2-stage2-a-update-rule} of \Cref{alg:phase2-stage2}. This implies $\tilde{a}_t(v)=(\tilde{a}_{r^*+1}(v)+(t-r^*-1)\cdot\hat{a}_{r^*+1}(v))\bmod\lambda$. Since $t^*_u=r^*+1$, in all rounds from $r^*+2$ to $t$ (both inclusive) in which $\phi(u)\in I_2$ (at the beginning of those rounds), $a(u)$ always equal to $a_{r^*+1}(u)\in[\lambda]$ (at the end of those rounds). In particular, $\hat{a}_t(u)=0$ and $\tilde{a}_t(v)=\tilde{a}_{r^*+1}(u)$. Now, to satisfy $\tilde{a}_t(v)=\tilde{a}_t(u)$, the following equality must hold:
$$(t-r^*-1)\cdot\hat{a}_{r^*+1}(v)+(\tilde{a}_{r^*+1}(v)-\tilde{a}_{r^*+1}(u))\equiv 0 \pmod \lambda.$$
Recall that we have assumed $a_{r^*+1}(u)\neq a_{r^*+1}(v)$, also recall that $\hat{a}_{r^*+1}(u)=0\neq\hat{a}_{r^*+1}(v)$, so in the above expression $\tilde{a}_{r^*+1}(v)$ may or may not equal to $\tilde{a}_{r^*+1}(u)$. Nonetheless, recall that $\hat{a}_{r^*+1}(v)$, $\tilde{a}_{r^*+1}(v)$, $\tilde{a}_{r^*+1}(u)$ are all in $[\lambda]$, also recall that $\lambda$ is a prime number, so when $t\in[r^*+2,\min\{r^*+1+\lambda,t^*_v-1\}]$, there is at most one choice of $t$ that satisfies the above expression.

\textsc{Scenario II:} $t^*_u\geq t^*_v$. Consider a round $t\in[r^*+2,\min\{r^*+1+\lambda,t^*_v-1\}]$. Since $t\leq t^*_v-1$, vertex $v$ updates its $a$ value in rounds $r^*+2$ to $t$ (both inclusive) using Line \ref{alg-line:phase2-stage2-a-update-rule} of \Cref{alg:phase2-stage2}. This implies $\tilde{a}_t(v)=(\tilde{a}_{r^*+1}(v)+(t-r^*-1)\cdot\hat{a}_{r^*+1}(v))\bmod\lambda$. Since $t\leq t^*_v-1\leq t^*_u-1$, we can similarly conclude $\tilde{a}_t(u)=(\tilde{a}_{r^*+1}(u)+(t-r^*-1)\cdot\hat{a}_{r^*+1}(u))\bmod\lambda$. Now, to satisfy $\tilde{a}_t(v)=\tilde{a}_t(u)$, the following equality must hold:
$$(t-r^*-1)\cdot (\hat{a}_{r^*+1}(v)-\hat{a}_{r^*+1}(u))+(\tilde{a}_{r^*+1}(v)-\tilde{a}_{r^*+1}(u))\equiv 0 \pmod \lambda.$$
Recall that we have assumed $a_{r^*+1}(u)\neq a_{r^*+1}(v)$, also recall that $\hat{a}_{r^*+1}(v)$, $\hat{a}_{r^*+1}(u)$, $\tilde{a}_{r^*+1}(v)$, $\tilde{a}_{r^*+1}(u)$ are all in $[\lambda]$ and $\lambda$ is a prime number. If $\hat{a}_{r^*+1}(v)=\hat{a}_{r^*+1}(u)$ and $\tilde{a}_{r^*+1}(v)\neq\tilde{a}_{r^*+1}(u)$, then the above expression cannot be satisfied. Otherwise, if $\hat{a}_{r^*+1}(v)\neq\hat{a}_{r^*+1}(u)$, then there is at most one choice of $t\in[r^*+2,\min\{r^*+1+\lambda,t^*_v-1\}]$ that satisfies the above expression.

\textsc{Scenario III:} $r^*+2\leq t^*_u\leq t^*_v-1$. In this scenario, in rounds $r^*+2$ to $\min\{r^*+1+\lambda,t^*_u-1\}$ (both inclusive), by an argument similar to \textsc{Scenario II}, there is at most one round in which $u$ and $v$ both have their colors in $I_2$ and have identical $\tilde{a}$ value (by the end of that round). In round $t^*_u$, the value of $a(u)$ reduces from $[\lambda,\lambda^2)$ to $[0,\lambda)$. Hence, either $\tilde{a}_{t^*_u}(v)=\tilde{a}_{t^*_u}(u)$, or $\tilde{a}_{t^*_u}(v)\neq\tilde{a}_{t^*_u}(u)$.
\begin{itemize}
	\item If $\tilde{a}_{t^*_u}(v)=\tilde{a}_{t^*_u}(u)$, then $t^*_u$ is another round in which $u$ and $v$ both have their colors in $I_2$ and have identical $\tilde{a}$ value (by the end of that round). Next, consider a round $t\in[t^*_u+1,\min\{r^*+1+\lambda,t^*_v-1\}]$. Since $t\leq t^*_v-1$, vertex $v$ updates its $a$ value in rounds $[t^*_u+1,t]$ using Line \ref{alg-line:phase2-stage2-a-update-rule} of \Cref{alg:phase2-stage2}. This implies $\tilde{a}_t(v)=(\tilde{a}_{t^*_u}(v)+(t-t^*_u)\cdot\hat{a}_{t^*_u}(v))\bmod\lambda$. Since $t\geq t^*_u+1$, in all rounds $[t^*_u+1,t]$ in which $\phi(u)\in I_2$, $a(u)$ always equal to $a_{t^*_t}(u)\in[\lambda]$. In particular, $\hat{a}_t(u)=0$ and $\tilde{a}_t(v)=\tilde{a}_{t^*_u}(u)$. If $\tilde{a}_t(v)=\tilde{a}_t(u)$ is satisfied, then it must be the case:
	$$(t-t^*_u)\cdot\hat{a}_{t^*_u}(v)+(\tilde{a}_{t^*_u}(v)-\tilde{a}_{t^*_u}(u))\equiv 0 \pmod \lambda.$$
	Recall that $\tilde{a}_{t^*_u}(v)=\tilde{a}_{t^*_u}(u)$ and $\hat{a}_{t^*_u}(v)\neq 0$. Since $\lambda$ is a prime number, the above expression can only be satisfied when $t-t^*_u\equiv 0 \pmod \lambda$. However, since $r^*+2\leq t^*_u$ and $t\in[t^*_u+1,\min\{r^*+1+\lambda,t^*_v-1\}]$, we know $1\leq t-t^*_u\leq \lambda-1$, implying $t-t^*_u\equiv 0 \pmod \lambda$ cannot be satisfied. At this point, we conclude, if $\tilde{a}_{t^*_u}(v)=\tilde{a}_{t^*_u}(u)$, then in rounds $[r^*+2,\min\{r^*+1+\lambda,t^*_v-1\}]$, there are at most two rounds in which $u$ and $v$ both have their colors in $I_2$ and have identical $\tilde{a}$ value (by the end of those rounds).
	
	\item If $\tilde{a}_{t^*_u}(v)\neq\tilde{a}_{t^*_u}(u)$, then consider a round $t\in[t^*_u+1,\min\{r^*+1+\lambda,t^*_v-1\}]$. By an analysis identical to the above, we know if $\tilde{a}_t(v)=\tilde{a}_t(u)$ is satisfied, then it must be the case:
	$$(t-t^*_u)\cdot\hat{a}_{t^*_u}(v)+(\tilde{a}_{t^*_u}(v)-\tilde{a}_{t^*_u}(u))\equiv 0 \pmod \lambda.$$
	Recall that $\tilde{a}_{t^*_u}(v)\neq\tilde{a}_{t^*_u}(u)$ and $\hat{a}_{t^*_u}(v)\neq 0$. Since $\lambda$ is a prime number, and since $1\leq t-t^*_u\leq \lambda-1$, we know there is at most one choice of $t\in[t^*_u+1,\min\{r^*+1+\lambda,t^*_v-1\}]$ that satisfies the above expression. At this point, we conclude, if $\tilde{a}_{t^*_u}(v)\neq\tilde{a}_{t^*_u}(u)$, then in rounds $[r^*+2,\min\{r^*+1+\lambda,t^*_v-1\}]$, there are at most two rounds in which $u$ and $v$ both have their colors in $I_2$ and have identical $\tilde{a}$ value (by the end of those rounds).
\end{itemize}

By now, we have proved our first claim. That is, for any pair of vertices $u,v$ such that $u\in N(v)$ and $\phi_{r^*+1}(u)\in I_2, \phi_{r^*+1}(v)\in I_2$, if $a_{r^*+1}(u)\neq a_{r^*+1}(v)\in[\lambda,\lambda^2)$, then in rounds $r^*+2$ to $\min\{r^*+1+\lambda,t^*_v-1\}$ (both inclusive), there are at most two rounds such that $u$ and $v$ both have their colors in $I_2$ and have identical $\tilde{a}$ value (by the end of those rounds).

\smallskip Our second claim is, for any pair of vertices $u,v$ such that $u\in N(v)$ and $\phi_{r^*+1}(u)\in I_2, \phi_{r^*+1}(v)\in I_2$, if $a_{r^*+1}(u)=a_{r^*+1}(v)\in[\lambda,\lambda^2)$, then in rounds $r^*+2$ to $\min\{r^*+1+\lambda,t^*_v-1\}$ (both inclusive), there are at most two rounds such that by the end of each such round, $u$ and $v$ both have their colors in $I_2$, $u$ and $v$ have identical $\tilde{a}$ value but different $a$ values.

To prove this claim, consider two complement scenarios depending on the value of $t^*_u$. In scenario one in which $t^*_u\geq t^*_v$, by \Cref{alg:phase2-stage2}, for any round $t\in[r^*+2,\min\{r^*+1+\lambda,t^*_v-1\}]$, we have $\phi_{t}(u)=\phi_{t}(v)\in I_2$ and $a_t(u)=a_t(v)$. The second scenario concerns with the case $t^*_u< t^*_v$. In this scenario, by \Cref{alg:phase2-stage2}, for any round $t\in[r^*+2,t^*_u-1]$, we have $\phi_{t}(u)=\phi_{t}(v)\in I_2$ and $a_t(u)=a_t(v)$. Then, in round $t^*_u$, vertex $u$ reduces its $a$ value from $[\lambda,\lambda^2)$ to $[0,\lambda)$, whereas $v$ keeps its $a$ value in $[\lambda,\lambda^2)$. Thus, by \Cref{alg:phase2-stage2}, $\phi_{t^*_u}(u)=\phi_{t^*_u}(v)\in I_2$, $a_{t^*_u}(u)\neq a_{t^*_u}(v)$, but $\tilde{a}_{t^*_u}(u)=\tilde{a}_{t^*_u}(v)$. Lastly, for every round $t\in[t^*_u+1,\min\{r^*+1+\lambda,t^*_v-1\}]$, by an argument identical to {Scenario III} in the preceding claim, we know $\tilde{a}_{t^*_u}(u)\neq\tilde{a}_{t^*_u}(v)$. This completes the proof of our second claim.

\smallskip Combining the two claims, we conclude, for any pair of neighbors $u,v$ such that $\phi_{r^*+1}(u)$ and $\phi_{r^*+1}(v)$ are both in $I_2$, in rounds $r^*+2$ to $\min\{r^*+1+\lambda,t^*_v-1\}$ (both inclusive), there are at most two rounds such that by the end of each such round, $u$ and $v$ both have their colors in $I_2$, $u$ and $v$ have identical $\tilde{a}$ value but different $a$ values.

Now, since vertex $v$ has at most $\Delta$ neighbors, and since $\lambda\geq 2\cdot\Delta^{3/4}$, by the pigeonhole principle, starting from round $r^*+2$, within $\lambda$ rounds, there must exist a round $t$ in which, by the end of that round, the number of neighbors of $v$ satisfying both $a_t(u)\neq a_t(v)$ and $\tilde{a}_t(u)=\tilde{a}_t(v)$ is at most $\Delta^{1/4}$. In other words, in round $t+1$, we have $|M_{t+1}(v)|\leq\Delta^{1/4}$. As a result, by \Cref{alg:phase2-stage2}, at the end of round $t+1\leq r^*+2+\lambda$, we have $a_{t+1}(v)\in[\lambda]$. This completes the proof of the lemma.
\end{proof}

\subsection{Transition-out stage}\label{subsec:phase2-stage3}

There are no missing details on the description of the transition-out stage, so we proceed to the analysis directly. We begin by showing that the parameter $\hat{k}$ defined in Line \ref{alg-line:phase2-stage3-k-rule} of \Cref{alg:phase2-stage3} must exist and is of bounded value.

\begin{claim}\label{claim:phase2-stage3-bounded-k}
For any vertex $v\in V$, let $t^+_v$ be the smallest round number such that $\phi_{t^+_v}(v)\in I_2$ and $d_{t^+_v}(v)\neq\mu$.
For any round $t\geq t^+_v+1$, if $\phi_{t-1}(v)\in I_2$, then the following set is non-empty:
$$L_{(t-1,d_{t-1}(v))}(v) \setminus \left( L_{t-1}(v) \bigcup \left(\cup_{u\in\{w\mid w\in A_{t-1}(v),d_{t-1}(w)\neq\mu\}}L_{(t-1,d_{t-1}(u))}(u)\right) \right).$$
Let $\hat{k}\in[\mu]$ be the smallest integer such that $P_{(t-1,v,d_{t-1}(v))}(\hat{k})+\mu\cdot{\hat{k}}$ is in the above set, then:
$$0\leq\hat{k} \leq \Delta/\mu + 4\cdot\Delta^{1/4}.$$
\end{claim}

\begin{proof}
Before diving into the details, we outline the high-level proof strategy. Recall the claim statement, for the ease of presentation, we define:
\begin{align*}
\mathfrak{A}_{\phantom{0}} & \triangleq L_{(t-1,d_{t-1}(v))}(v),\\
\mathfrak{B}_1 & \triangleq L_{t-1}(v),\\
\mathfrak{B}_2 & \triangleq \bigcup_{u\in\{w\mid w\in A_{t-1}(v),d_{t-1}(w)\neq\mu\}}L_{(t-1,d_{t-1}(u))}(u).
\end{align*}
To prove $\mathfrak{A}\setminus(\mathfrak{B}_1\cup\mathfrak{B}_2)\neq\emptyset$, we will show $|\mathfrak{A}\cap(\mathfrak{B}_1\cup\mathfrak{B}_2)|\leq|\mathfrak{A}\cap\mathfrak{B}_1|+|\mathfrak{A}\cap\mathfrak{B}_2|<|\mathfrak{A}|$, as $\mathfrak{A}\setminus(\mathfrak{B}_1\cup\mathfrak{B}_2)=\mathfrak{A}\setminus(\mathfrak{A}\cap(\mathfrak{B}_1\cup\mathfrak{B}_2))$. On the other hand, recall that $L_{(t-1,d_{t-1}(v))}(v)$, or $\mathfrak{A}$ equivalently, denotes the set $\{P_{(t-1,v,d_{t-1}(v))}(x)+x\cdot\mu\mid x\in[\mu]\}$. Moreover, the value of $P_{(t-1,v,d_{t-1}(v))}(x)+x\cdot\mu$ strictly increases as $x$ increases. As a result, to find $\hat{k}\in[\mu]$, which is the smallest integer such that $P_{(t-1,v,d_{t-1}(v))}(\hat{k})+\mu\cdot{\hat{k}}$ is in $\mathfrak{A}\setminus(\mathfrak{B}_1\cup\mathfrak{B}_2)$, it suffices to bound the size of the set $\mathfrak{A}\cap(\mathfrak{B}_1\cup\mathfrak{B}_2)$. In particular, $\hat{k}$ is at most the $(|\mathfrak{A}\cap(\mathfrak{B}_1\cup\mathfrak{B}_2)|+1)$-th smallest element in $[\mu]$; in other words, $\hat{k}\leq|\mathfrak{A}\cap(\mathfrak{B}_1\cup\mathfrak{B}_2)|$. To sum up, to prove the claim, it suffices to show $|\mathfrak{A}\cap(\mathfrak{B}_1\cup\mathfrak{B}_2)| \leq \Delta/\mu+4\cdot\Delta^{1/4}$, since by then we can conclude: (a) by definition $\Delta/\mu+4\cdot\Delta^{1/4}<\mu$, thus $|\mathfrak{A}\cap(\mathfrak{B}_1\cup\mathfrak{B}_2)|<\mu=|\mathfrak{A}|$, implying $\mathfrak{A}\setminus(\mathfrak{B}_1\cup\mathfrak{B}_2)\neq\emptyset$; and (b) $\hat{k}\leq|\mathfrak{A}\cap(\mathfrak{B}_1\cup\mathfrak{B}_2)|\leq\Delta/\mu+2\cdot\Delta^{1/4}$.

We now proceed to prove $|\mathfrak{A}\cap(\mathfrak{B}_1\cup\mathfrak{B}_2)| \leq \Delta/\mu+4\cdot\Delta^{1/4}$, and we do so by bounding the size of $\mathfrak{A}\cap\mathfrak{B}_1$ and $\mathfrak{A}\cap\mathfrak{B}_2$.

Consider a vertex $v$ and a round $t\geq t^+_v+1$ with $\phi_{t-1}(v)\in I_2$. By the definition of $t^+_v$ , the definition of $L_{(t-1,d_{t-1}(v))}(v)$, and algorithm description, it holds that:
$$L_{(t-1,d_{t-1}(v))}(v)=L_{\left(t-1,d_{t^+_v}(v)\right)}(v)=L_{\left(t^+_v-1,d_{t^+_v}(v)\right)}(v).$$
As a result:
$$|\mathfrak{A}\cap\mathfrak{B}_1| = |L_{(t-1,d_{t-1}(v))}(v) \cap L_{t-1}(v)| = \left| L_{\left(t^+_v-1,d_{t^+_v}(v)\right)}(v) \cap L_{t-1}(v) \right|.$$
Observe that as time proceeds from round $t^+_v$ to round $t$, more and more neighbors of $v$ may have done the transition-out stage and start running the third phase; in other words, $L_{t-1}(v)$ may increase as $t$ increases. More precisely, we have:
\begin{align*}
|\mathfrak{A}\cap\mathfrak{B}_1| &= \left| L_{\left(t^+_v-1,d_{t^+_v}(v)\right)}(v) \cap L_{t-1}(v) \right| \\
&= \left| L_{\left(t^+_v-1,d_{t^+_v}(v)\right)}(v) \cap L_{t^+_v-1}(v) \right| + \left| L_{\left(t^+_v-1,d_{t^+_v}(v)\right)}(v) \cap \left(L_{t-1}(v) \setminus L_{t^+_v-1}(v)\right) \right| \\
&\leq \Delta/\mu + \left| L_{\left(t^+_v-1,d_{t^+_v}(v)\right)}(v) \cap \left(L_{t-1}(v) \setminus L_{t^+_v-1}(v)\right) \right| ,
\end{align*}
where the last inequality is due to the fact that $| L_{\left(t^+_v-1,d_{t^+_v}(v)\right)}(v) \cap L_{t^+_v-1}(v) |\leq\Delta/\mu$ (recall we have argued why this is the case when describing \Cref{alg:phase2-stage3}).

On the other hand, notice that:
\begin{align*}
|\mathfrak{A}\cap\mathfrak{B}_2| &= \left| L_{(t-1,d_{t-1}(v))}(v) \bigcap \left(\cup_{u\in\{w\mid w\in A_{t-1}(v),d_{t-1}(w)\neq\mu\}}L_{(t-1,d_{t-1}(u))}(u)\right) \right| \\
&\leq \sum_{u\in\{w\mid w\in A_{t-1}(v),d_{t-1}(w)\neq\mu\}}|L_{(t-1,d_{t-1}(v))}(v)\cap L_{(t-1,d_{t-1}(u))}(u)|.
\end{align*}
Recall that $L_{(t-1,d_{t-1}(v))}(v)=\{P_{(t-1,v,d_{t-1}(v))}(x)+x\cdot\mu\mid x\in[\mu]\}$. For an element to be in both $L_{(t-1,d_{t-1}(v))}(v)$ and $L_{(t-1,d_{t-1}(u))}(u)$, it must be the case that $P_{(t-1,v,d_{t-1}(u))}(x)=P_{(t-1,v,d_{t-1}(v))}(x)$ for some $x\in[\mu]$. Recall that $P_{(t-1,v,d_{t-1}(v))}(x) = (\lfloor b_{t-1}(v)/\tau \rfloor \cdot x^2 + (b_{t-1}(v)\bmod\tau)\cdot x + d_{t-1}(v))\bmod\mu$ is a polynomial of degree (at most) two defined over finite field $GF(\mu)$. Since $u\in A_{t-1}(v)$, it must be the case that $b_{t-1}(u)\neq b_{t-1}(v)$, implying $P_{(t-1,v,d_{t-1}(v))}(x)$ and $P_{(t-1,v,d_{t-1}(u))}(u)$ are two distinct polynomials of degree (at most) two. Hence, there are at most two choices of $x\in[\mu]$ satisfying $P_{(t-1,v,d_{t-1}(u))}(x)=P_{(t-1,v,d_{t-1}(v))}(x)$, implying $|L_{(t-1,d_{t-1}(v))}(v)\cap L_{(t-1,d_{t-1}(u))}(u)|\leq 2$. As a result, we conclude:
$$|\mathfrak{A}\cap\mathfrak{B}_2| \leq 2\cdot|\{u\mid u\in A_{t-1}(v),d_{t-1}(u)\neq\mu\}| \leq 2\cdot|A_{t-1}(v)|,$$
which leads to the following upper bound on $|\mathfrak{A}\cap(\mathfrak{B}_1\cup\mathfrak{B}_2)|$:
$$|\mathfrak{A}\cap(\mathfrak{B}_1\cup\mathfrak{B}_2)| \leq \Delta/\mu + \left| L_{\left(t^+_v-1,d_{t^+_v}(v)\right)}(v) \cap \left(L_{t-1}(v) \setminus L_{t^+_v-1}(v)\right) \right| + 2\cdot|A_{t-1}(v)|.$$

Observe that as time proceeds from round $t^+_v$ to round $t$, more and more neighbors of $v$ may have done the transition-out stage and start running the third phase. To bound the above expression, consider a neighbor $u$ of $v$ with $\phi_{t^+_v-1}(u)\in I_2$.

It cannot be the case that $a_{t^+_v-1}(v)>a_{t^+_v-1}(u)$, since by the definition of $t^+_v$ we have $d_{t^+_v}(v)\neq\mu$, yet by \Cref{alg:phase2-stage3} updating $d(v)$ in round $t^+_v$ requires $a_{t^+_v-1}(v)\leq a_{t^+_v-1}(u)$. If $a_{t^+_v-1}(v)<a_{t^+_v-1}(u)$, then by \Cref{alg:phase2-stage3}, vertex $u$ will not start the transition to phase three until the transition of vertex $v$ is done, thus the behavior of $u$ will not change the above upper bound of $|\mathfrak{A}\cap(\mathfrak{B}_1\cup\mathfrak{B}_2)|$ in rounds $[t^+_v,t]$, as in round $t$ we still have $\phi_t(v)\in I_2$ (meaning by the end of round $t$ vertex $v$ has not completed the transition-out stage).

Now let us focus on the case $a_{t^+_v-1}(v)=a_{t^+_v-1}(u)$. If $c_{t^+_v-1}(v)<c_{t^+_v-1}(u)$, then $u\notin A_{t^+_v-1}(v)$. Notice that by the definition of $A_{t-1}(v)$ we have $A_{t^+_v-1}(v)=A_{t-1}(v)$, thus the behavior of $u$ will not change $|A_{t-1}(v)|$ in rounds $[t^+_v,t]$. On the other hand, if indeed $u$ obtains its phase three color in some round in $[t^+_v,t]$, then by Line \ref{alg-line:phase2-stage3-k-rule} of \Cref{alg:phase2-stage3}, when $u$ chooses its color, it will avoid all colors that might be used by $v$. This means the phase three color used by $u$ will not appear in $L_{\left(t^+_v-1,d_{t^+_v}(v)\right)}(v)$, implying the behavior of $u$ will not change $| L_{\left(t^+_v-1,d_{t^+_v}(v)\right)}(v) \cap (L_{t-1}(v) \setminus L_{t^+_v-1}(v))|$ in rounds $[t^+_v,t]$. By now, we conclude that if $a_{t^+_v-1}(v)=a_{t^+_v-1}(u)$ and $c_{t^+_v-1}(v)<c_{t^+_v-1}(u)$, then the behavior of $u$ will not change the upper bound of $|\mathfrak{A}\cap(\mathfrak{B}_1\cup\mathfrak{B}_2)|$ in rounds $[t^+_v,t]$.

As a result, the only scenario that the behavior of $u$ might change the upper bound of $|\mathfrak{A}\cap(\mathfrak{B}_1\cup\mathfrak{B}_2)|$ is when $a_{t^+_v-1}(v)=a_{t^+_v-1}(u)$ and $c_{t^+_v-1}(v)\geq c_{t^+_v-1}(u)$. That is, $u\in A_{t^+_v-1}(v)$. For each such vertex $u$, observe that as it transits to the third phase in some round, $|A_t(v)|$ decreases by one, while $|L_t(v)|$ increases by one. Recall the expression of the upper bound of $|\mathfrak{A}\cap(\mathfrak{B}_1\cup\mathfrak{B}_2)|$, the above discussion implies, as vertex $u$ transits to phase three, the value of the upper bound decreases. As a result, we conclude:
\begin{align*}
|\mathfrak{A}\cap(\mathfrak{B}_1\cup\mathfrak{B}_2)| &\leq \Delta/\mu + \left| L_{\left(t^+_v-1,d_{t^+_v}(v)\right)}(v) \cap \left(L_{t-1}(v) \setminus L_{t^+_v-1}(v)\right) \right| + 2\cdot|A_{t-1}(v)|\\
&\leq \Delta/\mu + \left| L_{\left(t^+_v-1,d_{t^+_v}(v)\right)}(v) \cap \left(L_{t^+_v-1}(v) \setminus L_{t^+_v-1}(v)\right) \right| + 2\cdot \left|A_{t^+_v-1}(v)\right|\\
&\leq \Delta/\mu + 0 + 2\cdot\left(2\cdot\Delta^{1/4}\right),
\end{align*}
where the last inequality is due to \Cref{lemma:phase2-stage2-bounded-arboricity}.

This completes the proof of the claim.
\end{proof}

We are now ready to bound the time cost of the transition-out stage (i.e., \Cref{lemma:phase2-stage3-time-cost}).

\begin{proof}[Proof of \Cref{lemma:phase2-stage3-time-cost}]
For each vertex $v$, let $t^*_v$ be the smallest round number such that $\phi_{t^*_v}(v)\in I_2$ and $a_{t^*_v}(v)\in[\lambda]$ are both satisfied, and let $t^+_v$ be the smallest round number such that $\phi_{t^+_v}(v)\in I_2$ and $d_{t^*_v}(v)\neq\mu$ are both satisfied. By algorithm description, $t^*_v < t^+_v < t^{\#}_v$.

To prove the lemma, it suffices to prove the following stronger stronger claim: for every vertex $v\in V$, it holds that $t^+_v\leq r^*+1+\lambda+2(a_{t^*_v}(v)+1)$, and that $t^{\#}_v\leq r^*+2+\lambda+2(a_{t^*_v}(v)+1)$.

To prove the claim, we do an induction on the value of $a_{t^*_v}(v)\in[\lambda]$.
First consider the base case, fix a vertex $v$ with the smallest $a_{t^*_v}$ value.
Due to \Cref{lemma:phase2-stage2-time-complexity}, for vertex $v$, as well as every vertex $u\in N(v)$, we have $t^*_v\leq r^*+2+\lambda$ and $t^*_u\leq r^*+2+\lambda$. Thus in round $r^*+3+\lambda$, if $\phi_{r^*+2+\lambda}(v)\in I_2$ and $d_{r^*+2+\lambda}(v)=\mu$, then for vertex $v$, the ``if'' condition in Line \ref{alg-line:phase2-stage3-if-cond-1} and Line \ref{alg-line:phase2-stage3-if-cond-2} of \Cref{alg:phase2-stage3} will both be satisfied. Moreover, the ``if'' condition in Line \ref{alg-line:phase2-stage3-if-cond-3} of \Cref{alg:phase2-stage3} will also be satisfied in this round. As a result, by the end of round $r^*+3+\lambda$, if $\phi_{r^*+3+\lambda}(v)\in I_2$, it must be the case that $d_{r^*+3+\lambda}(v)\neq\mu$. In other words, $t^+_v\leq r^*+3+\lambda$. Apply the same argument for every vertex $u\in A_{r^*+2+\lambda}(v)\supseteq A_{r^*+3+\lambda}(v)$, it holds that $t^+_u\leq r^*+3+\lambda$. Therefore, in round $r^*+4+\lambda$, if $\phi_{r^*+3+\lambda}(v)\in I_2$, then for vertex $v$, the ``if'' condition in Line \ref{alg-line:phase2-stage3-if-cond-1} and Line \ref{alg-line:phase2-stage3-if-cond-2} of \Cref{alg:phase2-stage3} will both be satisfied. Moreover, the ``if'' condition in Line \ref{alg-line:phase2-stage3-if-cond-4} of \Cref{alg:phase2-stage3} will also be satisfied in this round. As a result, by \Cref{claim:phase2-stage3-bounded-k}, by the end of round $r^*+4+\lambda$, vertex $v$ must have obtained a color in $I_3$. This completes the proof of the base case.

Assume our claim holds for all vertices $u$ with $a_{t^*_u}(u)\leq i\in[\mu-1]$, consider a vertex $v$ with $a_{t^*_v}(v)=i+1\in[\mu]$. The proof for the inductive step generally follow the same path as in the base case. Specifically, for every vertex $u\in N(v)$ with $a_{t^*_u}(u)\leq i$, by the induction hypothesis, it must be the case that $\phi_{r^*+2+\lambda+2(i+1)}(u)\in I_3$. Thus in round $r^*+3+\lambda+2(i+1)$, every vertex $u\in N(v)$ with $\phi_{r^*+2+\lambda+2(i+1)}(u)\in I_2$ must have $a_{t^*_u}\geq i+1$. Hence, in round $r^*+3+\lambda+2(i+1)$, if $\phi_{r^*+2+\lambda+2(i+1)}(v)\in I_2$ and $d_{r^*+2+\lambda+2(i+1)}(v)=\mu$, then for vertex $v$, the ``if'' condition in Line \ref{alg-line:phase2-stage3-if-cond-1} and Line \ref{alg-line:phase2-stage3-if-cond-2} of \Cref{alg:phase2-stage3} will both be satisfied. Moreover, the ``if'' condition in Line \ref{alg-line:phase2-stage3-if-cond-3} of \Cref{alg:phase2-stage3} will also be satisfied in this round. As a result, by the end of round $r^*+3+\lambda+2(i+1)$, if $\phi_{r^*+3+\lambda+2(i+1)}(v)\in I_2$, it must be the case that $d_{r^*+3+\lambda+2(i+1)}(v)\neq\mu$. In other words, $t^+_v \leq r^*+3+\lambda+2(i+1) = r^*+1+\lambda+2((i+1)+1)$. Apply the same argument for every vertex $u\in A_{r^*+2+\lambda+2(i+1)}(v)\supseteq A_{r^*+3+\lambda+2(i+1)}(v)$, it holds that $t^+_u \leq r^*+3+\lambda+2(i+1)$. Therefore, in round $r^*+4+\lambda+2(i+1)$, if $\phi_{r^*+3+\lambda+2(i+1)}(v)\in I_2$, then for vertex $v$, the ``if'' condition in Line \ref{alg-line:phase2-stage3-if-cond-1} and Line \ref{alg-line:phase2-stage3-if-cond-2} of \Cref{alg:phase2-stage3} will both be satisfied. Moreover, the ``if'' condition in Line \ref{alg-line:phase2-stage3-if-cond-4} of \Cref{alg:phase2-stage3} will also be satisfied in this round. As a result, by \Cref{claim:phase2-stage3-bounded-k}, by the end of round $r^*+4+\lambda+2(i+1)=r^*+2+\lambda+2((i+1)+1)$, vertex $v$ must have obtained a color in $I_3$. This completes the proof of the inductive step.
\end{proof}

We conclude by showing the correctness of the transition-out stage. In particular, when a vertex $v$ finishes the transition in round $t^{\#}_v$ and obtained a color in $I_3$, that color $\phi_{t^{\#}_v}(v)$ will not conflict with any neighbor $u\in N(v)$ that also have its color $\phi_{t^{\#}_v}(u)$ in $I_3$. In other words, \Cref{lemma:phase2-stage3-proper-color} is true.

\begin{proof}[Proof of \Cref{lemma:phase2-stage3-proper-color}]
By \Cref{alg:phase2-stage3}, vertex $v$ sets $\phi_{t^{\#}_v}(v)$ as the minimum elements in:
$$L_{t^{\#}_v-1,d_{t^{\#}_v-1}(v)}(v) \setminus \left( L_{t^{\#}_v-1}(v) \bigcup \left(\cup_{u\in\left\{w\mid w\in A_{t^{\#}_v-1}(v),d_{t^{\#}_v-1}(w)\neq\mu\right\}} L_{t^{\#}_v-1,d_{t^{\#}_v-1}(u)}(u)\right) \right).$$
Consider a neighbor $u\in N(v)$. If $\phi_{t^{\#}_v-1}(u)\in I_1$, then obviously $\phi_{t^{\#}_v}(u)\notin I_3$ as the transition-stage of vertex $u$ takes at least two rounds, thus $\phi_{t^{\#}_v}(u)$ will not conflict with $\phi_{t^{\#}_v}(v)$. If $\phi_{t^{\#}_v-1}(u)\in I_3$, then by \Cref{alg:phase3}, we have $\phi_{t^{\#}_v}(u)=\phi_{t^{\#}_v-1}(u)$. Moreover, when $v$ chooses $\phi_{t^{\#}_v}(v)$ it will not consider $\phi_{t^{\#}_v-1}(u)$ as $\phi_{t^{\#}_v-1}(u)\in L_{t^{\#}_v-1}(v)$. Hence, when $\phi_{t^{\#}_v-1}(u)\in I_3$, we also have $\phi_{t^{\#}_v}(u)\neq\phi_{t^{\#}_v}(v)$. Lastly, if $\phi_{t^{\#}_v-1}(u)\in I_2$, then there are four scenarios:
\begin{itemize}
	\item \textsc{Scenario I:} Vertex $u$ has $a_{t^{\#}_v-1}(u)<a_{t^{\#}_v-1}(v)$. This scenario cannot happen, since by \Cref{alg:phase2-stage3}, vertex $v$ will only set $d(v)$ to a value other than $\mu$ after all its neighbors with smaller $a$ values have done the transition to the third phase. Therefore, if $a_{t^{\#}_v-1}(u)<a_{t^{\#}_v-1}(v)$, then in round $t^{\#}_v$ vertex $v$ must have already started the third phase, a contradiction.
	
	\item \textsc{Scenario II:} Vertex $u$ has $a_{t^{\#}_v-1}(u)>a_{t^{\#}_v-1}(v)$. By \Cref{alg:phase2-stage3}, vertex $u$ cannot complete the transition-out stage in round $t^{\#}_v$, as $a_{t^{\#}_v-1}(u)>a_{t^{\#}_v-1}(v)$. Therefore, $\phi_{t^{\#}_v}(u)\in I_2$, implying it will not conflict with the color chosen by vertex $v$.
	
	\item \textsc{Scenario III:} Vertex $u$ has $a_{t^{\#}_v-1}(u)=a_{t^{\#}_v-1}(v)$ and $u\in A_{t^{\#}_v-1}(v)$. In such scenario, if indeed $u$ finishes the transition-out stage and obtains a color in $I_3$ by the end of round $t^{\#}_v$, then this color $\phi_{t^{\#}_v}(u)\in L_{t^{\#}_v-1,d_{t^{\#}_v-1}(u)}(u)$. On the other hand, by \Cref{alg:phase2-stage3}, the initial phase three color chosen by vertex $v$, which is $\phi_{t^{\#}_v}(v)$, will not appear in $L_{t^{\#}_v-1,d_{t^{\#}_v-1}(u)}(u)$. Hence, if indeed $u$ finishes the transition-out stage and obtains a color in $I_3$ by the end of round $t^{\#}_v$, then $\phi_{t^{\#}_v}(u)\neq\phi_{t^{\#}_v}(v)$. Otherwise, if $\phi_{t^{\#}_v}(u)\in I_2$, then obviously $\phi_{t^{\#}_v}(u)\neq\phi_{t^{\#}_v}(v)$, as $\phi_{t^{\#}_v}(v)\in I_3$ by the definition of $t^{\#}_v$.
	
	\item \textsc{Scenario IV:} Vertex $u$ has $a_{t^{\#}_v-1}(u)=a_{t^{\#}_v-1}(v)$ and $u\notin A_{t^{\#}_v-1}(v)$. In such scenario, we have $v\in A_{t^{\#}_v-1}(u)$. By an analysis similar to \textsc{Scenario III} (but from the perspective of vertex $u$), we conclude that $\phi_{t^{\#}_v}(u)\neq\phi_{t^{\#}_v}(v)$.
\end{itemize}
This completes the proof of the lemma.
\end{proof}

\subsection{Proof of the main lemma of the quadratic reduction phase}\label{subsec:phase2-proof-of-main-lemma}

\begin{proof}[Proof of \Cref{lemma:phase2-property-simple}]
By \Cref{lemma:phase2-stage3-time-cost}, we know by the end of round $r^*+2+3\lambda$, every vertex have completed the quadratic reduction phase and obtained a color in $I_3$. By the definition of $I_3$, we know $\phi_{r^*+2+3\lambda}$ is a $(\Delta+O(\Delta^{3/4}\log\Delta))$-coloring.

Next, we prove for every round $t\in[r^*+1,r^*+2+3\lambda]$, the coloring $\phi_{t}$ is proper.

By \Cref{lemma:phase2-stage1-property}, we know by the end of round $r^*+1$, the coloring $\phi_{r^*+1}$ is proper. From round $r^*+2$, every vertex starts running the core stage. For every vertex $v$, by \Cref{lemma:phase2-stage2-proper-coloring}, for every round $t\in[r^*+2,t^{\#}_v-1]$, its color $\phi_t(v)$ will not conflict with any of its neighbor. Then, in round $t^{\#}_v$, when vertex $v$ chooses its phase three color, by \Cref{lemma:phase2-stage3-proper-color}, we know $\phi_{t^{\#}_v}(v)$ will also not conflict with any of its neighbor. At this point, we have proved, for every vertex $v$, for every round $t\in[r^*+1,t^{\#}_v]$, its color $\phi_t(v)$ will not conflict with any of its neighbor.

Now consider a round $t\geq t^{\#}_v+1$, and a neighbor $u\in N(v)$. If $\phi_{t-1}(u)\in I_1$, then obviously $\phi_t(u)\notin I_3$ as $u$ can not map a color from $I_1$ to $I_3$ in one round, implying $\phi_t(v)\neq\phi_t(u)$. If $\phi_{t-1}(u)\in I_2$ and $\phi_{t}(u)\in I_2$, then trivially $\phi_t(v)\neq\phi_t(u)$. If $\phi_{t-1}(u)\in I_2$ but $\phi_{t}(u)\in I_3$, then apply \Cref{lemma:phase2-stage3-proper-color} from the perspective of $u$, we still have $\phi_t(v)\neq\phi_t(u)$. If $\phi_{t-1}(u)\in I_3$, then by the standard reduction phase algorithm, in round $t$, at most one of $u,v$ will change its color, and the updated color of that vertex will not conflict with the other vertex. Once again, we have $\phi_t(v)\neq\phi_t(u)$.

By now, we can conclude, for every round $t\in[r^*+1,r^*+2+3\lambda]$, the coloring $\phi_{t}$ is proper.
\end{proof}

\section{Omitted Proofs of The Standard Reduction Phase of the Locally-iterative Algorithm}\label{sec-app:standard-reduction-phase}

\begin{proof}[Proof of \Cref{lemma:phase3-time-cost}]
By \Cref{lemma:phase2-property-simple}, by the end of round $r^*+2+3\lambda$, every vertex has a color in $I_3$, and will run the standard reduction phase algorithm in round $r^*+3+3\lambda$. Starting from round $r^*+3+3\lambda$, in each round, every vertex with the maximum color value in its one-hop neighborhood will change its color to another one in $[\Delta+1]$. That is, starting from round $r^*+3+3\lambda$, in each round, the maximum color value will be reduced by at least one. Recall that $I_3=[0,\ell_3)$ where $\ell_3=\Delta+(2\sqrt{m_3}+1)\cdot\mu$. Therefore, by the end of round $r^*+2+3\lambda+(\ell_3-(\Delta+1))$, every vertex has its color in $[\Delta+1]$.
\end{proof}

\begin{proof}[Proof of \Cref{lemma:phase3-proper-color}]
Consider two neighboring vertices $u$ and $v$, we prove the lemma by an induction on $t$. In the base case in which $t=r^*+3+3\lambda$, by \Cref{lemma:phase2-property-simple}, we have $\phi_{r^*+2+3\lambda}(v)\in I_3$, $\phi_{r^*+2+3\lambda}(u)\in I_3$, and $\phi_{r^*+2+3\lambda}(v)\neq\phi_{r^*+2+3\lambda}(u)$. By the standard reduction phase algorithm, in round $r^*+3+3\lambda$, at most one of $u,v$ will change its color, and the updated color of that vertex will not conflict with the other vertex. Hence, we have $\phi_{r^*+3+3\lambda}(u)\neq\phi_{r^*+3+3\lambda}(v)$. This completes the proof of the base case.
The inductive step can be proved by a similar argument as in the base case.
\end{proof}

\section{Omitted Details and Proofs of The Self-stabilizing Coloring Algorithm}\label{sec-app:alg-stab}

In this section, we give the complete and detailed description of the self-stabilizing algorithm, along with the proofs for the key lemmas stated in the main body of the paper.

\subsection{The Linial phase and the transition-in stage of the quadratic reduction phase}

At the beginning of a round $t$, if a vertex $v$ has its color $\phi_{t-1}(v)$ not in interval $I_2\cup I_3$, then it should run either the Linial phase or the transition-in stage of the quadratic reduction phase. Nonetheless, before proceeding, it will do error-checking to see if any of the following conditions is satisfied:
\begin{itemize}
	\item Its color collide with some neighbor.
	\item Its color $\phi_{t-1}(v)$ is not in $\left(\bigcup_{i=1}^{r^*}I_1^{(i)}\right)\cup I_2\cup I_3$ (which means $v$ should be running the first iteration of the Linial phase), but that color is not $\ell_3+\ell_2+\sum_{i=1}^{r^*}n_i+id(v)$.
\end{itemize}
If any of these conditions is satisfied, then vertex $v$ treats itself in an improper state and resets its color to $\ell_3+\ell_2+\sum_{i=1}^{r^*}n_i+id(v)$. That is, it sets $\phi_t(v)=\ell_3+\ell_2+\sum_{i=1}^{r^*}n_i+id(v)$.

Otherwise, vertex $v$ first determines which interval $I_1^{(t')}$ it is in.

If $0\leq t'< r^*$, then it computes a $\Delta$-cover-free set family $\mathcal{F}_{t'}$ as in the locally-iterative algorithm, and sets its new color as the smallest number in $S_{t'}^{(\phi_{t-1}(v))}$, excluding all elements of $S_{t'}^{(\phi_{t-1}(v))}$ for all $v$'s neighbors $u\in N(v)$ satisfying $\phi_{t-1}(u)\in I_1^{(t')}$.
Due to the $\Delta$-cover freeness of set family $\mathcal{F}_{t'}$, such element must exist.

If $t'=r^*$, then vertex $v$ transforms its color from interval $I_1$ to $I_2$, effectively running the transition-in stage of the quadratic reduction phase.

To do the transformation, vertex $v$ first constructs a $\Delta$-union-$(\Delta^{1/4}+1)$-cover-free set family $\mathcal{F}_a$. Let $q$ be a prime such that $\frac{\Delta+1}{\Delta^{1/4}+1}\log(n_{r^*}) < q\leq 2\cdot \frac{\Delta+1}{\Delta^{1/4}+1}\log(n_{r^*})$, set family $\mathcal{F}_a$ is of size $n_{r^*}$ with $[q^2]\subseteq [m_1]$ as its ground set. More specifically, for every integer $i\in[n_{r^*}]$, we associate a unique polynomial $P_i$ of degree (at most) $\log(n_{r^*})$ over finite field $GF(q)$ to it. Then $\mathcal{F}_a\triangleq \{S_a^{(\ell_3+\ell_2)},\cdots,S_a^{(\ell_3+\ell_2+n_{r^*}-1)}\}$, where $S_a^{(i)}=\{x\cdot q + P_i(x)\mid x\in [q]\}$ for every $i\in[\ell_3+\ell_2,\ell_3+\ell_2+n_{r^*})$. Since the degree of the polynomials is (at most) $\log(n^{r^*})$, the intersection of any two sets $S_a^{(\cdot)}$ contains at most $\log(n_{r^*})$ elements. Recall that every set contains $q>\frac{\Delta+1}{\Delta^{1/4}+1}\log(n_{r^*})$ elements. To cover a set $S_a^{(i)}$ for any $i\in [n_{r^*}]$, we need at least $\Delta^{1/4}+1$ other set in $\mathcal{F}_a$. Thus $\mathcal{F}_a$ is a $\Delta$-union-$(\Delta^{1/4}+1)$-cover-free set family.

Then, define two sets $N_1'(v,x)$ and $N_2'(v,x)$:
\begin{align*}
N_1'(v,x) &\triangleq \{u\mid u\in N(v), \phi_{t-1}(u)\in I_1^{(t')}, x\in S_a^{(\phi_{t-1}(u))}\},\\
N_2'(v,x) &\triangleq \{ u\mid u\in N(v), \phi_{t-1}(u) \in I_2, x+\lambda=\hat{a}_{t-1}(u)\cdot \lambda + (\hat{a}_{t-1}(u)+ \tilde{a}_{t-1}(u))\bmod \lambda  \}.
\end{align*}

Let $\hat{x}$ be the smallest element in $S_a^{(\phi_{t-1}(v)-\ell_2-\ell_3)}$ satisfying $|N_1'(v,\hat{x})\cup N_2'(v,\hat{x})|\leq\Delta^{1/4}$, vertex $v$ then assigns $a(v)=\hat{x}+\lambda \in[\lambda^2]$.
Later in the analysis (particularly, in the proof of \Cref{lemma:self-stabilizing-correctness-I2-large-a}), via a counting argument, we will show such $\hat{x}$ must exist when there are no errors in the system.
Then, vertex $v$ computes $b(v)$ using the same method as in the locally-iterative algorithm. In particular, vertex $v$ sets value $b(v)$ as the smallest elements in the following set:
$$S_b^{(\phi_{t-1}(v))} \setminus \left( \left(\cup_{u\in N_1'(v,\hat{x})}S_b^{(\phi_{t-1}(u))}\right) \bigcup \{b_{t-1}(u)\mid u\in N_2'(v,\hat{x})\} \right).$$
Again, later in the analysis (particularly, in the proof of \Cref{lemma:self-stabilizing-correctness-I2-large-a}), we will argue the above set is non-empty when there are no errors in the system.

Lastly, vertex $v$ sets $c(v)=0$ and $d(v)=\mu$. It also sets $T_v[u]$ to $1$ if its neighbor $u$ is in the set $N_1'(v,\hat{x})\cup N_2'(x,\hat{x})$, otherwise $T_v[u]=0$. We recall that the orientation of edge $(u,v)$ is determined by $T_v[u]$ and $T_u[v]$: vertex $v$ points to vertex $u$ if and only if $T_v[u]=1$.

The pseudocode of the Linial phase and the transition-in stage of the quadratic reduction phase in the self-stabilizing setting is given in \Cref{alg:self-stabilized-phase1-to-phase2-stage1}.

\begin{algorithm}[t!]
\caption{The Linial phase and the transition-in stage at vertex $v$ in round $t$}\label{alg:self-stabilized-phase1-to-phase2-stage1}
\begin{algorithmic}[1]
\State Send $\langle \phi_{t-1}(v), T_v[u] \rangle$ to neighbor $u\in N(v)$, where $T_v[u]$ is the entry corresponds to $u$ in $T_v$.
\If {($\phi_{t-1}(v)\notin I_2 \cup I_3$)}
	\If{(($\exists u\in N(v), \phi_{t-1}(v) = \phi_{t-1}(u)$) \Comment{Error-checking.}
		\Statex \hspace{\algorithmicindent}\quad \textbf{or} ($\phi_{t-1}(v)\geq \ell_3+\ell_2+\sum_{i=1}^{r^*}n_i$ \textbf{and} $\phi_{t-1}(v)\neq \ell_3+\ell_2+\sum_{i=1}^{r^*}n_i+id(v)$))}
		\State $\phi_t(v)\gets \ell_3+\ell_2+\sum_1^{r^*} n_i+id(v)$.
	\Else
		\State Determine the interval $I_1^{(t^{\prime})}$ that  $\phi_{t-1}(v)$ is in.
		\If{($0 \leq t'< r^*$)} \Comment{Run the Linial phase.}
			\State $\phi_t(v)\gets\min S_{t'}^{(\phi_{t-1}(v))}\setminus\bigcup_{u\in N(v) \text{ and }\phi_{t-1}(u) \in I_1^{(t^{\prime})}} S_{t'}^{(\phi_{t-1}(u))}$.
		\Else \Comment{Run the transition-in stage.}
			\For {(every element $x\in S_a^{(\phi_{t-1}(v))}$)}
				\State \begin{small}$N_1'(v,x) \gets \left\{u \mid u\in N(v),\phi_{t-1}(u) \in I_1^{(t')}, x\in S_a^{(\phi_{t-1}(u))} \right\}$.\end{small}
				\State \begin{small}$N_2'(v,x) \gets \{ u\mid u\in N(v), \phi_{t-1}(u) \in I_2, x+\lambda=\hat{a}_{t-1}(u)\cdot \lambda + (\hat{a}_{t-1}(u)+ \tilde{a}_{t-1}(u))\bmod \lambda  \}$.\end{small}
			\EndFor
		\State $\hat{x}\gets\min\left\{ x \mid x\in S_a^{(\phi_{t-1}(v))}, \left|N_1'(v,x)\cup N_2'(v,x) \right|\leq\Delta^{1/4} \right\}$.
		\State $a_{t}(v)\gets\hat{x}+\lambda$.
		\State $b_t(v) \gets \min S_b^{(\phi_{t-1}(v))} \setminus \left( \left(\cup_{u\in N_1'(v,\hat{x})}S_b^{(\phi_{t-1}(u))}\right) \bigcup \{b_{t-1}(u)\mid u\in N_2'(v,\hat{x})\} \right)$.
		\State $c_{t}(v)\gets 0$, $d_{t}(v)\gets\mu$.
		\State $\phi_{t}(v) \gets \ell_3 + \langle a_{t}(v), b_{t}(v), c_{t}(v), d_{t}(v) \rangle$.
		\State Initialize $T_v[u]$ to $0$ for all $u\in N(v)$.
		\For{every element $u\in N_1'(v,\hat{x})\cup N_2'(v,\hat{x})$}
			\State $T_v[u]\gets 1$.
		\EndFor
		\EndIf
	\EndIf
\EndIf
\end{algorithmic}
\end{algorithm}

\subsection{The core stage of the quadratic reduction phase}

\begin{algorithm}[t!]
\caption{The core stage of the quadratic reduction phase at vertex $v$ in round $t$}\label{alg:self-stabilizing-phase2-stage2}
\begin{algorithmic}[1]
\State Send $\langle \phi_{t-1}(v), T_v[u] \rangle$ to neighbor $u\in N(v)$, where $T_v[u]$ is the entry corresponds to $u$ in $T_v$.
\If{($\phi_{t-1}(v)\in I_2$ \textbf{and} $a_{t-1}(v)\geq\lambda$)}
	\If{(($\exists u\in N(v), \phi_{t-1}(u)\in I_2, a_{t-1}(v)=a_{t-1}(u), b_{t-1}(u)=b_{t-1}(v)$) \Comment{Error-checking.}
	\Statex \hspace{\algorithmicindent}\quad \textbf{or} ($\exists u\in N(v), \phi_{t-1}(u)\in I_2,  a_{t-1}(v)=a_{t-1}(u), T_u[v]+T_v[u]=0$)
	\Statex \hspace{\algorithmicindent}\quad \textbf{or} ($|\{u\mid u\in N(v), T_v[u]=1\}|>\Delta^{1/4}$)
	\Statex \hspace{\algorithmicindent}\quad \textbf{or} ($\exists u\in N(v)\cup \{v\}, \phi_{t-1}(u)\in I_2, a_{t-1}(u)\geq\lambda, b_{t-1}(u)\notin [m_2]$))}
		\State $\phi_t(v)\gets \ell_3+\ell_2+\sum_1^{r^*}{n_i}+id(v)$.
	\Else \Comment{Run the core stage.}
		\State $M_t(v)\gets\left\{u\mid u\in N(v),\phi_{t-1}(u)\in I_2,\hat{a}_{t-1}(u)\neq \hat{a}_{t-1}(v),\tilde{a}_{t-1}(u)= \tilde{a}_{t-1}(v)\right\}$. \label{alg-line:self-stabilizing-phase2-stage2-start}
		\State $\overline{M}_t(v)\gets\left\{u\mid u\in N(v),\phi_{t-1}(u)\in I_2,\hat{a}_{t-1}(u)=\hat{a}_{t-1}(v),\tilde{a}_{t-1}(u)=\tilde{a}_{t-1}(v), T_v[u]=1\right\}$.
		\If{($|M_t(v)|\leq\Delta^{1/4}$)}
			\State $M'_{t}(v) \gets \left\{ u\mid u\in N(v),\phi_{t-1}(u)\in I_2,\hat{a}_{t-1}(u)=0,\tilde{a}_{t-1}(u)=\tilde{a}_{t-1}(v) \right\}$.
			\State $\overline{M}'_{t}(v) \gets (M_t(v)\cup\overline{M_t}(v))\setminus M_t'(v)$.
			\State $a_t(v) \gets \tilde{a}_{t-1}(v)$.
			\State $b_t(v) \gets \min S_c^{(a_{t-1}(v)\cdot m_2 + b_{t-1}(v))}\setminus \left( \{b_{t-1}(u)\mid u\in M'_{t}(v)\} \cup \bigcup_{u\in \overline{M}'_{t}(v)} S_c^{(a_{t-1}(u)\cdot m_2+b_{t-1}(u))} \right)$.
			\State $\phi_{t}(v) \gets \ell_3 + \langle a_t(v), b_t(v), c_t(v), d_t(v) \rangle$.
			\State Initialize $T_v[u]$ to $0$ for all $u\in N(v)$.
			\For{every elements in $M_t(v)\cup \overline{M_t}(v)$}
				\State $T_v[u]\gets 1$.
			\EndFor
		\Else
			\State $a_t(v) \gets \hat{a}_{t-1}(v)\cdot\lambda + ((\hat{a}_{t-1}(v)+\tilde{a}_{t-1}(v))\bmod\lambda)$.\label{alg-line:intermediate-a-update-rule}
			\State $\phi_{t}(v) \gets \ell_3 + \langle a_t(v), b_t(v), c_t(v), d_t(v) \rangle$.
		\EndIf \label{alg-line:self-stabilizing-phase2-stage2-end}
	\EndIf
\EndIf
\end{algorithmic}
\end{algorithm}

Recall that in the locally-iterative coloring algorithm, during the quadratic reduction phase, for a vertex $v$ with $\phi(v)\in I_2$, if $a(v)$ is already in $[\lambda]$, then its core stage is done, and may proceed to the transition-out stage. This is still the case in the self-stabilizing settings: a vertex runs the core stage only if it finds $a(v)\geq\lambda$. (See \Cref{alg:self-stabilizing-phase2-stage2} for the pseudocode.) Moreover, in case $a(v)\geq\lambda$, before proceeding, vertex $v$ will do error-checking to see if any of the following conditions is satisfied:
\begin{itemize}
	\item There exists a neighbor $u$ of $v$ such that $a(v)=a(u)$ and $b(v)=b(u)$, effectively implying $u$ and $v$ have identical color.
	\item There exists a neighbor $u$ of $v$ such that $a(v)=a(u)$ yet $T_v[u]+T_u[v]=0$, implying that the orientation of edge $(u,v)$ is still undetermined when $a(v)=a(u)$.
	\item The number of neighbors $u\in N(v)$ with $T_v[u]=1$ is larger than $\Delta^{1/4}$, violating the bounded arboricity assumption during the core stage.
	\item There exists a vertex $u\in N(v)\cup \{v\}$ has its color in $I_2$ and $a(u)\geq\lambda$, yet $b(u)\geq m_2$, violating the range of $b$ values during the core stage.
\end{itemize}
If any of these conditions is satisfied, then vertex $v$ treats itself in an improper state and resets its color to $\ell_3+\ell_2+\sum_{i=1}^{r^*}n_i+id(v)$. Otherwise, it executes Line \ref{alg-line:self-stabilizing-phase2-stage2-start} to Line \ref{alg-line:self-stabilizing-phase2-stage2-end} of \Cref{alg:self-stabilizing-phase2-stage2} to try to reduce its $a$ value from $[\lambda,\lambda^2)$ to $[0,\lambda)$.

The procedure we use in the self-stabilizing settings to transform a $\Delta^{1/4}$-arbdefective $\lambda^2$-coloring to a $(2\cdot\Delta^{1/4})$-arbdefective $\lambda$-coloring is almost identical to the one we use in the locally-iterative settings. The only difference is that we alter the definition of $\overline{M}_t(v)$ by adding an extra constraint $T_v[u]=1$, which means the orientation of edge $(u,v)$ is $v$ pointing to $u$. This ensures that $|\overline{M}_t(v)|$ is still bounded by $\Delta^{1/4}$.

Once the reduction occurs in some round $t$, vertex $v$ obtains an $a_t(v)\in[\lambda]$, and updates $b_t(v)$ to differentiate itself from the neighbors that may have colliding $a$ value.
By an analysis similar to the locally-iterative setting, such $b_t(v)$ must exist.
It also sets $T_v[u]=1$ if and only if $u\in M_t(v)\cup \overline{M}_t(v)$, recording the orientation of such edges.

\subsection{The transition-out stage of the quadratic reduction phase}

\begin{algorithm}[t!]
\caption{The transition-out stage of the quadratic reduction phase at vertex $v$ in round $t$}\label{alg:self-stabilizing-phase2-stage3}
\begin{algorithmic}[1]
\State Send $\langle \phi_{t-1}(v), T_v[u] \rangle$ to neighbor $u\in N(v)$, where $T_v[u]$ is the entry corresponds to $u$ in $T_v$.
\If{($\phi_{t-1}(v)\in I_2$ \textbf{and} $a_{t-1}(v)<\lambda$)}
	\If{(($\exists u\in N(v), \phi_{t-1}(u)\in I_2,a_{t-1}(v)=a_{t-1}(u),b_{t-1}(u)=b_{t-1}(v)$) \Comment{Error-checking.}
	\Statex \hspace{\algorithmicindent}\quad \textbf{or} ($\exists u\in N(v), \phi_{t-1}(u)\in I_2, a_{t-1}(v)=a_{t-1}(u), T_u[v]+T_v[u]=0$)
	\Statex \hspace{\algorithmicindent}\quad \textbf{or} ($|\{u\mid u\in N(v), T_v[u]=1\}|>2\cdot \Delta^{1/4}$))}
		\State $\phi_t(v)\gets \ell_3+\ell_2+\sum_1^{r^*} n_i+id(v)$.
		\Else \Comment{Run the transition-out stage.}
			\If{($\forall u\in N(v)$, either ($\phi_{t-1}(u)\in I_2$ and $a_{t-1}(v) \leq a_{t-1}(u)<\lambda$) or ($\phi_{t-1}(u)\in I_3$))}\label{alg-line:self-stabilizing-phase2-stage3-if-cond-2}\label{alg-line:self-stabilizing-phase2-stage3-start}
				\State $L_{t-1}(v)\gets\{\phi_{t-1}(u)\mid u\in N(v),\phi_{t-1}(u)\in I_3\}$.
				\State $A_{t-1}(v)\gets\{u\mid u\in N(v),\phi_{t-1}(u)\in I_2, a_{t-1}(u)=a_{t-1}(v), T_v[u]=1\}$.
				\For{(every $i\in[\mu]$)}
					\State $L_{(t-1,i)}(v)\gets\{ (\lfloor{b_{t-1}(v)/\tau}\rfloor\cdot{x^2}+(b_{t-1}(v)\bmod\tau)\cdot{x}+i)\bmod\mu + x\cdot\mu \mid x\in[\mu] \}$.
				\EndFor
				\If{($d_{t-1}(v)= \mu$)}\label{alg-line:self-stabilizing-phase2-stage3-if-cond-3}
					\State $d_t(v)\gets$ the integer $\hat{i}$ in $[\mu]$ that minimizes $|L_{(t-1,\hat{i})}(v)\cap L_{t-1}(v)|$.\label{alg-line:self-stabilizing-phase2-stage3-d-rule}
					\State $\phi_t(v) \gets \ell_3 + \langle a_{t}(v),b_{t}(v),c_{t}(v),d_{t}(v)\rangle$.
				\Else
					\State $L_{t-1}'(v)\gets\{\phi_{t-1}(u)\mid u\in N(v), \phi_{t-1}(u)\in I_3, T_v[u]=0\}$.
					\If{($|L_{t-1}'(v)\cap L_{(t-1,d_{t-1}(v))}|>\Delta/\mu$)}\label{alg-line:self-stabilizing-phase2-stage3-proper-d}\label{alg-line:self-stabilizing-phase2-stage3-if-cond-4}
						\State $d_t=\mu$.\label{alg-line:self-stabilizing-phase2-stage3-reset-d}
						\State $\phi_t(v) \gets \ell_3 + \langle a_{t}(v),b_{t}(v),c_{t}(v),d_{t}(v)\rangle$.
					\ElsIf{(every $u\in A_{t-1}(v)$ has $d_{t-1}(u)\neq\mu$)}\label{alg-line:self-stabilizing-phase2-stage3-if-cond-5}
						\State $\phi_{t}(v)\gets \min L_{(t-1,d_{t-1}(v))}(v) \setminus \left( L_{t-1}(v) \bigcup \left(\cup_{u\in A_{t-1}(v)} L_{(t-1,d_{t-1}(u))}(u)\right) \right)$.
					\EndIf
				\EndIf
			\EndIf
	   \EndIf \label{alg-line:self-stabilizing-phase2-stage3-end}
\EndIf
\end{algorithmic}
\end{algorithm}

At the beginning of a round $t$, if vertex $v$ has color $\phi_{t-1}(v)\in I_2$ and $a(v)\in[\lambda]$, then it is in the transition-out stage. Once again, it does error-checking before proceeding. (See \Cref{alg:self-stabilizing-phase2-stage3} for the pseudocode.) Specficially, vertex $v$ checks if any of the following conditions is satisfied:
\begin{itemize}
	\item There exists a neighbor $u$ of $v$ such that $a(v)=a(u)$ and $b(v)=b(u)$, effectively implying $u$ and $v$ have identical color.
	\item There exists a neighbor $u$ of $v$ such that $a(v)=a(u)$ yet $T_v[u]+T_u[v]=0$, implying that the orientation of edge $(u,v)$ is still undetermined when $a(v)=a(u)$.
	\item The number of neighbors $u\in N(v)$ with $T_v[u]=1$ is larger than $2\cdot\Delta^{1/4}$, violating the bounded arboricity assumption during the transition-out stage.
\end{itemize}
If any of these conditions is satisfied, then vertex $v$ treats itself in an improper state and resets its color to $\ell_3+\ell_2+\sum_{i=1}^{r^*}n_i+id(v)$. Otherwise, it executes Line \ref{alg-line:self-stabilizing-phase2-stage3-start} to Line \ref{alg-line:self-stabilizing-phase2-stage3-end} of \Cref{alg:self-stabilizing-phase2-stage3} to transform its color from $I_2$ to $I_3$. The transformation is similar to the transition-out stage of the locally-iterative algorithm, except that we redefine $A_t(v)\triangleq \{u\mid u\in N(v),\phi_{t-1}(u)\in I_2, a_{t-1}(u)=a_{t-1}(v), T_v[u]=1\}$, replacing the constraint on $c$ values with a constraint on $T_v$.

More specifically, in a round $t$ in the transition-out stage of the self-stabilizing algorithm, for a vertex $v$ in proper state with $\phi_{t-1}(v)\in I_2$ and $a_{t-1}(v)\in [\lambda]$, if every $u\in N(v)$ satisfies either ``$\phi_{t-1}(u)\in I_2$ and $a_{t-1}(v)\leq a_{t-1}(u)<\lambda$'', or ``$\phi_{t-1}(u)\in I_3$'', then it is ready to transform from interval $I_2$ to $I_3$.
In such scenario, if $d_{t-1}(v)=\mu$, then it updates $d(v)$ in the same manner as in the locally-iterative algorithm.
Otherwise, if $d_{t-1}(v)\neq\mu$, then vertex $v$ makes sure $d_{t-1}(v)$ is proper for further operations by examining whether $|L_{t-1}'(v)\cap L_{(t-1,d_{t-1}(v))}|\leq\Delta/\mu$ is satisfied, where $L'_{t-1}(v)\triangleq \{\phi_{t-1}(u)\mid u\in N(v), \phi_{t-1}(v)\in I_3, T_v[u]=0\}$.
If $|L_{t-1}'(v)\cap L_{(t-1,d_{t-1}(v))}|\leq\Delta/\mu$, then vertex $v$ finds a color in $I_3$ by first finding the smallest integer $\hat{k}\in[\mu]$ satisfying: $$P_{(t-1,v,d_{t-1}(v))}(\hat{k}) + \mu\cdot{\hat{k}}~~\in~~L_{(t-1,d_{t-1}(v))}(v) \setminus \left( L_{t-1}(v) \bigcup \left(\cup_{u\in A_{t-1}(v)} L_{(t-1,d_{t-1}(u))}(u)\right) \right),$$
and then sets $\phi_t(v)=P_{(t-1,v,d_{t-1}(v))}(\hat{k}) + \mu\cdot{\hat{k}}$.
Otherwise, it resets its $d$ value to $\mu$, so that later it can obtain a proper $d(v)$.

Before proceeding to the next part, we note that once there are no errors occurring in the system: (1) the above mentioned $\hat{k}$ must exist, which is proved in the following claim; and (2) the above mentioned mechanism of resetting of $d$ occurs at most once for each vertex, which will be shown in the proof of \Cref{lemma:self-stabilizing-time-complexity-I2-part2}.

\begin{claim}\label{claim:self-stabilizing-phase2-stage3-bounded-k}
Consider a round $t$, consider a vertex $v$ that passes the error-checking procedure at the beginning of round $t$, and satisfies ``$\phi_{t-1}(v)\in I_2$, $a_{t-1}(v)\in [\lambda]$ and $d_{t-1}(v)\neq \mu$''. If no error occurs in round $t$ and $|L_{(t-1,d_{t-1}(v))}(v)\cap L'_{t-1}(v)|\leq \Delta/\mu$, then let $\hat{k}\in[\mu]$ be the smallest integer satisfying
$$P_{(t-1,v,d_{t-1}(v))}(k_t) + \mu\cdot{k_t} \in L_{(t-1,d_{t-1}(v))}(v) \setminus \left( L_{t-1}(v) \bigcup \left(\cup_{u\in \{w\mid w\in  A_{t-1}(v), d_{t-1}(w)\neq \mu \}} L_{(t-1,d_{t-1}(u))}(u)\right) \right),$$
it holds that
$$0\leq\hat{k}\leq \Delta/\mu +4\cdot\Delta^{1/4}.$$
\end{claim}

\begin{proof}
For simplicity, let $\mathfrak{A}$ denote set $L_{(t-1,d_{t-1}(v))}(v)$, let $\mathfrak{B}_1$ denote set $L_{t-1}(v)$ and let $\mathfrak{B}_2$ denote set $\bigcup_{u\in \{w\mid w\in  A_{t-1}(v), d_{t-1}(w)\neq \mu \}} L_{(t-1,d_{t-1}(u))}(u)$. Similar to the proof of \Cref{claim:phase2-stage3-bounded-k}, we bound $\hat{k}$ by bounding $|\mathfrak{A}\cap(\mathfrak{B}_1\cup \mathfrak{B}_2)|$.

Let $\mathfrak{B}_{1,0}$ denote $L_{t-1}'(v)$ and let $\mathfrak{B_{1,1}}$ denote $\mathfrak{B}_1\setminus\mathfrak{B}_{1,0}$. That is, $\mathfrak{B}_{1,1}\triangleq\{\phi_{t-1}(u)\mid u\in N(v), \phi_{t-1}(u)\in I_3, T_v[u]=1\}$. Since vertex $v$ passes the error-checking at the beginning of round $t$, we have $|\mathfrak{B}_{1,1}|+|A_{t-1}(v)|\leq 2\cdot \Delta^{1/4}$. Moreover, for every neighbor $u$ of $v$ with $a_{t-1}(u)=a_{t-1}(v)$, we have $b_{t-1}(u)\neq b_{t-1}(v)$, which leads to $|\mathfrak{A}\cap \mathfrak{B}_2|\leq 2\cdot |A_{t-1}(v)|$. Hence, we have:
\begin{align*}
|\mathfrak{A}\cap(\mathfrak{B}_1\cup \mathfrak{B}_2)| &= |\mathfrak{A}\cap(\mathfrak{B}_{1,0}\cup \mathfrak{B}_{1,1} \cup \mathfrak{B}_2)| \\
&\leq |\mathfrak{A}\cap\mathfrak{B}_{1,0}| +|\mathfrak{A}\cap\mathfrak{B}_{1,1}| + |\mathfrak{A}\cap\mathfrak{B}_{2}| \\
&\leq \Delta/\mu+ |\mathfrak{B}_{1,1}| + 2\cdot |A_{t-1}(v)| \\
&\leq \Delta/\mu+4\cdot \Delta^{1/4},
\end{align*}
which implies $0\leq \hat{k}\leq \Delta/\mu+4\cdot \Delta^{1/4}$.
\end{proof}

\subsection{The standard reduction phase}

For a vertex $v$ with its color in $I_3$, it considers itself in the standard reduction phase, whose error-checking procedure is very simple: if the color of itself collides with any neighbor, then it resets $\phi(v)$ to $\ell_3+\ell_2+\sum_1^{r^*} n_i+id(v)$. Otherwise, vertex $v$ considers itself in a proper state, and runs the standard reduction procedure described in the locally-iterative settings: if all neighbors of $v$ have colors in $I_3$, and if $v$ has the maximum color value in its inclusive one-hop neighborhood, then $v$ sets its color to be the minimum value in $[\Delta+1]$ that has not been used by any of its neighbors yet. The pseudocode of the standard reduction phase in the self-stabilizing setting is given in \Cref{alg:self-stabilizing-phase3}.

\begin{algorithm}[t!]
\caption{The standard reduction phase at vertex $v$ in round $t$}\label{alg:self-stabilizing-phase3}
\begin{algorithmic}[1]
\State Send $\langle \phi_{t-1}(v), T_v[u] \rangle$ to neighbor $u\in N(v)$, where $T_v[u]$ is the entry corresponds to $u$ in $T_v$.
\If{($\phi_{t-1}(v)\in I_3$)}
	\If{($\exists u\in N(v)$ with $\phi_{t-1}(u)=\phi_{t-1}(v)$)} \Comment{Error-checking.}
		\State $\phi_t(v)\gets \ell_3+\ell_2+\sum_1^{r^*} n_i+id(v)$.
	\Else \Comment{Run the standard reduction phase.}
		\If {(for every $u\in N(v)$ it holds that $\phi_{t-1}(u)\in I_3$)}
			\If {(for every $u\in N(v)$ it holds that $\phi_{t-1}(u)< \phi_{t-1}(v)$)}
				\State $\phi_t(v)\gets\min([\Delta+1]\setminus \{\phi_{t-1}(u)\mid u\in N(v)\})$.
			\EndIf
		\EndIf
	\EndIf
\EndIf
\end{algorithmic}
\end{algorithm}

\subsection{Algorithm analysis}

We now argue the correctness and the stabilization time of our algorithm.

\subsubsection{Correctness}

Recall that if $T_0$ is the last round in which the adversary corrupts vertices' states, our algorithm guarantees, the error-checking procedure will detect any anomalies at the beginning of round $T_0+1$, and resets the colors of those vertices. Moreover, staring from round $T_0+2$, the error-checking procedure will always pass, allowing the algorithm to make progress without disruption.

This property is summarized in \Cref{lemma:self-stabilizing-correctness}. To prove the it, we divide vertices into different categories according their their color values at the end of round $T_0+1$. We first consider vertices with color values in $I_1$ by the end of round $T_0+1$.

\begin{lemma}\label{lemma:self-stabilizing-correctness-I1}
If $T_0$ is the last round in which the adversary makes any changes to the RAM areas of vertices, then for every round $t\geq T_0+1$, for every vertex $v$ with $\phi_{t}(v)\in I_1$, the error-checking procedure will not reset vertex $v$'s color in the next round.
\end{lemma}

\begin{proof}
According to algorithm description, $v$ has $\phi_t(v) \in I_1$ if and only if it is in some improper state at the beginning of round $t$ or it is in some proper state and $\phi_{t-1}\in I_1^{(t')}$ for some $t'\in[r^*]$.
\begin{itemize}
	\item If a vertex $v\in V$ find itself in some improper state at the beginning of round $t$, then it resets $\phi_{t}(v)=\ell_3+\ell_2+\sum_1^{r^*} n_i+id(v)\in I_1$. For every neighbor $u\in N(v)$, in round $t$, either $u$ resets $\phi_{t}(u)=\ell_3+\ell_2+\sum_1^{r^*} n_i+id(u)$, or $u$ obtains a color $\phi_t(u)<\ell_3+\ell_2+\sum_1^{r^*} n_i$. In both cases, $\phi_t(v)\neq\phi_t(u)$. Hence, by \Cref{alg:self-stabilized-phase1-to-phase2-stage1}, the error-checking procedure will not reset $v$'s color in the next round.

	\item If vertex $v$ finds itself in some proper state and $\phi_{t-1}\in I_1^{(t')}$ for some $t'\in[r^*]$, then it computes its new color $\phi_t(v)\in I_1^{(t'+1)}$ based on set family $\mathcal{F}_{t'}$. For every neighbor $u\in N(v)$, in round $t$, if the error-checking fails at $u$, then $u$ resets $\phi_{t}(u)=\ell_3+\ell_2+\sum_1^{r^*} n_i+id(u)$, implying $\phi_t(v)\neq\phi_t(u)$. Otherwise, if the error-checking passes at $u$ and $\phi_{t-1}(u)\notin I_1^{(t')}$, by algorithm description we know $\phi_{t}(u)\notin I_1^{(t'+1)}$, implying $\phi_t(v)\neq\phi_t(u)$. Lastly, if the error-checking passes at $u$ and $\phi_{t-1}(u)\in I_1^{(t')}$, then due to the $\Delta$-cover-freeness of $\mathcal{F}_{t'}$, we know $\phi_t(v)\neq\phi_t(u)$. Hence, by \Cref{alg:self-stabilized-phase1-to-phase2-stage1}, the error-checking procedure will not reset $v$'s color in the next round.
\end{itemize}
This completes the proof of the lemma.
\end{proof}

Next, we consider vertices with color values in $I_2$ by the end of round $T_0+1$, and we further divide vertices in this category into two sub-categories: ones with $a_t(v)\geq\lambda$, and ones with $a_t(v)<\lambda$.

\begin{lemma}\label{lemma:self-stabilizing-correctness-I2-large-a}
If $T_0$ is the last round in which the adversary makes any changes to the RAM areas of vertices, then for every round $t\geq T_0+1$, for every vertex $v$ with $\phi_{t}(v)\in I_2$ and $a_{t}(v)\geq\lambda$, the error-checking procedure will not reset vertex $v$'s color in the next round.
\end{lemma}

\begin{proof}
According to algorithm description, $v$ has $\phi_t(v)\in I_2$ and $a_t(v)\geq \lambda$ iff it is in some proper state at the beginning of round $t$, and either ``$\phi_{t-1}(v)\in I_1^{(r^*)}$'', or ``$\phi_{t-1}(v)\in I_2$ and $a_{t-1}(v)\geq \lambda$''.

\smallskip\textsc{Scenario I}: vertex $v$ is in some proper state at the beginning of round $t$ and has $\phi_{t-1}(v)\in I_1^{(r^*)}$. In such case, $v$ transforms its color from interval $I_1$ to $I_2$ in round $t$. In particular, vertex $v$ first selects the smallest element $\hat{x}\in S_a^{(\phi_{t-1}(v))}$ satisfying $|N_1'(v,\hat{x})\cup N_2'(v,\hat{x})|\leq \Delta^{1/4}$, and sets $a_t(v)=\hat{x}+\lambda\in[\lambda^2]$.

Notice, such $\hat{x}$ must exist. To see this, recall the definition of $N_1'(v,x)$ and $N_2'(v,x)$:
\begin{align*}
N_1'(v,x) &\triangleq \{u\mid u\in N(v), \phi_{t-1}(u)\in I_1^{(t')}, x\in S_a^{(\phi_{t-1}(u))}\},\\
N_2'(v,x) &\triangleq \{ u\mid u\in N(v), \phi_{t-1}(u) \in I_2, x+\lambda=\hat{a}_{t-1}(u)\cdot \lambda + (\hat{a}_{t-1}(u)+ \tilde{a}_{t-1}(u))\bmod\lambda \}.
\end{align*}
Since $\phi_{t-1}(v)\in I_1^{(r^*)}$, we have $t'=r^*$. We say a neighbor $u\in N(v)$ creates a ``collision'' for some element $x\in S_a^{(\phi_{t-1}(v))}$ if ``$\phi_{t-1}(u)\in I_1^{(t')}\text{ and }x\in S_a^{(\phi_{t-1}(u))}$'' or ``$\phi_{t-1}(u) \in I_2\text{ and }x+\lambda=\hat{a}_{t-1}(u)\cdot \lambda + (\hat{a}_{t-1}(u)+ \tilde{a}_{t-1}(u))\bmod \lambda$''. Call the former as type one collision, and the latter as type two collision. If $\hat{x}$ cannot be found, then for each $x\in S_a^{(\phi_{t-1}(v))}$, the number of collisions created by by all neighbors for $x$ must reach $\Delta^{1/4}+1$; furthermore, the total number of collisions created by all neighbors for all elements in $S_a^{(\phi_{t-1}(v))}$ must reach $|S_a^{(\phi_{t-1}(v))}|\cdot(\Delta^{1/4}+1)>(\Delta+1)\log(n_{r^*})$. Now, for every $u\in N(v)$ with $\phi_{t-1}(u)\in I_1^{(t')}$, vertex $u$ can create at most $\log(n_{r^*})$ (type one) collisions for all elements in $S_a^{(\phi_{t-1}(v))}$, as $|S_a^{(\phi_{t-1}(v))}\cap S_a^{(\phi_{t-1}(u))}|\leq\log(n_{r^*})$. Moreover, for every $u\in N(v)$ with $\phi_{t-1}(u)\in I_2$, it can create at most one (type two) collision. Thus, the total number of collisions that can be created by the $\Delta$ neighbors of $v$ is bounded by $\Delta\log(n_{r^*})$. Therefore, $\hat{x}$ must exist.

By algorithm description, it is easy to verify that the neighbors in $N_1'(v,\hat{x})\cup N_2'(v,\hat{x})$ are all the neighbors that might have colliding $a$ value with $v$ by the end of round $t$.

Vertex $v$ then selects $b_t(v)$ which will not conflict with any neighbor in $N_1'(v,\hat{x})\cup N_2'(v,\hat{x})$, and sets $T_v[u]=1$ if and only if $u\in N_1'(v,\hat{x})\cup N_2'(v,\hat{x})$. Since $|N_1'(v,\hat{x})\cup N_2'(v,\hat{x})|\leq\Delta^{1/4}$ and $\mathcal{F}_b$ is a $\Delta$-union-$(\Delta^{1/4}+1)$-cover-free, $b_t(v)$ must exist. Moreover, $b_t(v)\in [m_2]$ by the definition of $\mathcal{F}_b$.

At this point, we can conclude: (1) every neighbor $u\in N(v)$ has either $a_t(u)\neq a_t(v)$ or $b_t(u)\neq b_t(v)$; (2) every neighbor $u\in N(v)$ that may have $a_t(u)=a_t(v)$ satisfies $T_v[u]=1$, which leads to $T_v[u]+T_u[v]\neq 0$; (3) the number of neighbors $u$ with $T_v[u]=1$ is bounded by $|N_1'(v,\hat{x})\cup N_2'(v,\hat{x})|<\Delta^{1/4}$; and (4) $b_t(v)\in [m_2]$.

\smallskip\textsc{Scenario II}: vertex $v$ is in some proper state at the beginning of round $t$ and satisfies: $\phi_{t-1}(v)\in I_2$ and $a_{t-1}(v)\geq \lambda$. Then, by algorithm description, for every $u\in N(v)$ with $\phi_t(u)\in I_2$ and $a_t(u)=a_t(v)\geq \lambda$, it must be in some proper state at the beginning of round $t$. Moreover, for each such $u$, either ``$\phi_{t-1}(u)\in I_1^{(r^*)}$'' or ``$\phi_{t-1}(u)\in I_2$ and $a_{t-1}(u)\geq \lambda$''.
\begin{itemize}
	\item If it is the case ``$\phi_{t-1}(u)\in I_1^{(r^*)}$'', then by the same argument as in \textsc{Scenario I} (but from the perspective of $u$), vertex $u$ must select a $b_t(u)\in [m_2]$ not equal to $b_t(v)$, and sets $T_u[v]=1$, which leads to $T_v[u]+T_u[v]\neq 0$. Moreover, since $v$ is in some proper state at the beginning of round $t$, the error-checking procedure in \Cref{alg:self-stabilizing-phase2-stage2} passes, which implies: (1) $b_t(v)=b_{t-1}(v)\in[m_2]$; and (2) $|\{w\mid w\in N(v), T_v[w]=1\}|\leq \Delta^{1/4}$ at the beginning of round $t$. Since $T_v$ stays unchanged in round $t$, we know $|\{w\mid w\in N(v), T_v[w]=1\}|\leq \Delta^{1/4}$ still holds at the end of round $t$.
	\item If it is the case ``$\phi_{t-1}(u)\in I_2$ and $a_{t-1}(u)\geq \lambda$'', by \Cref{claim:phase2-stage2_last_round_equal_a} and the assumption that $u,v$ are both in proper states at the beginning of round $t$, we have $a_{t-1}(u)=a_{t-1}(v)\geq\lambda$. Since $u,v$ are both in proper states at the beginning of round $t$, the error-checking procedure in \Cref{alg:self-stabilizing-phase2-stage2} passes, which further implies: (1) $b_t(u)=b_{t-1}(u)\neq b_t(v)=b_{t-1}(v)$, as well as $b_t(u)\in[m_2]$ and $b_t(v)\in[m_2]$; (2) $T_v[u]+T_u[v]\neq 0$ at the beginning of round $t$; and (3) $|\{w\mid w\in N(v), T_v[w]=1\}|\leq \Delta^{1/4}$ at the beginning of round $t$. Since vectors $T_v$ and $T_u$ stay unchanged in round $t$, we know $T_v[u]+T_u[v]\neq 0$ and $|\{w\mid w\in N(v), T_v[w]=1\}|\leq \Delta^{1/4}$ are both true at the end of round $t$.
\end{itemize}

\smallskip Finally, notice that according to the analysis for the two scenarios, every vertex with $\phi_{t}(v)\in I_2$ and $a_t(v)\geq\lambda$ has $b_t(v)\in [m_2]$. Hence, for every vertex $v$, every $u\in \{v\}\cup N(v)$ with $\phi_t(u)\in I_2$ and $a_t(u)\geq \lambda$ has $b_t(u)\in [m_2]$. By now, we conclude that, at the beginning of round $t+1$, the error-checking procedure in \Cref{alg:self-stabilizing-phase2-stage2} will not reset vertex $v$'s color.
\end{proof}

\begin{lemma}\label{lemma:self-stabilizing-correctness-I2-small-a}
If $T_0$ is the last round in which the adversary makes any changes to the RAM areas of vertices, then for every round $t\geq T_0+1$, for every vertex $v$ with $\phi_{t}(v)\in I_2$ and $a_{t}(v)<\lambda$, the error-checking procedure will not reset vertex $v$'s color in the next round.
\end{lemma}

\begin{proof}
According to the algorithm description, a vertex $v$ has $\phi_t(v)\in I_2$ and $a_t(v)<\lambda$ if and only if it is in some proper state at the beginning of round $t$ and has either ``$\phi_{t-1}(v)\in I_2$ and $a_{t-1}(v)\geq \lambda$'' or ``$\phi_{t-1}(v)\in I_2$ and $a_{t-1}(v)< \lambda$''.

\smallskip\textsc{Scenario I}: vertex $v$ is in some proper state at the beginning of round $t$ and has $\phi_{t-1}(v)\in I_2$ and $a_{t-1}(v)\geq \lambda$. In this case, $v$ runs \Cref{alg:self-stabilizing-phase2-stage2} and reduces its $a$ value from $[\lambda,\lambda^2)$ to $[\lambda]$ in round $t$. Define $\overline{M}_{t,0}\triangleq \{ u\mid u\in N(v),\phi_{t-1}(u)\in I_2,\hat{a}_{t-1}(u)=\hat{a}_{t-1}(v),\tilde{a}_{t-1}(u)=\tilde{a}_{t-1}(v), T_v[u]=0\}$.

By definition, $M_t(v)\cup \overline{M}_t(v)\cup\overline{M}_{t,0}(v)$ contains all neighbors that may have colliding $a$ value with $u$ by the of round $t$. For neighbors in $M_t(v)\cup \overline{M}_t(v)$, vertex $v$ selects a $b$ value that will not be used by any of them. Vertex $v$ also sets $T_v[u]=1$ for every $u\in M_t(v)\cup \overline{M}_t(v)$ by algorithm description, hence $T_v[u]+T_u[v]\neq 0$ by the end of round $t$. For every neighbor $u$ in $\overline{M}_{t,0}(v)$, since vertex $v$ is in some proper state at the beginning of round $t$, we have $T_u[v]=1$ at the beginning of round $t$. Thus, if indeed $u$ reduces its $a$ value to $[\lambda]$ in round $t$ which leads to $a_t(u)=a_t(v)$, we have $v\in \overline{M}_t(u)$ and vertex $u$ will select a $b_t(u)$ not equal to $b_t(v)$ and set $T_u(v)=1$. By now, we know that every neighbor $u$ of $v$ with $\phi_t(u)\in I_2$ has either $a_t(u)\neq a_t(v)$ or $b_t(u)\neq b_t(v)$. Moreover, every neighbor $u$ with $\phi_t(u)\in I_2$ and $a_t(u)=a_t(v)$ has $T_v[u]+T_u[v]=1$.

Since vertex $v$ reduces its $a$ value in round $t$, by the description of the algorithm, we have $|M_t(v)|\leq \Delta^{1/4}$. Since vertex $v$ is in proper state in round $t-1$ (otherwise it cannot be the case that $\phi_{t-1}(v)\in I_2$), we have the number of neighbors $u$ with $T_v[u]=1$ is bounded by $\Delta^{1/4}$ at the end of round $t-1$; that is, $|\overline{M}_t(v)|\leq \Delta^{1/4}$. Therefore, by the end of round $t$, the number of neighbors of $v$ with $T_v[u]=1$ is bounded by $|M_t(v)\cup \overline{M}_t(v)|\leq |M_t(v)| + |\overline{M}_t(v)| \leq 2\cdot \Delta^{1/4}$.

At this point, we conclude that, in \textsc{Scenario I}, at the beginning of round $t+1$, the error-checking procedure in \Cref{alg:self-stabilizing-phase2-stage3} will not reset vertex $v$'s color.

\smallskip\textsc{Scenario II}: vertex $v$ is in some proper state at the beginning of round $t$ and has $\phi_{t-1}(v)\in I_2$ and $a_{t-1}(v)< \lambda$. In this case, it maintains values $a$ and $b$, and vector $T_v$ unchanged in round $t$. For any neighbor $u$ with $a_t(v)=a_t(u)$, vertex $u$ must in some proper state at the beginning of round $t$ and either satisfies $a_{t-1}(u)=a_{t-1}(v)<\lambda$ or reduces its $a$ value from $[\lambda,\lambda^2)$ to $[\lambda]$ in round $t$.
\begin{itemize}
	\item For a neighbor $u\in N(v)$ that is in some proper state at the beginning of round $t$ and satisfies $a_{t-1}(u)=a_{t-1}(v)<\lambda$: we have  $a_{t-1}(v)=a_t(v)=a_t(u)=a_{t-1}(u)$,  $b_{t-1}(u)=b_t(u)$, and $T_u$ stays unchanged in round $t$. Since vertex $v$ is in proper state in round $t-1$ (otherwise it cannot be the case that $\phi_{t-1}(v)\in I_2$), we have $b_{t-1}(u)\neq b_{t-1}(v)$ and $T_v[u]+T_u[v]\neq 0$ at the beginning of round $t$. Thus, we have $b_t(u)=b_{t-1}(u)\neq b_{t-1}(u)=b_{t}(v)$, and $T_u[v]+T_v[u]\neq 0$ still holds at the end of round $t$.
	
	\item For a neighbor $u\in N(v)$ that is in some proper state at the beginning of round $t$ and reduces its $a$ value from $[\lambda,\lambda^2)$ to $[\lambda]$ in round $t$: we have $v\in M_t(u)$. By an analysis similar to \textsc{Scenario I} but from the perspective of $u$, vertex $u$ will select a $b$ value different from $b_t(v)$ and set $T_u[v]=1$.
\end{itemize}
Lastly, since $v$ is in proper state in round $t$, and since vector $T_v$ stays unchanged in round $t$, we know the number of neighbors with  $T_v[u]=1$ at the end of round $t$ is still bounded by $2\cdot \Delta^{1/4}$.

At this point, we conclude that, in \textsc{Scenario II}, at the beginning of round $t+1$, the error-checking procedure in \Cref{alg:self-stabilizing-phase2-stage3} will not reset vertex $v$'s color.
\end{proof}

We continue to consider vertices with color values in $I_3$ by the end of round $T_0+1$.

\begin{lemma}\label{lemma:self-stabilizing-correctness-I3}
If $T_0$ is the last round in which the adversary makes any changes to the RAM areas of vertices, then for every round $t\geq T_0+1$, for every vertex $v$ with $\phi_{t}(v)\in I_3$, the error-checking procedure will not reset vertex $v$'s color in the next round.
\end{lemma}

\begin{proof}
According to algorithm description, a vertex $v$ has $\phi_t(v)\in I_3$ if and only if it is in some proper state at the beginning of round $t$ and has either ``$\phi_{t-1}(v)\in I_2$ and transforms to $I_3$ in round $t$'' or ``$\phi_{t-1}(v)\in I_3$''.
\begin{itemize}
	\item \textsc{Scenario I}: vertex $v$ is in some proper state at the beginning of round $t$ with $\phi_{t-1}(v)\in I_2$ and transforms to $I_3$ in round $t$. In this case, any neighbor $u\in N(v)$ with $\phi_t(u)\in I_3$ must be in some proper state at the beginning of round $t$. Since faults no longer occur in round $t$, by an identical argument as in the proof of \Cref{lemma:phase2-stage3-proper-color}, it holds that $\phi_t(u)\neq \phi_t(v)$.

	\item \textsc{Scenario II}: vertex $v$ is in some proper state at the beginning of round $t$ with $\phi_{t-1}(v)\in I_3$. Any neighbor $u\in N(v)$ that may have $\phi_t(u)\in I_3$ must be in some proper state at the beginning of round $t$. Moreover, either ``$\phi_{t-1}(u)\in I_2$ and vertex $u$ transforms its color to $I_3$ in round $t$'' or ``$\phi_{t-1}(u)\in I_3$''. In the former case, by an analysis similar to \textsc{Scenario I} (but swapping the role of $u$ and $v$), it holds that $\phi_t(u)\neq \phi_t(v)$. In the latter case, we have $\phi_{t-1}(u)\neq\phi_{t-1}(v)$ since $u,v$ are both in proper states at the beginning of round $t$. Moreover, by \Cref{alg:self-stabilizing-phase3}, in round $t$, at most one of $u,v$ will change its color, and the updated color of that vertex will not conflict with the other vertex, thus $\phi_t(u)\neq \phi_t(v)$.
\end{itemize}
We conclude that at the beginning of round $t+1$, the error-checking procedure in \Cref{alg:self-stabilizing-phase3} will not reset vertex $v$'s color.
\end{proof}

The following lemma is the last missing piece before we can prove \Cref{lemma:self-stabilizing-correctness}.

\begin{lemma}\label{lemma:self-stabilizing-correctness-color-in-I123}
If $T_0$ is the last round in which the adversary makes any changes to the RAM areas of vertices, then for every round $t\geq T_0+1$, for every vertex $v$, it holds that $\phi_t(v)\in I_1\cup I_2\cup I_3$.
\end{lemma}

\begin{proof}
If vertex $v$ finds itself in some improper state at the beginning of round $t$, then it resets itself by setting $\phi_{t}(v)=\ell_3+\ell_2+\sum_1^{r^*} n_i+id(v)\in I_1$. Otherwise, we have vertex $v$ in proper state with its color in $I_1\cup I_2\cup I_3$. By the description of the algorithm, if no error occurs in round $t$, for vertex $v$ with  $\phi_{t-1}(v)\in I_1$, it has $\phi_{t}(v)\in I_1\cup I_2$; for vertex $v$ with  $\phi_{t-1}(v)\in I_2$, it has $\phi_{t}(v)\in I_2\cup I_3$; for vertex $v$ with  $\phi_{t-1}(v)\in I_3$, it has $\phi_{t}(v)\in I_3$. This completes the proof of the lemma.
\end{proof}

At this point, it is easy to see \Cref{lemma:self-stabilizing-correctness-I1}, \Cref{lemma:self-stabilizing-correctness-I2-large-a}, \Cref{lemma:self-stabilizing-correctness-I2-small-a}, \Cref{lemma:self-stabilizing-correctness-I3}, \Cref{lemma:self-stabilizing-correctness-color-in-I123} together immediately lead to \Cref{lemma:self-stabilizing-correctness}.

\subsubsection{Stabilization time}

To analyze the time cost of our self-stabilizing algorithm, which is summarized in \Cref{lemma:self-stabilizing-time-complexity-I3}, we take a similar approach as in the analysis of the locally-iterative algorithm. Specifically, we will show once the adversary stops disrupting algorithm execution, the maximum amount of time for vertices to progress through each phase/stage is limited, resulting in a bounded stabilization time.

We begin with the Linial phase and the transition-in stage by proving \Cref{lemma:self-stabilizing-time-complexity-I1}.

\begin{proof}[Proof of \Cref{lemma:self-stabilizing-time-complexity-I1}]
By the correctness guarantee provided by \Cref{lemma:self-stabilizing-correctness}, we have that for every round from $T_0+2$, at the beginning of that round, every vertex $v$ has its color in $I_1\cup I_2\cup I_3$ and is in a proper state. Hence, by algorithm description, in a round $t'\geq T_0+2$, every vertex $v$ with $\phi_{t'-1}(v)\in I_1^{(j)}$ computes its new color $\phi_{t'}(v)\in I_1^{(j+1)}$, where $j\in[r^*]$; every vertex $v$ with $\phi_{t'-1}(v)\in I_1^{(r^*)}$ computes its new color $\phi_t(v)\in I_2$; and every vertex $v$ with $\phi_{t'-1}(v)\in I_2\cup I_3$ computes its new color $\phi_t(v)\in I_2\cup I_3$. Now, by an induction on $k$ from $0$ to $r^*$ (both inclusive), it is easy to see, by the end of round $T_0+1+k$, for any vertex $v$, its color $\phi_{T_0+1+k}(v)$ is in:
$$\left(\cup_{i=k}^{r^*} I_1^{(i)}\right) \cup I_2\cup I_3.$$
Therefore, for every vertex $v$, it holds that $\phi_{T_0+1+r^*}(v)\in I_1^{(r^*)}\cup I_2\cup I_3$. After one more round, for every vertex $v$, it holds that $\phi_{T_0+2+r^*}(v)\in I_2\cup I_3$.
\end{proof}

Next, we consider the core stage and the transition-out stage, and prove \Cref{lemma:self-stabilizing-time-complexity-I2-part1} and \Cref{lemma:self-stabilizing-time-complexity-I2-part2}.

\begin{proof}[Proof of \Cref{lemma:self-stabilizing-time-complexity-I2-part1}]
By \Cref{lemma:self-stabilizing-time-complexity-I1}, every vertex $v$ has $\phi_{T_0+r^*+2}\in I_2\cup I_3$. If ``$\phi_{T_0+r^*+2}(v)\in I_2$ and $a_{T_0+r^*+2}(v)\in [\lambda]$'' or ``$\phi_{T_0+r^*+2}\in I_3$'', then we are already done. Otherwise, vertex $v$ has ``$\phi_{T_0+r^*+2}(v)\in I_2$ and $a_{T_0+r^*+2}(v)\in [\lambda,\lambda^2)$'' and runs \Cref{alg-line:self-stabilizing-phase2-stage2-start} to \Cref{alg-line:self-stabilizing-phase2-stage2-end} of \Cref{alg:self-stabilizing-phase2-stage2} from round $T_0+r^*+3$ to $t_v^*$. In such case, we use the same proof strategy as in the proof of \Cref{lemma:phase2-stage2-time-complexity} (see \cpageref{proof:lemma:phase2-stage2-time-complexity}). Specifically, the first claim and the second claim in that proof still hold with an offset $(T_0+1)$ on round number. Combining the two claims, we know starting from round $T_0+r^*+3$, within $\lambda$ rounds, there must exists a round $t$ in which, the reduction condition $|M_t(v)|\leq \Delta$ is satisfied. As a result, by the end of round $t\leq T_0+r^*+3+\lambda $, we have $a_t(v)\in [\lambda]$ and $t^*_v=t$.
\end{proof}

\begin{proof}[Proof of \Cref{lemma:self-stabilizing-time-complexity-I2-part2}]
By \Cref{lemma:self-stabilizing-time-complexity-I2-part1}, every vertex $v$ has either ``$\phi_{T_0+r^*+\lambda+3}(v)\in I_2$ and $a_{T_0+r^*+\lambda+3}(v) \in [\lambda]$'' or ``$\phi_{T_0+r^*+\lambda+3}(v)\in I_3$''. If $\phi_{T_0+r^*+\lambda+3}(v)\in I_3$, then $t^{\#}_v\leq T_0+r^*+4\lambda+3$ holds trivially and we are done. So, assume this is not the case.

Consider a vertex $v$ with $\phi_{T_0+r^*+\lambda+3}(v)\in I_2$ and $a_{T_0+r^*+\lambda+3}(v) \in [\lambda]$, let $t^-_v>T_0+r^*+\lambda+3$ be the smallest round number such that every $u\in N(v)$ with $\phi_{t^-_v-1}(u)\in I_2$ has $a_{{t^-_v}-1}(v)\leq a_{{t^-_v}-1}(u)<\lambda$ or $\phi_{{t^-_v}-1}(u)\in I_3$. (That is, the ``if'' condition in Line \ref{alg-line:self-stabilizing-phase2-stage3-if-cond-2} of \Cref{alg:self-stabilizing-phase2-stage3} is first satisfied in round $t^-_v$.) Further define $t^+_v\geq t^-_v$ to be the smallest round number such that $\phi_{t^+_v}(v)\in I_2$ and $d_{t^+_v}(v)\neq \mu$ or $\phi_{t^+_v}(v)\in I_3$. By definition and the algorithm description, we have $t^*_v\leq t^-_v\leq t^+_v\leq t^{\#}_v$. Moreover, if faults no longer occur, it is easy to verify that once the ``if'' condition in Line \ref{alg-line:self-stabilizing-phase2-stage3-if-cond-2} of \Cref{alg:self-stabilizing-phase2-stage3} is satisfied for vertex $v$ in round $t^-_v$, then it is satisfied for any round $t> t^-_v$ as long as $\phi_{t-1}(v)\in I_2$.

To prove the lemma, we prove a stronger claim: for any vertex $v\in V$ with $\phi_{T_0+r^*+\lambda+3}(v)\in I_2$ and $a_{T_0+r^*+\lambda+3}(v) \in [\lambda]$, it holds that $t^-_v\leq T_0+r^*+\lambda +1+ 3(a_{T_0+r^*+\lambda+3}(v)+1)$, $t^+_v\leq T_0+r^*+\lambda +2+ 3(a_{T_0+r^*+\lambda+3}(v)+1)$ and $t^{\#}_v\leq T_0+r^*+\lambda +3 +  3(a_{T_0+r^*+\lambda+3}(v)+1)$.

We prove the claim via an induction on the value of $a$ at the end of round $T_0+r^*+\lambda+3$, which is in $[\lambda]$. For the base case, fix a vertex $w$ with the minimum $a$ value at the end of round $T_0+r^*+\lambda+3$. By \Cref{lemma:self-stabilizing-time-complexity-I2-part1} and algorithm description, every vertex $v\in V$ has either  $\phi_{T_0+r^*+\lambda+3}(v)\in I_2$ and $a_{T_0+r^*+\lambda+3}(v) \in [\lambda]$ or $\phi_{T_0+r^*+\lambda+3}(v)\in I_3$. Recall $w$ has the minimum $a$ value at the end of round $T_0+r^*+\lambda+3$, we know $t^-_w=T_0+r^*+\lambda+4$. In round $t^-_w$, there are three potential cases:
\begin{itemize}
	\item Case 1: $d_{t^-_w-1}(w)=\mu$. Then vertex $w$ selects a value $d$ in round $t^-_w$, and we have $t^+_w=t^-_w$.
	\item Case 2: $d_{t^-_w-1}(w)\neq\mu$ and the ``if'' condition in Line \ref{alg-line:self-stabilizing-phase2-stage3-proper-d} of \Cref{alg:self-stabilizing-phase2-stage3} is satisfied. Then, vertex $w$ sets $d_{t^-_w}(t)=\mu$ and in round $t^-_w+1$ it will select a new $d$ value not equaling to $\mu$. Thus, we have $t^+_w=t^-_w+1$ in this case.
	\item Case 3: $d_{t^-_w-1}(w)\neq \mu$ and the ``if'' condition in Line \ref{alg-line:self-stabilizing-phase2-stage3-proper-d} of \Cref{alg:self-stabilizing-phase2-stage3} is not satisfied. Then, in round $t^-_w$, vertex $w$ either transforms its color to $I_3$ or stay in $I_2$. In both cases, we have $t^+_w=t^-_w$.
\end{itemize}

Before proceeding, we prove an auxiliary claim, which intuitively states that once there are no errors, resetting $d$ to $\mu$ (i.e., Line \ref{alg-line:self-stabilizing-phase2-stage3-reset-d} of \Cref{alg:self-stabilizing-phase2-stage3}) occurs at most once for each vertex.

\begin{claim*}
For any round $t>t^+_w$ with $\phi_{t-1}(w)\in I_2$, it holds that
$$|L_{(t-1,d_{t-1}(w))}(v)\cap L'_{t-1}(w)|\leq \Delta/\mu.$$
\end{claim*}

\begin{proof}
We prove by induction on $t$, and we begin with the base case $t=t^+_w+1$.
\begin{itemize}
	\item In case 1 and case 2, vertex $w$ selects a new $d$ value in round $t^+_w=t-1$. By Line \ref{alg-line:self-stabilizing-phase2-stage3-d-rule} in \Cref{alg:self-stabilizing-phase2-stage3} for setting $d_{t-1}(w)$ and the pigeonhole principle, we have $|L_{(t-1,d_{t-1}(w))}(w)\cap L_{t-2}(w)|\leq\Delta/\mu$. Since $L_{t-2}'(w)\subseteq L_{t-2}(w)$, we have $|L_{(t-1,d_{t-1}(w))}(w)\cap L'_{t-2}(w)|\leq\Delta/\mu$. Observe that, some neighbors of $w$ may map their colors from $I_2$ to $I_3$ in round $t-1$, we continued to prove that $L'_{t-2}(w)=L'_{t-1}(w)$. Consider such a neighbor $u$ of $w$, it must have $a_{t-1}(w)= a_{t-1}(u)$, as being able to map its color from $I_2$ to $I_3$ means the ``if'' condition in Line \ref{alg-line:self-stabilizing-phase2-stage3-if-cond-2} of \Cref{alg:self-stabilizing-phase2-stage3} is satisfied for $u$ in round $t-1$. Since vertex $w$ is in proper state, we have $T_w[u]+T_u[w]\neq 0$ at the beginning of round $t$. For the case $T_w[u]\neq 0$, although $u$ has a color in $I_3$, its new color will not be in $L_{t-1}'(w)$ as $T_w[u]\neq 0$. For the case $T_u[w]\neq 0$, since $d_{t-2}(w)=\mu$, the ``if'' condition at Line \ref{alg-line:self-stabilizing-phase2-stage3-if-cond-5} of \Cref{alg:self-stabilizing-phase2-stage3} will not be satisfied for $u$ in round $t-1$, meaning $u$ cannot map its color from $I_2$ to $I_3$ in round $t-1$. Thus we have $L'_{t-2}(w)=L'_{t-1}(w)$ and $|L_{(t-1,d_{t-1}(w))}(w)\cap L'_{t-1}(w)|= |L_{(t-1,d_{t-1}(w))}(w)\cap L'_{t-2}(w)| \leq \Delta/\mu$.

	\item In case 3, for vertex $w$, the ``if'' condition in Line \ref{alg-line:self-stabilizing-phase2-stage3-proper-d} of \Cref{alg:self-stabilizing-phase2-stage3} is not satisfied in round $t^+_w=t-1$, thus $|L_{(t-2,d_{t-2}(w))}(w)\cap L_{t-2}'(w)|\leq \Delta/\mu$. Since  $d_{t-1}(w)=d_{t-2}(w)$, we have $L_{(t-2,d_{t-2}(w))}(w) = L_{(t-1,d_{t-1}(w))}(w)$. Consider a neighbor $u$ of $w$ that maps its color from $I_2$ to $I_3$ in round $t-1$, it must be the case that $a_{t-1}(w)= a_{t-1}(u)$. We continued to prove that  $\phi_{t-1}(u)$  either is not in $L_{(t-1,d_{t-1}(w))}(w)$ or not in $L_{t-1}'(w)$. Since vertex $w$ is in proper state, we have $T_w[u]+T_u[w]\neq 0$ at the beginning of round $t$. For the case $T_w[u]\neq 0$, although $u$ has a color in $I_3$, its new color will not be in $L_{t-1}'(w)$ as it has $T_w[u]\neq 0$. For the case $T_u[w]\neq 0$, then the color in $I_3$ selected by $u$ in round $t-1$ is not in $L_{(t-2,d_{t-2}(v))}(v)=L_{(t-1,d_{t-1}(v))}(v)$. Thus we have $|L_{(t-1,d_{t-1}(v))}(v)\cap L'_{t-1}(v)|= |L_{(t-2,d_{t-2}(v))}(v)\cap L'_{t-2}(v)| \leq  \Delta/\mu$.
\end{itemize}
By now we have proved the base case. Notice that the inductive step can be proved by the same argument as in case 3, we conclude the claim is true.
\end{proof}

We resume the lemma proof. Due to the above claim, we know for any round $t>t^+_w$ with $\phi_{t-1}(w)\in I_2$, it holds that $d_{t-1}(w)\neq \mu$ and $|L_{(t-1,d_{t-1}(w))}(w)\cap L'_{t-2}(w)|\leq \Delta/\mu$, and its $d$ value will not change anymore. Now, recall vertex $w$ has the minimum $a$ value at the end of round $T_0+r^*+\lambda+3$, and that $t^-_w=T_0+r^*+\lambda+4$, $t^+_w\leq t^-_w+1= T_0+r^*+\lambda+5$. In round $T_0+r^*+\lambda+6$, if $\phi_{T_0+r^*+\lambda+5}\in I_2$, then the ``if" condition in Line \ref{alg-line:self-stabilizing-phase2-stage3-if-cond-2} and Line \ref{alg-line:self-stabilizing-phase2-stage3-if-cond-5} of \Cref{alg:self-stabilizing-phase2-stage3} will be satisfied, and the ``if" condition in Line \ref{alg-line:self-stabilizing-phase2-stage3-if-cond-3} and Line \ref{alg-line:self-stabilizing-phase2-stage3-if-cond-4} of \Cref{alg:self-stabilizing-phase2-stage3} will not not be satisfied. As a result, by \Cref{claim:self-stabilizing-phase2-stage3-bounded-k}, vertex $w$ will obtain a new color $\phi_{T_0+r^*+\lambda+6}\in I_3$. Hence, we have $t^{\#}_w=T_0+r^*+\lambda+6$. This completes the proof for the base case.

Assume our claim holds for every vertex $v$ with $a_{T_0+r^*+\lambda+3}(v)\leq i \in [\lambda-1]$. Consider a vertex $w$ with $a_{T_0+r^*+\lambda+3}(w)=i+1$. By the induction hypothesis, for every vertex $v$ with $a_{T_0+r^*+\lambda+3}(v)\leq i$, it holds that $t^{\#}_v\leq T_0+r^*+\lambda +3 +  3(a_{T_0+r^*+\lambda+3}(v)+1)\leq T_0+r^*+\lambda +3 + 3(i+1) = T_0+r^*+\lambda + 3((i+1)+1)$. Then, we have $t^-_w\leq T_0+r^*+\lambda + 3((i+1)+1) + 1$. Apply the same argument as in the base case, we have $t^+_w\leq T_0+r^*+\lambda  + 3((i+1)+1)  + 2$ and $t^{\#}_w\leq T_0+r^*+\lambda  + 3((i+1)+1) + 3$. This completes the proof for the inductive step.

We conclude that for every vertex $v$ with $\phi_{T_0+r^*+\lambda+3}(v)\in I_2$ and $a_{T_0+r^*+\lambda+3}(v) \in [\lambda]$, it holds that $t^{\#}_w\leq T_0+r^*+4\lambda +3$. This completes the proof of this lemma.
\end{proof}

We can now prove \Cref{lemma:self-stabilizing-time-complexity-I3}, which implies the stabilization time of our algorithm.

\begin{proof}[Proof of \Cref{lemma:self-stabilizing-time-complexity-I3}]

By \Cref{lemma:self-stabilizing-time-complexity-I2-part2}, by the end of round $T_0+r^*+4\lambda+3$, every vertex must have its color in $I_3$, and will run \Cref{alg:self-stabilizing-phase3} starting from round $T_0+r^*+4\lambda+4$. In each such round, by \Cref{lemma:self-stabilizing-correctness}, every vertex is in some proper state. Hence, by \Cref{alg:self-stabilizing-phase3}, if there still exists a vertex with color not in $[\Delta+1]$, then the maximum value of the color used by any vertex will be reduced by at least one. Recall that every vertex in interval $I_3$ has its color in $[\ell_3]$ with $\ell_3=\Delta+2(\sqrt{m_3}+1)\cdot\mu$. Therefore, by the end of round $T_0+r^*+4\lambda+3+(\ell_3-(\Delta+1))=T_0+r^*+4\lambda+2+2(\sqrt{m_3}+1)\mu$, every vertex has its color in $[\Delta+1]$. Moreover, in any later round, by \Cref{alg:self-stabilizing-phase3}, the color of any vertex will remain in $[\Delta+1]$.
\end{proof}

\end{document}